\documentclass[12pt]{article}
\parskip=\medskipamount

\DeclareFontEncoding{LS1}{}{}
\DeclareFontSubstitution{LS1}{stix}{m}{n}
\DeclareSymbolFont{stixsymbols}{LS1}{stixscr}{m}{n}
\SetSymbolFont{stixsymbols}{bold}{LS1}{stixscr}{b}{n}
\DeclareMathSymbol{\kay}{\mathalpha}{stixsymbols}{"6B}
\DeclareMathSymbol{\hay}{\mathalpha}{stixsymbols}{"68}

\usepackage{verbatim}
\usepackage{tensor}
\usepackage{amsmath,amssymb,calc, amsthm,bbm, epsfig,psfrag, mathtools}
\usepackage{latexsym,bm,amsfonts}
\usepackage{float, caption}
\usepackage{booktabs}
\usepackage{siunitx}

\newcommand{\bk}[1]{\langle #1 \rangle}

\newcommand{\bR}{\mathbb{R}}

\newtheorem{thm}{Theorem}[section]
\newtheorem{prop}{Proposition}[section]
\newtheorem{conjecture}{Conjecture}

\usepackage{graphicx, enumerate}
\usepackage{enumitem}

\newcommand{\ul}[1]{{\underline{#1}}}
 \usepackage[all,cmtip]{xy}
\usepackage[numbers,sort&compress]{natbib}
\usepackage[dvipsnames]{xcolor}
\usepackage{mathrsfs} 
\usepackage{graphicx,tikz,tikz-cd}
\usepackage{braket}
\usepackage{physics}

\allowdisplaybreaks

\newcommand{\Comment}[1]{{}}
\definecolor{darkblue}{rgb}{0.15,0.35,0.55}
\definecolor{reddish}{rgb}{0.65, 0.2, 0.2}
\usepackage[linktocpage=true]{hyperref}
\hypersetup{
colorlinks=true,
citecolor=darkblue,
linkcolor=reddish,
urlcolor=darkblue,
pdfauthor={},
pdftitle={},
pdfsubject={}
}

\makeatletter
\renewcommand\section{\@startsection {section}{1}{\z@}%
                                   {-3.5ex \@plus -1ex \@minus -.2ex}%nn
                                   {2.3ex \@plus.2ex}%
                                   {\normalfont\large\bfseries}}
\renewcommand\subsection{\@startsection{subsection}{2}{\z@}%
                                     {-3.25ex\@plus -1ex \@minus -.2ex}%
                                     {1.5ex \@plus .2ex}%
                                     {\normalfont\bfseries}}
\makeatother

\newfont{\goth}{ygoth.tfm scaled 1200}                   % gothic font (usual)
 \numberwithin{equation}{section}

\newcommand{\overbar}[1]{\mkern 1.5mu\overline{\mkern-1.5mu#1\mkern-1.5mu}\mkern 1.5mu}

\begin{document}
%%%%%%%%%%%%%%%%
%%%%%%%%%%%%%%%%
\begin{titlepage}
\begin{flushright}
\today
\end{flushright}
\vspace{2cm}

\begin{center}
{\Large \bf 
Quantum Mechanics and Neural Networks
}
\end{center}

\begin{center}

{\bf
Christian Ferko${}^{a, b}$ and James Halverson${}^{a, b}$
} \\
\vspace{5mm}

\footnotesize{
${}^{a}$
{\it 
Department of Physics, Northeastern University, Boston, MA 02115, USA
}
 \\~\\
${}^{b}$
{\it 
The NSF Institute for Artificial Intelligence
and Fundamental Interactions
}
}
\vspace{3mm}
~\\
\texttt{c.ferko@northeastern.edu,
j.halverson@northeastern.edu}\\
\vspace{3mm}

\end{center}

\begin{abstract}
\baselineskip=14pt

\noindent 

We demonstrate that any Euclidean-time quantum mechanical theory may be represented as a neural network, ensured by the Kosambi-Karhunen-Lo\`eve theorem, mean-square path continuity, and finite two-point functions. The additional constraint of reflection positivity, which is related to unitarity, may be achieved by a number of mechanisms, such as imposing neural network parameter space splitting or the Markov property. Non-differentiability of the networks is related to the appearance of non-trivial commutators. Neural networks acting on Markov processes are no longer Markov, but still reflection positive, which facilitates the definition of deep neural network quantum systems. We illustrate these principles in several examples using numerical implementations, recovering classic quantum mechanical results such as Heisenberg uncertainty, non-trivial commutators, and the spectrum.

\end{abstract}
\vspace{5mm}

\vfill
\end{titlepage}

\newpage
\renewcommand{\thefootnote}{\arabic{footnote}}
\setcounter{footnote}{0}

\tableofcontents{}
\vspace{1cm}
\bigskip\hrule

%%%%%%%%%%%%%%%%%%%%%%%%%%%%%%%%%%%%%%%%%%%%%%%%%%%%%%
%%%%%%%%%%%%%%%%%%%%%%%%%%%%%%%%%%%%%%%%%%%%%%%%%%%%%%
%%%%%%%%%%%%%%%%%%%%%%%%%%%%%%%%%%%%%%%%%%%%%%%%%%%%%%

\allowdisplaybreaks

\clearpage 
\section{Introduction}

An emerging connection between field theory and neural networks provides a physical approach to understanding machine learning (ML) and a new method for defining field theories. The essential idea is that a neural network (NN) is a family of functions $\phi_\theta$ accompanied by a density $P(\theta)$ on parameters,\footnote{$P(\theta)$ could be the density at initialization, or after flowing for some amount of learning time $t$.} which together define correlation functions
\begin{equation}
G^{(n)}(x_1,\dots,x_n) = \bk{\phi(x_1)\dots \phi(x_n)} = \int d\theta \, P(\theta)\,  \phi_\theta(x_1)\dots \phi_\theta(x_n) \, ,
\end{equation}
where the last equality provides a method of studying correlators using a NN description, independent of knowledge of an action. The functional form of $\phi_\theta$ is known as the \emph{architecture}, and henceforth we call the NN description of a field theory a NN-FT. Together, the data $(\phi_\theta, P(\theta))$ furnishes an ensemble of fields and associated correlation functions, the minimal defining data of a field theory. Beginning from this data, one might wish to engineer additional interesting features, such as Gaussianity \cite{Halverson:2020trp,Halverson:2021aot}, interactions \cite{Demirtas:2023fir}, symmetries \cite{Maiti:2021fpy}, e.g. conformal symmetry \cite{Halverson:2024axc}, the axioms of constructive quantum field theory (QFT), or other features relevant for studying physical systems.\footnote{Some of these ideas have parallels in the ML community, e.g. \cite{neal, williams, Matthews2018GaussianPB, GarrigaAlonso2019DeepCN, yangTPorig, yangTP1, yangTP2, Novak2018BayesianCN} for understanding neural networks at infinite width, \cite{Yaida2019NonGaussianPA, Halverson_2021, Roberts_2022, Naveh_2021, antognini2019finite, Demirtas:2023fir} for finite-width corrections, \cite{Dyer2020AsymptoticsOW, maiti2021symmetryviaduality, erdmenger2021quantifying, grosvenor2022edge, Erbin_2022, erbin2022renormalization, maloney2022solvablemodelneuralscaling,banta2023structures,10.1088/2632-2153/adc872} for more general field theory techniques applied to neural networks, and \cite{halverson2021building} as well as references above advocating applications in physics. For a lengthy summary of some of the references, see the Introduction of \cite{Demirtas:2023fir}; for lectures in the ``for-ML" direction, see \cite{halverson2024tasilecturesphysicsmachine}; and for a quick summary of essential NN-FT results, see J.H.'s talk at Strings 2025.}

When are these theories quantum? Since a NN is usually defined on Euclidean space, and therefore the NN-FT defines Euclidean correlators, the question amounts to asking when Euclidean correlators admit a nice-enough quantum Lorentzian continuation.\footnote{The realization of NN-FTs satisfying the OS axioms was proposed in \cite{Halverson:2021aot}, with Gaussian examples.} This question is central in constructive field theory, where the Osterwalder-Schrader (OS) axioms \cite{Osterwalder:1973dx,Osterwalder:1974tc} are necessary and sufficient conditions on Euclidean correlators to ensure that the Lorentzian continuation satisfies the Wightman axioms \cite{osti_4606723}.\footnote{It is also natural to ask how these relate to other constructions. Theories defined by certain local Lagrangians \cite{glimm2012quantum} or satisfying Euclidean CFT axioms \cite{Kravchuk:2021kwe} are known to satisfy the OS axioms.} An essential role is played by reflection positivity (RP), a Euclidean constraint that ensures unitarity in the quantum theory. A potential drawback of the OS axioms is that Lorentz invariance is baked into the axioms, restricting to a specific type of quantum field theory, rather than the general case. Nevertheless, RP may be studied outside of the OS context, providing a more general notion of when a Euclidean system is, in fact, quantum.

In this paper we restrict our attention to quantum \emph{mechanics}, considering fields $\phi_\theta$ on $\bR^d$ with $d=1$ and renaming them $x ( t ) = \phi_\theta ( t )$. Instead of asking when a Euclidean field theory is the continuation of a QFT, we ask when a stochastic process (SP) ($d=1$ Euclidean field theory) is the continuation of a quantum mechanics theory. Since we will aim for clarity in the main text, rather than brevity, we state our main results here:
\begin{enumerate}
\item \textbf{Universality of NN Approach to QM.} Mean-square path continuity and the K\"all\'en-Lehmann spectral representation of QM ensure that QM systems satisfy the assumptions of the Kosambi-Karhunen-Lo\`eve theorem of stochastic processes. This gives a statistical decomposition of the paths $x(t)$ that may be interpreted as a NN.

\item \textbf{Mechanisms for Unitarity and Reflection Positivity.} 

We develop a ``parameter splitting'' mechanism sufficient for RP, and provide a method by which any seed architecture may be turned into an architecture satisfying RP. However, the method breaks translation invariance, which is required by the OS axioms in QM, unless the splitting (which depends on $t$) is engineered for all $t$.

Alternatively, Markov processes are RP. We 
demonstrate that a neural network acting on any Markov process may not be Markov, but is still RP.

Both of these mechanisms have a common feature: the non-differentiability of paths, which is essential for obtaining non-trivial commutators and Heisenberg uncertainty.

\item \textbf{Deep NN-QM: A Prescription for Defining Quantum Systems.} Since any neural network acting on a Markov or RP process yields an RP process, deep NNs provide a means of defining a vast array of Euclidean quantum systems. We instantiate this idea numerically using the Ornstein-Uhlenbeck process (the Euclidean analog of quantum harmonic oscillator) and a variety of NNs acting on it, studying commutators, Heisenberg uncertainty, and the spectrum in each example.
\end{enumerate}

Our methodology shares some ingredients with other techniques for understanding quantum theories using neural networks or stochastic processes. For instance, see \cite{Hashimoto:2024aga} for another approach to representing quantum systems via neural networks, focusing on models that can be defined by a path integral; in this work, we do not assume a path integral representation for the models under consideration. Another strategy, called stochastic quantization \cite{PhysRev.150.1079,Fényes1952,DELAPENAAUERBACH1967603,Parisi:1980ys}, realizes a $d$-dimensional quantum theory via a $(d+1)$-dimensional model with an extra ``fictitious time'' coordinate. For instance, the stochastic quantization of a $1d$ model describing a particle trajectory $x(t)$ introduces a second time coordinate $s$ and postulates a stochastic differential equation $\frac{\partial x(s, t)}{\partial s} = - \frac{\delta S_E [ x ]}{\delta x ( s, t )} + \eta ( x, s )$, where $\eta ( x, s )$ is a noise term and $S_E [ x ]$ is a Euclidean action. In contrast, we work with a single time coordinate $t$ and do not assume the existence of a Euclidean action.\footnote{Another difference between these two approaches is that stochastic quantization appears to be in tension with reflection positivity, at least in some examples \cite{Jaffe2015}. In contrast, we will be able to engineer reflection positivity in our framework by any of several mechanisms.}

This paper is organized as follows. In Section \ref{sec:quantum_sp_nn} we develop the connection between NNs, SPs, and QM, and present the universality result. Section \ref{sec:rp} introduces RP and develops the associated mechanisms and results. Section \ref{sec:deep} presents a novel construction of a large class of RP quantum systems using deep NNs and numerically implement examples. In Section \ref{sec:conclusion}, we summarize our results and present directions for future research. In Appendix \ref{app:commutators}, we present an ancillary discussion, reviewing the observation that commutation relations in QM arise due to the contribution of nowhere-differentiable paths.

\section{Quantum Systems, Stochastic Processes, \& Neural Networks}\label{sec:quantum_sp_nn}

We begin by specifying what we mean by a quantum mechanical theory, at least in Euclidean signature (or imaginary time). There are various conditions that one might wish to impose on a model in order to reasonably describe it as ``quantum,'' such as the existence of uncertainty relations, entanglement, or other cherished features of conventional quantum mechanics. However, at least initially, we will adopt a ``minimal'' requirement of a quantum system which is motivated by restricting the usual defining data that characterizes the local operators\footnote{This data does not \emph{completely} define a quantum field theory, since there may be additional non-local operators, such as Wilson lines or symmetry defect operators, in the spectrum of the theory.} of a $d$-dimensional Euclidean quantum field theory to $d = 1$ dimension, along with two mild assumptions that are motivated by basic physical principles. As we will see, it is natural to speak about such a minimal quantum theory using the language of stochastic processes and the related machinery of neural networks.

\subsection{Requirements of Quantum Models}

When discussing a Euclidean field theory in $d$ spacetime dimensions, 
the local observables of fundamental interest are the correlation functions
\begin{align}
    G^{(n)} ( y_1, \ldots, y_n ) = \langle \phi ( y_1 ) \ldots \phi ( y_n ) \rangle \, ,
\end{align}
where $\phi$ is a local operator in the spectrum of the theory (which may be a scalar or carry additional indices that are suppressed) and the quantities $y_i \in \mathbb{R}^d$, for $i = 1, \ldots, n$, are spacetime points at which the field operators have been inserted. The correlators $G^{(n)}$ are also referred to as Schwinger functions. These are the primary objects of study in constructive approaches to quantum field theory, where one investigates whether these functions satisfy certain conditions such as the Osterwalder-Schrader axioms, which guarantee that a Euclidean theory may be analytically continued to a corresponding Lorentzian theory. When $d \geq 2$, it is natural to refer to the theory associated with such correlation functions as a \emph{field} theory, since the degree of freedom $\phi$ varies in both Euclidean time and in space. 

Restricting to the case $d=1$ that will be relevant for quantum mechanics, we instead consider ``fields'' --- such as the position $x(t)$ of a quantum particle subject to a harmonic potential --- which depend on a time variable $t \in \mathbb{R}$, but not on any spatial coordinates. Therefore, in a quantum mechanical theory, the natural observables are instead correlators
\begin{align}\label{qm_correlator}
    G^{(n)} ( t_1 , \ldots , t_n ) = \langle x ( t_1 ) \ldots x ( t_n ) \rangle \, ,
\end{align}
which depend on $n$ time coordinates $t_i \in \mathbb{R}$. We view the correlation function (\ref{qm_correlator}) as the expectation value, in some probability distribution, of a product of random variables $x ( t_i )$ for $i = 1 , \ldots , n$. A theory of quantum mechanics should yield a prescription for computing any such expectation value, for any choice of times $t_i$ and for any finite $n$.

However, not all collections of correlation functions $G^{(n)}$ are suitable to be interpreted as arising from physically realistic quantum theories. For instance, we envision the random variables $x ( t_i )$ as representing points on the trajectories of quantum particles. On physical grounds, any such trajectory should be a continuous function of Euclidean time, since a particle should not be able to ``jump'' instantaneously from one position to another. We will therefore impose a technical condition called ``mean-square continuity'' on the distribution of random variables $x(t_i)$ in our quantum model. Formally, we require that
\begin{align}\label{mean_square_continuity}
    \lim_{t \to s} \left\langle \, \left| x ( t ) - x ( s ) \right|^2 \, \right\rangle = 0
\end{align}
for all $s$. In particular, this condition implies that the sample paths $x(t)$ are continuous and that the two-point function $G^{(2)} ( t, s )$ is a continuous function of $t$ and $s$.\footnote{We impose the stronger condition of mean-square continuity, rather than merely assuming continuity of sample paths, because the latter is not strong enough to imply continuity of the two-point function.}

A second basic requirement comes from the K\"all\'en-Lehmann spectral representation \cite{Kallen:1952zz,Lehmann1954berEV} of the two-point function, which is a non-perturbative result that holds for quite general quantum field theories. In one Euclidean spacetime dimension, this representation for the two-point function $G^{(2)} ( t, s )$ in any theory of quantum mechanics takes the form
\begin{align}\label{KL_QM}
    G^{(2)} ( t, s ) = \int_0^{\infty} d m \, \rho ( m ) e^{-m | t - s |} \, ,
\end{align}
where $\rho ( m )$ is a finite Borel measure, which means that $\rho$ is positive definite and that its integral is finite (though not necessarily equal to $1$, in which case it is called a probability measure). To see why this assumption is reasonable, it is convenient to consider a conventional quantum system with a discrete spectrum that is bounded below. After shifting the Hamiltonian by a constant so that the ground state energy is zero, one may insert a complete set of energy eigenstates $1 = \sum_n \ket{n} \bra{n}$ to find
\begin{align}\label{spectral_two_point}
    \langle x ( t ) x ( s ) \rangle = \sum_{n=1}^{\infty} e^{- E_n | t - s | } \left| \langle 0 \mid x \mid n \rangle \right|^2 = \int_0^{\infty} d m \, \rho ( m ) e^{-m | t - s |} \, ,
\end{align}
where
\begin{align}
    \rho ( m ) = \sum_{n=1}^{\infty} \left| \langle 0 \mid x \mid n \rangle \right|^2 \delta \left( m - E_n \right) \, .
\end{align}
Although this result is most often applied in higher-dimensional field theory, the version of this expansion for $1d$ Euclidean systems was already used in some early works such as \cite{nmj/1118796540}. Since $\rho ( m )$ must be normalizable, equation (\ref{KL_QM}) implies that the two-point function $G^{(2)} ( t, s )$ must be finite for any values of $t$ and $s$. Indeed, in the coincident-point limit one has $G^{(2)} ( t, t ) = \int_0^\infty d m \, \rho ( m ) < \infty$, and for any other $t \neq s$ the integrand is positive-definite and upper-bounded by $\rho ( m )$, so $G^{(2)} ( t, s )$ converges to a finite value for any $t, s$.

These two basic physical inputs 
are encoded in the following two properties.
\begin{enumerate}[label = (P\arabic*)]
    \item\label{assumption_one} The expectation values in our quantum model obey the mean-square continuity condition (\ref{mean_square_continuity}). This implies that the two-point\footnote{For technical reasons, mean-square continuity does not imply continuity of \emph{all} correlators $G^{(n)}$, but only of the two-point function $G^{(2)}$. However, the latter is all we will need in what follows.} function $G^{(2)} ( t, s )$ is a continuous function of its arguments, and that the sample paths $x(t)$ are continuous in $t$.

    \item\label{assumption_two} The two-point function $G^{(2)} ( t, s )$ is finite for all values of the input times $t$, $s$.
    %That is, there is no choice of arguments $t, s$ for which $G^{(2)} ( t, s )$ diverges.
\end{enumerate}
We therefore adopt the definition that a ``minimal quantum mechanical model'' or MQM is an assignment, for any $n \in \mathbb{N}$ and to every collection of points $t_1 , \ldots, t_n$, of a joint probability distribution for a collection of random variables $x ( t_1 )$, $\ldots$, $x ( t_n )$, such that the expectation values in this distribution obeys \ref{assumption_one} and \ref{assumption_two}.

Here we have focused on Schwinger functions $G^{(n)}$ rather than other common quantities of interest in quantum theories, such as the eigenvalues of the Hamiltonian. However, given the collection of all such correlation functions (\ref{qm_correlator}) for arbitrary times, it is possible to reconstruct other desired observables. For instance, the energy spectrum can be extracted by performing exponential fits of the fall-offs of correlation functions at large time separations. In fact, for quantum systems that satisfy the Osterwalder-Schrader axioms (which we will discuss in Section \ref{sec:rp}) and can thus can be continued to real time, the Osterwalder-Schrader reconstruction theorem guarantees that the correlation functions $G^{(n)}$ are sufficient to reconstruct the Hilbert space and all local observables of the real-time theory \cite{Osterwalder:1973dx}. Therefore, at least for models obeying the OS axioms, we suffer essentially no loss of generality by restricting attention to correlation functions.

Furthermore, although here we primarily discuss quantum systems with a single degree of freedom, one can straightforwardly generalize this definition to models with a collection of bosonic coordinates $x^i ( t )$, to models with anticommuting fields $\psi^i ( t )$ (such as supersymmetric quantum mechanics), and to other settings. Let us also remark that seemingly ``discrete'' quantum models, such as a two-level system, can be accommodated in this framework. For instance, to describe a qubit, we choose three random variables $x^i ( t ) = \left( \phi ( t ) , \theta ( t ) , \psi ( t ) \right)$ using the Euler angle representation, as in the standard construction of the path integral for a spin-$\frac{1}{2}$ particle (see e.g. section 3.3 of \cite{Altland:2006si} for a review).

We refer to the models defined above as ``minimal'' because, depending on which properties one insists upon as necessary for a theory to qualify as genuinely \emph{quantum}, these conditions may not be sufficient, although we believe it is fair to say that they are necessary. For example, one might reasonably adopt the view that a Euclidean (imaginary time) model should only be called ``quantum'' if it can be analytically continued to a Lorentzian (real time) model with unitary time evolution. We will eventually demand that our systems satisfy additional properties, such as time translation symmetry and reflection positivity, which ensure the existence of such a Lorentzian continuation.

\subsection{Quantum Models as Stochastic Processes}

Our definition of a minimal quantum system, obeying \ref{assumption_one} and \ref{assumption_two}, is a special case of the mathematical notion of a stochastic process. In modern treatments, the general definition of a stochastic process is phrased in terms of a $\sigma$-algebra $\mathcal{F}$ of events in a set $\Omega$ called the sample space, along with a measure $\mathbb{P}$ which assigns probability to elements of $\mathcal{F}$. However, for our purposes it will suffice to work with a more elementary definition: by a \emph{stochastic process}, we mean a collection of real-valued random variables
\begin{align}
    \left\{ x ( t ) \mid t \in T \right\} \, ,
\end{align}
where $T$ is the \emph{index set} for the process, which we typically choose to be either $\mathbb{R}$ or a finite interval $[a, b]$. We will often abbreviate a stochastic process $x ( t )$ with a subscript, as in $x_t$. Since for each fixed $t_i \in \mathbb{R}$ the quantity $x ( t_i )$ is a random variable, one can speak of the joint probability distribution\footnote{To avoid confusion, we will always use blackboard bold letters like $\mathbb{P}$ for joint probability distributions over the \emph{outputs} of a stochastic process, to avoid confusion with joint distributions over neural network \emph{parameters} $\theta$ which we will introduce shortly. The latter are written with undecorated letters like $P ( \theta )$.}
\begin{align}\label{stochastic_joint_pde}
    \mathbb{P} \left( x ( t_1 ) , \ldots , x ( t_n ) \right)
\end{align}
for any finite collection of these random variables. The set of all such joint probability distributions, for any collection of $n$ variables $t_i$ in the index set $T$ and for any $n$, is the defining data of the stochastic process $x_t$. 

We therefore see that the definition of a minimal quantum theory given above is an example of a stochastic process, since the defining data for both of these objects is an assignment of a joint probability distribution (\ref{stochastic_joint_pde}) to every finite collection of times $t_i$. The only distinction between these notions is that a minimal quantum theory must satisfy conditions \ref{assumption_one} and \ref{assumption_two}, which means that such minimal quantum models are a proper subset of stochastic processes. More precisely, a MQM is a stochastic process $x_t$ which is mean-square continuous, i.e. it obeys (\ref{mean_square_continuity}), and \emph{square-integrable}, which implies that the correlation function $G^{(2)} ( t, s )$ is finite for every pair of inputs $t$ and $s$, ensuring \ref{assumption_two}.

The preceding definition referred only to the values $x ( t_i )$ of the process $x_t$ at specific points $t_i \in \mathbb{R}$. However, a complementary viewpoint of a stochastic process is as a probability distribution over the space of functions $x : \mathbb{R} \to \mathbb{R}$. From this perspective, one performs a random ``draw'' of a function $x_t$, and then the values of this sample function at points $t_i$ furnish us with realizations of the random variables $x ( t_i )$. 
This interpretation of a stochastic process is aligned with our usual intuition in the path integral formulation of quantum mechanics, where we view expressions like
\begin{align}
    \langle x ( t_1 ) x ( t_2 ) \rangle = \int \mathcal{D} x \, x ( t_1 ) x ( t_2 ) e^{- S [ x ]}
\end{align}
as formally computing an expectation value with respect to a probability measure on the space of paths which is determined by the action $S[x]$ and path integral measure $\mathcal{D} x$. However, we note that the stochastic processes (or minimal quantum systems) that we consider here are \emph{not} assumed to be defined by an action: instead we work directly with an abstract probability distribution (\ref{stochastic_joint_pde}) which plays the role of the combination $\mathcal{D} x \, e^{- S [ x ] }$ and thus captures the combined information (both classical and quantum, whereas the action $S[x]$ alone would define only the classical dynamics) of the model. 

In some mathematical treatments of stochastic processes, one avoids this picture of a ``random function'' (which requires imposing additional regularity assumptions on $x_t$ for technical reasons), preferring to speak only of the random variables $x ( t_i )$ associated with specific points in time. We will not concern ourselves with such technical subtleties in this work, and instead use the two perspectives interchangeably.

If one wishes to take this interpretation of a stochastic process $x_t$ as a random function seriously, it is natural to ask whether one can parameterize the family of functions that are randomly generated by the process $x_t$ in terms of some ``coordinates'' which are real-valued quantities that label a particular function in the family. This leads us to the other main definition of interest in this work, which is that of a neural network.

\subsection{Universality of Neural Networks}

For the purposes of this work,\footnote{Of course, one can consider more general neural networks $\phi_\theta : \mathbb{R}^n \to \mathbb{R}^m$, but since we focus on single-particle quantum mechanics we will take $n = m = 1$.} a \emph{neural network} is a parameterized family of functions $\phi_\theta : \mathbb{R} \to \mathbb{R}$ determined by a set of parameters $\theta$, which are random variables with a joint probability distribution $P ( \theta )$. There can, in general, be many such parameters $\theta^i$ for $i = 1 , \ldots, N$ (or infinitely many), but we will suppress indices and write $\theta$ to indicate the collection of all parameters. Performing a random draw of the parameters $\theta$ then gives a specific instance of a random function $\phi_\theta$. A simple example of a neural network is
\begin{align}\label{cos_example}
    \phi_\theta ( t ) = \cos ( t + \theta ) \, , \qquad \theta \sim U [ 0 , 2 \pi ] \, ,
\end{align}
where the latter notation indicates that $\theta$ is drawn from a uniform distribution on the interval $[0 , 2 \pi]$. We see that the function $\phi_\theta$ therefore inherits randomness from the random variable $\theta$; a particular draw of the parameter, say $\theta = \pi$, then determines an instance of the neural network, here $\phi_\pi ( t ) = \cos ( t + \pi )$. We refer to the fixed functional form that describes the dependence of $\phi_\theta$ on its parameters, such as the cosine function of (\ref{cos_example}), as the \emph{architecture} of the neural network, and we refer to the probability distribution $P ( \theta )$ over the parameters $\theta$ as the \emph{parameter density}.

Clearly every neural network satisfies the definition of a stochastic process, since taken together, the architecture and joint probability distribution $P ( \theta )$ over the parameters determine a joint probability distribution
\begin{align}
    \mathbb{P} \left( \phi_\theta ( t_1 ) , \ldots, \phi_\theta ( t_n ) \right)
\end{align}
for the collection of random variables $\phi_\theta ( t_i )$ given any finite collection of times $t_i$. We say that a stochastic process $x_t$ which admits a representation using a neural network, so that $x_t = \phi_\theta ( t )$ with some parameter density $P ( \theta )$, is a \emph{neural network stochastic process} or NN-SP; similarly, any minimal quantum model $x_t = \phi_\theta ( t )$ which can be represented as a neural network is said to be a \emph{neural network quantum mechanics} or NN-QM. Said differently, a NN-QM is a NN-SP that enjoys properties \ref{assumption_one} and \ref{assumption_two}. For any NN-SP, the correlation functions (\ref{qm_correlator}) can be computed using parameter space integrals as
\begin{align}\label{parameter_space_correlator}
    G^{(n)} ( t_1 , \ldots, t_n ) = \int d \theta \, P ( \theta ) \, \phi_\theta ( t_1 ) \ldots \phi_\theta ( t_n ) \, ,
\end{align}
where the integral runs over all of the parameters $\theta$ but again we suppress indices.

Given that every neural network $\left( \phi_\theta , P ( \theta ) \right)$ determines a stochastic process, one might ask about the opposite direction: does every stochastic process $x_t$ admit a neural network representation in terms of an architecture $\phi_\theta$ and a parameter density $P ( \theta )$? One could answer this question in a rather trivial way by choosing an uncountably infinite set of parameters $\theta^t$, indexed by $t \in \mathbb{R}$, and then declaring that $\phi_\theta ( t ) = \theta^t$. This is simply choosing a separate random variable for the output of the neural network at each point. Since the data of the stochastic process $x_t$ involves a joint probability distribution (\ref{stochastic_joint_pde}) for each collection of points $t_i$, one could then define the parameter density $P ( \theta )$ to agree with the joint distribution of $x_t$ on any finite set of inputs. However, this construction is unwieldy and not very interesting.

One might instead ask whether every stochastic process can be represented via a neural network in a \emph{useful} way, with only countably many parameters and in terms of an architecture which is amenable to theoretical analysis. The answer to this question is provided by the Kosambi-Karhunen-Lo\`{e}ve theorem \cite{kosambi1943statistics,karhunen1947,Loeve1948}, which we now recall.

\begin{thm}[Kosambi-Karhunen-Lo\`{e}ve]\label{KKL_theorem}
    Let $x_t$ be a square-integrable stochastic process defined on an interval $[a, b]$ such that $\langle x ( t ) \rangle = 0$ for all $t \in [ a, b ]$, and suppose that the two-point function $\langle x ( t_1 ) x ( t_2 ) \rangle$ is continuous. Then $x_t$ admits a decomposition
    \begin{align}\label{KKL_decomposition}
        x_t = \sum_{k=1}^{\infty} \theta^k e_k ( t ) \, ,
    \end{align}
    where $e_k$ is a set of continuous, orthogonal real-valued functions on $[a, b]$ and $\theta^k$ are a collection of pairwise uncorrelated (but not necessarily\footnote{We remind the reader that two random variables $x$ and $y$ are uncorrelated if $\langle x y \rangle = \langle x \rangle \langle y \rangle$, and they are independent if their joint probability distribution $P ( x, y ) = P ( x ) P ( y )$ factorizes into a product of marginals. Independent variables are uncorrelated, but uncorrelated variables need not be independent.} independent) random variables.
\end{thm}

First let us discuss the assumptions of the KKL theorem. Two of the conditions of the theorem -- continuity of the two-point function and square-integrability -- are implied by \ref{assumption_one} and \ref{assumption_two}, which are satisfied by any MQM by definition. The last requirement of Theorem \ref{KKL_theorem} is that the stochastic process $x_t$ has zero expectation value for all $t$, which is not a condition which we impose in general. However, this is easy to remedy. Consider some $x_t$ which is square-integrable and has a continuous two-point function but which does not satisfy the zero-mean condition. Any square-integrable stochastic process has finite first moment $\langle x ( t ) \rangle$ for all $t$, by the Cauchy-Schwarz inequality. We define
\begin{align}
    \hat{x}_t = x_t - \langle x ( t ) \rangle \, .
\end{align}
Then $\hat{x}_t$ satisfies all of the assumptions of the Kosambi-Karhunen-Lo\`{e}ve theorem. It follows that the original stochastic process can be written as the sum of the KKL decomposition for $\hat{x}_t$ and the deterministic function $\langle x ( t ) \rangle$. Viewing the expansion coefficients $\theta^k$ appearing in (\ref{KKL_decomposition}) as neural network parameters, and choosing the architecture
\begin{align}\label{general_QM_NN}
    \phi_\theta ( t ) = \langle x ( t ) \rangle + \sum_{k=1}^{\infty} \theta^k e_k ( t )  \, ,
\end{align}
we see that every square-integrable stochastic process $x_t$ with continuous two-point function admits a neural network representation, albeit generically one with a countably infinite set of parameters. As we mentioned, every minimal quantum model obeys \ref{assumption_one} and \ref{assumption_two} and thus satisfies the conditions of square-integrability and continuity of $G^{(2)}$, so we conclude that every MQM is also a NN-QM. In fact, it is now redundant to use the term ``NN-QM'' since we have proven that the space of NN-QMs (i.e. stochastic processes obeying \ref{assumption_one} and \ref{assumption_two}) precisely coincides with the space of MQMs. A Venn diagram of various relevant subsets of the set of stochastic processes appears in Figure \ref{venn}.

\begin{figure}[htbp]
\centering
\includegraphics[width=0.6\linewidth]{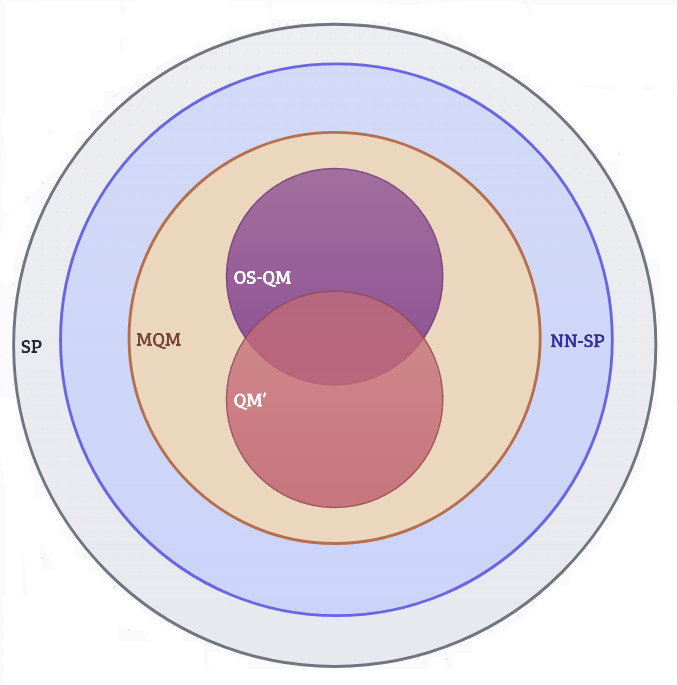}
\caption{A Venn diagram illustrating the various definitions used in this section. Within the space of all stochastic processes (SP), a subset (NN-SP) admit a representation as a neural network. We have proven that every MQM, which is a stochastic process obeying \ref{assumption_one} and \ref{assumption_two}, has a neural network description, and thus $\mathrm{MQM} \subset \mathrm{NN}\text{-}\mathrm{SP}$. One might impose additional restrictions upon minimal quantum models, such as the Osterwalder-Schrader axioms (OS-QM) or another set of conditions defining a notion of quantum mechanics of one's choosing (QM'), which carve out different subsets of MQM.}
\label{venn}
\end{figure}

This observation furnishes us with a sort of ``quantum universal approximation theorem,'' in the sense that any MQM can be approximated to arbitrary accuracy by a truncated architecture of the form (\ref{general_QM_NN}) which includes finitely many of the terms in the sum.\footnote{See \cite{Zhu_2023} for a related but distinct ``stochastic universal approximation theorem'' which establishes that recurrent neural networks can approximate processes obtained as the solution to a stochastic differential equation (SDE) to arbitrary accuracy. Solutions to SDEs are a proper subset of all stochastic processes.} As the number of included terms grows larger, this architecture becomes a better and better approximation to the true statistics of the stochastic process. The specific expansion generated by the Kosambi-Karhunen-Lo\`{e}ve theorem, where the $e_k ( t )$ are chosen to be eigenfunctions of a certain linear operator associated with the two-point function of the process, is ``optimal'' in the sense that a truncated KKL expansion gives the smallest mean-square error among all truncated orthonormal expansions of the stochastic process.

Let us emphasize that this result holds for any quantum model obeying \ref{assumption_one} and \ref{assumption_two}, and is not restricted to cases which satisfy additional conditions that we typically assume when modeling most physical quantum systems, such as time translation invariance or reflection positivity (again, the latter leads to unitary time evolution). This is important since each of these standard assumptions about quantum systems is relaxed in some use cases. For instance, in the study of Floquet systems, one does not impose continuous time translation symmetry, instead assuming the weaker condition that the Hamiltonian is periodic in time, i.e. $H ( t + T ) = H ( t )$ for some period $T$. Likewise, when considering open quantum systems, one does not generally have unitary time evolution; rather the dynamics are described by the Lindblad master equation, which models a system that interacts with some environment or ``bath'' and thus evolves in a non-unitary way. Both of these applications can be accommodated within the framework of NN-QM, since we have made only the mild assumptions \ref{assumption_one}, \ref{assumption_two} about the underlying process $x_t$.

Having said this, we will now set aside these ``unconventional'' applications, and for the remainder of this work we will focus on standard quantum systems which enjoy time translation invariance and unitary time evolution. The additional assumptions on $x_t$, or $\phi_\theta$, which guarantee these properties are the subject of the next section.

\section{Continuation to Real Time via Reflection Positivity}\label{sec:rp}

In quantum field theory -- including the case of $(0+1)$-dimensional field theories which we consider here -- the conditions which guarantee that one can analytically continue a Euclidean theory to Lorentz signature are encoded by the Osterwalder-Schrader (OS) axioms \cite{Osterwalder:1973dx}. The OS axioms are phrased as assumptions about the correlation functions (\ref{qm_correlator}) which, in our case, are expectation values taken with respect to the joint probability distribution (\ref{stochastic_joint_pde}) associated with a quantum system $x_t$ or $\phi_\theta ( t )$. Let us briefly recall each of the OS axioms in general, and then specialize to the case of quantum mechanics.

\begin{enumerate}[label = (E\arabic*), start = 0]
    \item\label{temperedness} The Schwinger functions should be tempered distributions away from coincident points. This is a technical assumption that is mostly relevant for theories in $d \geq 2$.

    \item\label{euclidean_invariance} The functions $G^{(n)} ( t_1, \ldots , t_n )$ must be Euclidean covariant. In one spacetime dimension, this means that the stochastic process $x_t$ must enjoy symmetry under time translations $t \to t + \tau$ and reflections $t \to - t$.

    \item\label{symmetric} The Schwinger functions must be symmetric under permutation of the insertion points. This condition is automatic when $G^{(n)} ( t_1, \ldots , t_n )$ is an expectation value, since the product of the random variables appearing under an integral commutes.\footnote{Note that permutation symmetry holds for expectation values of arbitrary operators at \emph{distinct} points, but that coincident-point limits of correlators involving both local operators and their derivatives may depend on ordering of limits; the latter encode, for instance, non-trivial commutation relations.} 
    
    \item\label{cluster_decomposition} Correlation functions should satisfy cluster decomposition, whereby $G^{(p+q)}$ reduces to a product $G^{(p)} G^{(q)}$ when collections of insertion points are taken far apart.

    \item\label{reflection_positivity} The Schwinger functions must exhibit \emph{reflection positivity}, which means that
    \begin{align}\label{rp_condition}
        \langle F ( x ( t_1 ) , \ldots, x ( t_n ) ) \left( F ( x ( - t_1 ) , \ldots, x ( - t_n ) ) \right)^\ast \rangle \geq 0
    \end{align}
    for any collection of positive times $t_i > 0$ and any bounded measurable function $F$ of $n$ variables, which may be complex-valued. Here $\ast$ denotes complex conjugation.
\end{enumerate}

If all of the OS axioms \ref{temperedness} - \ref{reflection_positivity} are satisfied, then we are guaranteed that the Euclidean-time quantum system described by $x_t$ can be analytically continued to a real-time model which inherits certain desirable properties from the imaginary-time system. For instance, the assumption \ref{euclidean_invariance} for the Euclidean model implies the corresponding time-translation and time-reflection symmetries for the real-time model. As we mentioned, the axiom \ref{reflection_positivity} of reflection positivity is translated to the property of unitary time evolution.\footnote{See \cite{cmp/1104162595} for a study of Euclidean field theories where the assumption of reflection positivity is relaxed.}

It will be useful to introduce some terminology in order to characterize whether a given Euclidean quantum system satisfies the various OS axioms. Since we have seen in Section \ref{sec:quantum_sp_nn} that quantum systems, stochastic processes, and neural networks are closely related, we will freely pass between these descriptions in what follows.

Two stochastic processes $x_t$ and $y_t$ are said to be \emph{equivalent} if, for any collection of times $t_1 , \ldots , t_n$, the joint probability distributions $\mathbb{P} \left( x ( t_1 ) , \ldots , x ( t_n ) \right)$ and $\mathbb{P} \left( y ( t_1 ) , \ldots , y ( t_n ) \right)$ are equal. In particular, equivalent stochastic processes have identical correlators. A stochastic process $x_t$ is called \emph{stationary} if for any $\tau \in \mathbb{R}$ the processes $x_t$ and $x_{t + \tau}$ are equivalent, i.e. if for any collection of times $t_1 , \ldots , t_n$ and any fixed $\tau \in \mathbb{R}$ one has
\begin{align}
    \mathbb{P} \left( x ( t_1 ) , \ldots , x ( t_n ) \right) = \mathbb{P} \left( x ( t_1 + \tau ) , \ldots , x ( t_n + \tau ) \right) \, .
\end{align}
Furthermore, we say that a stochastic process $x_t$ is \emph{symmetric} if the processes $x_t$ and $x_{-t}$ are equivalent, so that for any collection of times $t_1 , \ldots , t_n$,
\begin{align}
    \mathbb{P} \left( x ( t_1 ) , \ldots , x ( t_n ) \right) = \mathbb{P} \left( x ( - t_1 ) , \ldots , x ( - t_n ) \right) \, .
\end{align}
Translating to physical language, we see that a stationary stochastic process corresponds to a time-translation invariant quantum system, and a symmetric stochastic process represents a time-reflection invariant quantum system. It is then easy to characterize the quantum systems which satisfy the Euclidean covariance condition: a theory described by a stochastic process $x_t$ obeys \ref{euclidean_invariance} if and only if $x_t$ is stationary and symmetric.

This completes the discussion of \ref{euclidean_invariance}, and the axioms \ref{temperedness} and \ref{symmetric} are essentially automatic, which leaves only two of the OS axioms. For stationary processes, the assumption \ref{cluster_decomposition} of cluster decomposition is equivalent to a property called ``mixing'' in the literature on ergodic theory. A stochastic process is said to be \emph{mixing} if its values at widely-separated times are asymptotically independent -- for instance, mixing implies that the statistical dependence between $x ( t_1 )$ and $x ( t_2 )$ goes to zero as $| t_1 - t_2 |$ tends to infinity. Mixing can also be thought of as an assumption about the decay of correlations,\footnote{This can be made precise: a stationary stochastic process $x_t$ is mixing if and only if the covariance between any pair of bounded observables $f ( x_0 )$ and $g ( x_t )$ tends to zero as $t \to \infty$.} and is closely related to ergodicity (in fact, mixing implies ergodicity in general). From the perspective of physics, the condition \ref{cluster_decomposition} is a weak form of locality in time, and is automatically satisfied by conventional quantum systems which have local interactions.

Although it is interesting to investigate the various necessary and/or sufficient conditions for a stochastic process to enjoy cluster decomposition (i.e. to be mixing), we will now set this axiom aside and focus most of our subsequent discussion on \ref{reflection_positivity}, since reflection positivity turns out to be more subtle. We stress that RP is distinct from time reflection symmetry, i.e. the statement that $x_t$ and $x_{-t}$ are equivalent. A symmetric stochastic process is one whose statistical properties (joint probability distributions) are invariant under a simultaneous replacement of \emph{all} input times $t_i$ by their time-reversed values $- t_i$. In contrast, the reflection positivity condition (\ref{rp_condition}) is a statement about expectation values for which \emph{half} of the insertion times take time-reversed values $- t_i$. 

Reflection positivity is a rather delicate condition, so we devote the next two subsections to an exploration of various scenarios in which RP holds.

\subsection{A Parameter Splitting Mechanism}\label{sec:param_split}

In this section, we will pass to the neural network description of a quantum system and outline one mechanism by which reflection positivity \ref{reflection_positivity} can be achieved. The advantage of the neural network formulation is that we may write expectation values as explicit parameter space integrals, which can make certain manipulations easier.

Consider a neural network quantum system $x_t$ which is presented via an architecture $x ( t ) = \phi_\theta ( t )$ with associated parameter density $P ( \theta )$. First we will express the condition for reflection positivity in this language. Fix a collection of positive times $t_1 , \ldots, t_n > 0$ and a (possibly complex) function $F$ of $n$ variables. The condition (\ref{rp_condition}) then reads
\begin{align}\label{nn_rp_condition}
    \int d \theta \, P ( \theta ) \, F \left( \phi_\theta ( t_1 ) , \ldots, \phi_\theta ( t_n ) \right) \left( F \left( \phi_\theta ( - t_1 ) , \ldots, \phi_\theta ( - t_n ) \right) \right)^\ast \geq 0 \, ,
\end{align}
where again $\ast$ denotes complex conjugation. Reflection positivity requires that the inequality (\ref{nn_rp_condition}) holds for any suitable function $F$ and any finite collection of positive times.

Let us now ask how one might engineer a NN-QM such that reflection positivity is automatically satisfied. One obvious way to ensure that a quantity is positive is to write it as a perfect square. We can attempt to force the integrand of (\ref{nn_rp_condition}) to be a perfect square by ``splitting'' the parameter space integral and writing it in a factorized form as follows. Assume that the set $\left\{ \theta \right\}$ of parameters associated with the NN-QM $( \phi_\theta , P ( \theta ) )$ decomposes into three mutually disjoint subsets:
\begin{align}
    \left\{ \theta \right\} = \left\{ \theta^0 \right\} \cup \left\{ \theta^+ \right\} \cup \left\{ \theta^- \right\} \, .
\end{align}
As usual, there may be many parameters in each subset, but we suppress indices and write $\theta^0$, $\theta^+$, and $\theta^-$ to denote the collection of all parameters in each subset.

The critical assumption we now make is the following. Suppose that, for positive times $t > 0$, the architecture $\phi_\theta ( t )$ depends on the parameters $\theta^+$ and $\theta^0$ but not on $\theta^-$. Likewise, at negative times $t < 0$, assume that $\phi_\theta ( t )$ depends on $\theta^-$ and $\theta^0$ but not $\theta^+$. 
This condition allows us to rewrite (\ref{nn_rp_condition}) as
\begin{align}\label{RP_intermediate}
    \langle F \left( T F \right)^\ast \rangle &= \int d \theta^0 \, P ( \theta^0, \theta^+, \theta^- ) \, \left( \int d \theta^+ F \left( \phi_\theta ( t_1 ) , \ldots, \phi_\theta ( t_n ) \right) \right) \nonumber \\
    &\qquad \qquad \cdot \left( \int d \theta^- \left( F \left( \phi_\theta ( - t_1 ) , \ldots, \phi_\theta ( - t_n ) \right) \right) \right)^\ast \, ,
\end{align}
where we have introduced the shorthand $\langle F \left( T F \right)^\ast \rangle$ for the integral on the left side of (\ref{nn_rp_condition}). Here $T$ is the time-reversal operator which acts on any function of a collection of insertions $\phi_\theta ( t_i )$ by reversing each time as $T : \phi_\theta ( t_i ) \mapsto \phi_\theta ( - t_i )$.

We have succeeded in factorizing the $F$ and $\left( T F \right)^\ast$ insertions, but not the parameter density. To achieve the latter, let us pause to consider which conditions on $P ( \theta^0, \theta^+, \theta^- )$ are natural to impose. Although we are presently focused on reflection positivity, we are ultimately interested in engineering a NN-QM that obeys \emph{all} of the OS axioms. In particular, if our neural network satisfies the OS axiom \ref{euclidean_invariance}, then the associated stochastic process $x_t = \phi_\theta ( t )$  is symmetric and hence $x_t$ and $x_{-t}$ are equivalent processes. This equivalence means that various joint probability distributions must be equal, such as
\begin{align}\label{rp_joint_pdf_one}
    \mathbb{P} \left( \phi_\theta ( t_1 ) , \ldots , \phi_\theta ( t_n ) \right) = \mathbb{P} \left( \phi_\theta ( - t_1 ) , \ldots , \phi_\theta ( - t_n ) \right) \, .
\end{align}
In the neural network description, joint probability distributions on outputs such as (\ref{rp_joint_pdf_one}) are controlled by the parameter density. A sufficient set of conditions that one may impose on the density $P ( \theta )$ which ensures that $\phi_\theta ( t )$ generates a symmetric stochastic process is
\begin{align}\label{factorized_density}
    P ( \theta^0, \theta^+, \theta^- ) = P_0 ( \theta^0 ) \cdot P_+ ( \theta^0, \theta^+ ) \cdot P_- ( \theta^0, \theta^- ) \, ,
\end{align}
along with an assumption that time-reversal of the input $t$ can be ``absorbed'' into an action on the parameters, in the sense that for $t > 0$ the architecture obeys
\begin{align}\label{param_absorb_one}
    \phi_{\theta^0 , \theta^+} ( t ) = \phi_{\theta^0 , \theta^{+ \prime}} ( - t ) \, , \qquad \phi_{\theta^0 , \theta^-} ( - t ) = \phi_{\theta^0 , \theta^{- \prime}} ( t ) \, ,
\end{align}
where $\theta^{\pm \prime}$ are re-defined sets of parameters, such that
\begin{align}\label{param_absorb_two}
    P_+ \left( \theta^0 , \theta^{- \prime} \right) = P_- ( \theta^0 , \theta^- ) \, , \qquad P_- ( \theta^0 , \theta^{+ \prime} ) = P_+ ( \theta^0 , \theta^+ ) \, .
\end{align}
Equation (\ref{param_absorb_two}) looks somewhat strange due to expressions like $P_\pm \left( \theta^0 , \theta^{\mp \prime} \right)$ which involve evaluation of $P_\pm$ on inputs with the ``wrong-sign'' subscripts $\mp$. However, let us emphasize that all of the parameters in these equations are simply dummy variables. At the risk of being pedantic, suppose there are $N_0$ parameters $\theta^0$ and $N_{\pm}$ parameters $\theta^{\pm}$. We assume that $N_+ = N_- = N$, so that both $P_+$ and $P_-$ are functions of $N_0 + N$ variables. Then the first equation of (\ref{param_absorb_two}) states that the output of the function $P_-$, when evaluated on the $N_0 + N$ inputs $\theta^0$ and $\theta^-$, matches the output of the (different) function $P_+$, when $P_+$ is evaluated on the (different) set of $N_0 + N$ inputs $\theta^0$ and $\theta^{- \prime}$. Here $\theta^{- \prime}$ is a new set of $N$ variables that are obtained by applying some $\mathbb{Z}_2$ transformation to the $N$ variables $\theta^-$.

This ``parameter absorption mechanism'' was first introduced in \cite{Maiti:2021fpy} to engineer neural network field theories with particular symmetries. The specific mechanism of equations (\ref{param_absorb_one}) - (\ref{param_absorb_two}) has been chosen so that the NN-QM defined by $x_t = \phi_\theta ( t )$ enjoys time-reflection invariance, i.e. so that $x_t$ and $x_{-t}$ are equivalent stochastic processes.

Finally, let us return to reflection positivity. With the additional conditions on the parameter density described above, (\ref{RP_intermediate}) becomes
\begin{align}\label{param_splitting_final}
    \langle F \left( T F \right)^\ast \rangle &= \int d \theta^0 \, P ( \theta^0 ) \, \left( \int d \theta^+ \, P_+ ( \theta^0 , \theta^+ ) \, F \left( \phi_\theta ( t_1 ) , \ldots, \phi_\theta ( t_n ) \right) \right) \nonumber \\
    &\qquad \qquad \qquad \qquad \qquad \cdot \left( \int d \theta^- \, P_- ( \theta^0 , \theta^- ) \, \left( F \left( \phi_\theta ( - t_1 ) , \ldots, \phi_\theta ( - t_n ) \right) \right) \right)^\ast \, \nonumber \\
    &= \int d \theta^0 \, P ( \theta^0 ) \, \left( \int d \theta^+ \, P_+ ( \theta^0 , \theta^+ ) \, F \left( \phi_\theta ( t_1 ) , \ldots, \phi_\theta ( t_n ) \right) \right) \nonumber \\
    &\qquad \qquad \qquad \qquad \qquad \cdot \left( \int d \theta^- \, P_- ( \theta^0 , \theta^- ) \, \left( F \left( \phi_{\theta'} ( t_1 ) , \ldots, \phi_{\theta'} ( t_n ) \right) \right) \right)^\ast \nonumber \\
    &= \int d \theta^0 \, P ( \theta^0 ) \, \left( \int d \theta^+ \, P_+ ( \theta^0 , \theta^+ ) \, F \left( \phi_\theta ( t_1 ) , \ldots, \phi_\theta ( t_n ) \right) \right) \nonumber \\
    &\qquad \qquad \qquad \qquad \qquad \cdot \left( \int d \theta^{- \prime} \, P_+ ( \theta^0 , \theta^{- \prime} ) \, \left( F \left( \phi_{\theta'} ( t_1 ) , \ldots, \phi_{\theta'} ( t_n ) \right) \right) \right)^\ast \nonumber \\
    &= \int d \theta^0 \, P ( \theta^0 ) \, \left| \int d \theta^+ \, P_+ ( \theta^0 , \theta^+ ) \, F \left( \phi_\theta ( t_1 ) , \ldots, \phi_\theta ( t_n ) \right) \right|^2 \, .
\end{align}
In the first step, we have used (\ref{param_absorb_one}) to absorb all minus signs in the evaluations of $\phi_\theta$ at negative times into redefinitions of the parameters. We then applied (\ref{param_absorb_two}) to replace $P_-$ in the second integral with $P_+$ evaluated on the transformed parameters, and we have changed integration variables from $\theta^-$ to $\theta^{- \prime}$, which we assume gives a trivial Jacobian so that no additional factors are generated from the change of measure.\footnote{Recall that the transformation from $\theta^-$ to $\theta^{- \prime}$ absorbs the time-reversal operation $t \to -t$, which is a $\mathbb{Z}_2$ involution. Therefore, the map from $\theta^-$ to $\theta^{- \prime}$ squares to the identity, so it is reasonable to assume that it has trivial Jacobian. For instance, if the transformation on the $\theta^-$ is implemented by a linear operator $M$, i.e. $\theta^{\prime - i} = \tensor{M}{^i_j} \theta^{- j}$, this is guaranteed since $M^2 = I$ so $| \det ( M ) | = 1$.} Finally, in the last step we recognize that both $\theta^+$ and $\theta^{- \prime}$ are dummy variables that are integrated over, which means that the integrals inside the two sets of parentheses are in fact identical. This allowed us to write $\langle F \left( T F \right)^\ast \rangle$ as an integral of a non-negative density $P ( \theta^0 )$ multiplied by a perfect square, which is thus also non-negative, as desired. We conclude that
\begin{align}
    \langle F \left( T F \right)^\ast \rangle \geq 0
\end{align}
for a NN-QM satisfying the assumptions described above. This establishes RP for any quantum system represented by a neural network that enjoys parameter splitting.

Let us now summarize the (rather involved) argument presented above, and make some further comments. We have determined a set of sufficient conditions for a NN-QM to enjoy reflection positivity. The key ingredients in our mechanism are:
\begin{enumerate}[label = (\Roman*)]
    \item\label{decouple} The parameter-space degrees of freedom which are relevant for positive times and negative times decouple from one another. By this, we mean that there is one set of parameters $\theta^+$ that describe the NN-QM at positive times, but which have no effect on the model at negative times; likewise, there is another set of parameters $\theta^-$ that affect the NN at negative times but not positive times.

    \item The parameter density enjoys a factorized form (\ref{factorized_density}), and one can employ a ``parameter absorption trick'' to relate the densities at positive and negative times. This assumption also implies time-reflection symmetry (not to be confused with RP).
\end{enumerate}
The first ingredient \ref{decouple} is reminiscent of locality in time; the behavior of the neural network at large positive times is independent of its behavior in the far past, at negative times, which suggests some form of local interactions. 

Alternatively, one can view the parameter splitting as a sort of ``parameter space Markov property'' in the following sense. Suppose we view the parameters $\theta^+$ as describing $\phi_\theta ( t )$ for $t \geq 0$, $\theta^-$ as describing the NN for $t \leq 0$, and $\theta^0$ as determining the behavior of $\phi_\theta ( 0 )$. Say that we are given the output of a fixed realization of the NN on the half-line $t \leq 0$ and wish to predict its outputs in the future, at positive times. Parameter splitting roughly states that, to make this prediction, it is not necessary to condition on the neural network outputs at negative times (which depend on $\theta^-$), but only on the most recent time $t = 0$ associated with $\theta^0$. This can be interpreted as a ``memorylessness'' assumption, at least around the point $t = 0$. The connection between Markov-like behavior and reflection positivity is not an accident; in Section \ref{sec:markov} we will return to this topic, although there we consider the more conventional ``output-space Markov property'' rather than this ``parameter-space Markov property'' represented by the splitting mechanism.

We also note that, while our parameter splitting mechanism is sufficient for RP, it depends on the precise parameterization of the neural network. A generic quantum system can be presented by neural networks in multiple equivalent ways, and even if one representation enjoys parameter splitting, another presentation may fail to split. Therefore it cannot be necessary that a NN with a given parameterization must enjoy a parameter splitting in order to be RP. A simple way to see this is to consider the KKL decomposition of a minimal quantum mechanical model on an interval $[-T, T]$, which we repeat:
\begin{align}
    \phi_\theta ( t ) = \langle x ( t ) \rangle + \sum_{k=1}^{\infty} \theta^k e_k ( t )  \, .
\end{align}
Each of the continuous, pairwise orthogonal functions $e_k ( t )$ on $[-T, T]$ generically has support both on the positive half-interval $[0, T]$ and on the negative half-interval $[-T, 0]$. For instance, in the KKL decomposition of a Brownian path or an Ornstein-Uhlenbeck process, the $e_k ( t )$ can be chosen to be trigonometric functions. In such a parameterization, there is no splitting of the parameters $\theta^k$ into appropriate subsets $\theta^+$ and $\theta^-$; every one of the $\theta^k$ affects the outputs of the neural network $\phi_\theta$ at both positive and negative times. But on the other hand, \emph{every} MQM admits a KKL representation, including the reflection-positive processes. Therefore an RP process, presented via a neural network with some choice of parameterization, need not exhibit parameter splitting with respect to that particular parameterization. Said differently, the existence or non-existence of a parameter splitting is a question which depends on the given neural network representation $\phi_\theta ( t )$ for a quantum system, rather than an intrinsic question about the system itself.\footnote{As an extreme example, we commented below (\ref{parameter_space_correlator}) that any stochastic process can be represented by a neural network $\phi_\theta ( t ) = \theta^t$ with an uncountable infinity of parameters $\theta^t$, $t \in \mathbb{R}$. In such a parameterization, the parameters associated with $t > 0$ trivially split from those that describe $t < 0$. This does not imply that any process is RP, since $P ( \theta )$ need not have the form required by our argument above.}

\subsubsection*{\ul{\it Engineering parameter splitting about $t = 0$}}

The preceding subsection explained that a set of parameter splitting assumptions on a NN-QM implies reflection positivity. One might then ask whether this set of assumptions is actually satisfied in any explicit examples. As it turns out, given \emph{any} neural network $\left( \phi_\theta , P ( \theta ) \right)$, it is possibly to transform this NN into a modified model $\left( \widetilde{\phi}_\theta , \widetilde{P} ( \theta ) \right)$ which obeys the parameter splitting requirements and therefore exhibits reflection positivity. In principle, this allows us to generate a large class of examples of NN-QM systems which enjoy reflection positivity, although they generally do not satisfy the other OS axioms.

Let us now describe the construction. Begin with any NN-QM defined by an architecture $\phi_\theta : \mathbb{R} \to \mathbb{R}$ and a probability distribution $P ( \theta )$ over parameters. In addition to the parameters $\theta$ which are already present in the ``seed'' model, we introduce two additional parameters $\theta^+$ and $\theta^-$. Next, we promote $\phi_\theta$ to an upgraded architecture
\begin{align}\label{promote_phis}
    \widetilde{\phi}_\theta ( t ) = \begin{cases} \phi_\theta ( \theta^+ t) &\text{if } t \geq 0 \\  \phi_\theta ( \theta^- t ) &\text{if } t < 0 \end{cases} \, .
\end{align}
The choice of which parameter $\theta^{\pm}$ describes the behavior at $t = 0$ is immaterial; for concreteness, we have chosen $\theta^+$. By construction, we see that the architecture is independent of $\theta^-$ at positive times and independent of $\theta^+$ at negative times, so these parameters satisfy the first requirement of our parameter splitting mechanism.

We must also specify the form of the joint density $\widetilde{P} ( \theta^+ , \theta^- , \theta )$ over all of the parameters. Here there is a great deal of freedom. If there are $N$ parameters $\theta$ in the seed model, choose \emph{any} density $p ( \theta^+ , \theta )$ on $N+1$ parameters. Then we define
\begin{align}\label{general_RP_mechanism_density}
    \widetilde{P} ( \theta^+ , \theta^- , \theta ) = p \left( \theta^+ , \theta \right) \cdot p \left( - \theta^- , \theta \right) \, .
\end{align}
For instance, if one wishes to choose the new density so that its dependence on the seed parameters $\theta$ is unchanged, one could take a factorized form
\begin{align}
    p \left( \theta^+ , \theta \right) = \sqrt{ P ( \theta ) } \cdot f ( \theta^+ ) \, ,
\end{align}
for some function $f$ of one variable, where $P ( \theta )$ is the parameter density of the original model. In this case, we see that the density for the new NN-QM takes the simple form $\widetilde{P} ( \theta^+ , \theta^- , \theta ) = P ( \theta ) f ( \theta^+ ) f ( - \theta^- )$. However, this special choice is not required, so we will continue to work with the general density (\ref{general_RP_mechanism_density}).

We have engineered the new architecture (\ref{promote_phis}) and density (\ref{general_RP_mechanism_density}) so that they admit a parameter absorption trick of the form described around equations (\ref{param_absorb_one}) - (\ref{param_absorb_two}). To see this, consider a positive time $t > 0$ and compare $\widetilde{\phi}_\theta ( t ) = \phi_\theta ( \theta^+ t )$ to its time-reversal $\widetilde{\phi}_\theta ( - t ) = \phi_\theta ( - \theta^- t )$. Clearly the two expressions agree if one identifies $\theta^+ = - \theta^-$. Thus we see that the time-reversal of the input $t$ can be absorbed into a redefinition of the parameters. To make this very explicit, let us define $\theta^{\prime -} = - \theta^-$ and carefully show the parameter absorption mechanism, writing $\widetilde{\phi}_{\theta^+, \theta^-, \theta}$ so that the dependence on all of the parameters is apparent. Then one has
\begin{align}\label{explicit_absorption}
    \widetilde{\phi}_{\theta^+ , \theta^- , \theta } ( - t ) = \phi_\theta ( - \theta^- t ) = \phi_\theta \left( \theta^{\prime -} t \right) = \widetilde{\phi}_{\theta^{\prime -} , \theta^- , \theta } ( t ) \, .
\end{align}
The point is that, in the final expression of (\ref{explicit_absorption}), the parameter $\theta^+$ which usually sits in the first slot of $\widetilde{\phi}_{\theta^+, \theta^-, \theta}$ has been replaced with $\theta^{\prime -} = - \theta^-$. This procedure allows us to perform the parameter absorption required by condition (\ref{param_absorb_one}) in our splitting mechanism. The second condition (\ref{param_absorb_two}) is ensured by our choice (\ref{general_RP_mechanism_density}) for the density. Identifying
\begin{align}
    P_+ ( \theta^+ , \theta ) = p ( \theta^+ , \theta ) \, , \qquad P_- ( \theta^- , \theta ) = p ( - \theta^- , \theta ) \, ,
\end{align}
to facilitate comparison with the notation of (\ref{param_absorb_two}), clearly one has
\begin{align}
    P_+ ( - \theta^+ , \theta ) = P_- ( \theta^+ , \theta ) \, , \qquad P_- ( - \theta^- , \theta ) = P_+ ( \theta^- , \theta ) \, ,
\end{align}
which is the other condition required in the parameter splitting mechanism.

By the general argument of the preceding subsection, we therefore conclude that the architecture (\ref{promote_phis}) and density (\ref{general_RP_mechanism_density}) define a reflection-positive NN-QM given \emph{any} input network $\phi_\theta$. Although this conclusion follows from the results above, it may be instructive to explicitly present the proof of reflection positivity for this model to see how the pieces fit together. Given positive times $t_1 , \ldots, t_n > 0$ and a function $F$ of $n$ variables, we have
\begin{align}
    \langle F \left( T F \right)^\ast \rangle &= \int d \theta \, \left( \int d \theta^+ \, p ( \theta^+ , \theta ) \, F ( \widetilde{\phi}_\theta ( t_1 ) , \ldots , \widetilde{\phi}_\theta ( t_n ) ) \right) \nonumber \\
    &\qquad \qquad \cdot \left( \int d \theta^- \, p ( - \theta^- , \theta ) \, F ( \widetilde{\phi}_\theta ( - t_1 ) , \ldots , \widetilde{\phi}_\theta ( - t_n ) ) \right)^\ast \nonumber \\
    &= \int d \theta \, \left( \int d \theta^+ \, p ( \theta^+ , \theta ) \, F ( \widetilde{\phi}_\theta ( t_1 ) , \ldots , \widetilde{\phi}_\theta ( t_n ) ) \right) \nonumber \\
    &\qquad \qquad \cdot \left( \int d \theta^{\prime -} \, p ( \theta^{\prime -} , \theta ) \, F ( \widetilde{\phi}_{\theta'} ( t_1 ) , \ldots , \widetilde{\phi}_{\theta'} ( t_n ) ) \right)^\ast \nonumber \\
    &= \int d \theta \, \left| \int d \theta^+ \, p ( \theta^+ , \theta ) \, F ( \widetilde{\phi}_\theta ( t_1 ) , \ldots , \widetilde{\phi}_\theta ( t_n ) ) \right|^2 \nonumber \\
    &\geq 0 \, .
\end{align}
The key step is the redefinition of parameters to $\theta^{\prime -} = - \theta^-$ and the change of variables\footnote{Note that no signs or factors are generated when we change integration variables from $\theta^-$ to $\theta^{\prime -}$, since $\int_{-\infty}^{\infty} d \theta^- \, f ( - \theta^- ) = \int_{\infty}^{-\infty} \left( - d \theta^{\prime -} \right) \, f ( \theta^{\prime -} ) = + \int_{-\infty}^{\infty} d \theta^{\prime -} \, f ( \theta^{\prime -} )$. The sign from the change of measure from $d \theta^-$ to $d \theta^{\prime -}$ cancels against a compensating sign from the reversal of the bounds of integration.
} in the integral, which upon recognizing that $\theta^+$ and $\theta^{\prime -}$ are dummy variables which are integrated over, allows us to conclude that the two integrals in parentheses are equal and thus the expectation value $\langle F \left( T F \right)^\ast \rangle$ is the integral of a perfect square. This completes the proof of reflection positivity for this class of NN-QM models.

\subsubsection*{\ul{\it Translation invariance and the inevitability of non-analyticity}}

The construction of $\tilde{\phi}_\theta$ in equation (\ref{promote_phis}), with an appropriate parameter density (\ref{general_RP_mechanism_density}), produces a reflection positive NN-QM given any input neural network $\phi_\theta$. However, the resulting model is \emph{not} time-translation invariant, even if the original network $\phi_\theta$ enjoys translational symmetry. This is manifest from the form of the architecture (\ref{promote_phis}), which explicitly singles out the time $t = 0$ as special. For example, a correlation function like $\langle \widetilde{\phi}_\theta ( - 1 ) \widetilde{\phi}_\theta ( 1 ) \rangle $ is sensitive to the statistics of both $\theta^+$ and $\theta^-$, while shifting $t \to t + 2$ gives the correlator $\langle \widetilde{\phi}_\theta ( 1 ) \widetilde{\phi}_\theta ( 3 ) \rangle$, which now only involves $\theta^+$. A priori there is no reason for the shifted and un-shifted correlators to agree, and in general they will not match.

This brings us to an important point about NN-QM systems which satisfy all of the OS axioms. In the preceding discussion we have been focused on reflection positivity in the form (\ref{rp_condition}), which involves reflections about the point $t = 0$. However, for a quantum system which enjoys both time-translation invariance and reflection positivity, the model must exhibit positivity with respect to reflections about \emph{any} time $t = a$. This is a much stronger condition, and not one which can be easily engineered using parameter splitting.

To better understand the tension between the parameter splitting mechanism for RP and translation invariance, let us think about how one might try to engineer reflection positivity about two points $t_1$, $t_2$ using the above construction. Na\"ively, one would now need to introduce four extra parameters $\theta^{1 \pm}$ and $\theta^{2 \pm}$ with the properties that $\phi_\theta ( t )$ is independent of $\theta^{1 +}$ for $t < t_1$, independent of $\theta^{1 -}$ for $t > t_1$, independent of $\theta^{2 +}$ for $t < t_2$, and independent of $\theta^{2 -}$ for $t > t_2$. If one proceeds in this way, one would be naturally led to introduce an uncountable infinity of parameters $\theta^{a \pm}$ for $a \in \mathbb{R}$, which seems excessive.

Another way to think of this tension is to return to our intuitive picture of the parameter splitting mechanism as a ``decoupling'' between the parameter-space degrees of freedom $\theta^+$ which control the positive-time behavior of the network and the parameters $\theta^-$ which determine the behavior for negative times. In order to engineer reflection positivity about \emph{any} point $t = a$ via parameter splitting, it would therefore seem that we would need the behavior of the neural network at any time $t$ to admit a decoupling from its behavior at other times $t' \neq t$. In particular, it appears that the neural network outputs at time $t$ should never be determined from the outputs in a neighborhood of $t$, since this would preclude any such decoupling. This can be made precise in the following result.

\begin{prop}
    Let $\phi_\theta$ be a NN-QM which is reflection-positive with respect to all times. Assume that the mechanism by which $\phi_\theta$ achieves RP about each time $t$ is by parameter splitting. Then $\phi_\theta$ cannot be an analytic function of $t$ at any point in its domain.
\end{prop}

\begin{proof}

Suppose by way of contradiction that $\phi_\theta$ is an analytic function of time near a point $t = a$. This means that, within some finite radius of convergence $| t - a| < R$, the outputs of the neural network are described by a convergent Taylor series expansion
\begin{align}\label{NN_analytic}
    \phi_\theta ( t ) = \sum_{n=0}^{\infty} c_n ( \theta ) t^n \, .
\end{align}
Note that $\phi_\theta$ is still a random function, so different draws of the neural network will give different analytic functions near $t = a$, with different Taylor coefficients that are determined by the values of the parameters. We indicate this dependence by writing the series coefficients as $c_n ( \theta )$. Equation (\ref{NN_analytic}) defines what we mean by analyticity for NNs.

By assumption, $\phi_\theta$ achieves reflection positivity about $t = a$ by a parameter splitting mechanism. We therefore wish to identify two sets of parameters $\{ \theta^+ \}$, $\{ \theta^- \}$ such that $\phi_\theta ( t )$ is independent of $\theta^+$ for $t < a$ and $x_\theta ( t )$ is independent of $\theta^-$ for $t > a$.

However, for each fixed parameter $\overbar{\theta}$, there are only two possibilities: either none of the Taylor coefficients depend on $\overbar{\theta}$, or at least one of the $c_n$ depends on this $\overbar{\theta}$. 
\begin{enumerate}[label = (\alph*)]\label{prop_nonan}
    \item In the first case, the behavior of $\phi_\theta$ is independent of this parameter $\overbar{\theta}$ for both $t < a$ and $t > a$, in which case $\overbar{\theta}$ cannot be used for parameter splitting. This is because, if $\phi_\theta$ is independent of $\overbar{\theta}$ on both sides of $t = a$, then (\ref{param_absorb_one}) can never hold.

    \item In the second case, where at least one of the $c_n$ depends on $\overbar{\theta}$, then the behavior of $\phi_\theta$ depends on $\overbar{\theta}$ for both $t < a$ and $t > a$ because the same convergent Taylor series expansion involving the same $c_n$ is used in both regions. We can see this directly by computing a quantity that depends on $\overbar{\theta}$ using only $\phi_\theta ( t )$ for $t > a$, and again using only data for $t < a$. For instance, one could compute any of the $c_n$ via the formula
    \begin{align}\label{cn_from_derivatives}
        c_n ( \theta ) = \frac{1}{n!} \frac{d^n}{d t^n} \phi_\theta ( t ) \Big\vert_{t = a} \, ,
    \end{align}
    and the derivatives (\ref{cn_from_derivatives}) could be computed using one-sided limits which only require knowing the behavior of $\phi_\theta ( t )$ for $t > a$ (or, equivalently, only knowing the function for $t < a$). This shows explicitly that, if $c_n$ depends on some $\overbar{\theta}$, then $\phi_\theta ( t )$ depends on $\overbar{\theta}$ on both sides of the $t = a$ point. Thus any parameter $\overbar{\theta}$ on which at least one of the $c_n$ depends also cannot be used for splitting mechanism.
\end{enumerate}
This argument applies to \emph{any} parameter, so we conclude that there do not exist any parameters $\theta^+$ and $\theta^-$ that can be used for the parameter splitting mechanism. This contradicts that $\phi_\theta$ achieves RP about $t = a$ by parameter splitting.
\end{proof}

Proposition \ref{prop_nonan} may seem disturbing if one is accustomed to thinking of neural networks with analytic architectures like the cosine in (\ref{cos_example}), or of conventional NNs such as a feedforward network $\phi_\theta$ with ReLU activation functions, since this $\phi_\theta$ is an analytic function of its arguments away from a set of measure zero corresponding to the kink discontinuities of the functions $\mathrm{ReLU} ( z )$ at $z = 0$. However, in the rigorous formulation of the Euclidean quantum mechanics path integral, one may also recall that the path integral measure is supported on paths which are differentiable nowhere. This construction proceeds by combining the na\"ive path integral measure $\mathcal{D} x$ with the exponentiated kinetic part of the action $e^{-S_{\text{kin}}}$, where $S_{\text{kin}} \sim \int dt \, \dot{x}^2$, in a particular discretization scheme. Taking the continuum limit of this discretization, one can prove that the combination $e^{-S_{\text{kin}}} \, \mathcal{D} x$ converges to the Wiener measure on continuous functions. Functions which are differentiable at even a single point form a set of measure zero with respect to the Wiener measure; the full-measure contribution comes from nowhere-differentiable paths, such as typical sample paths of Brownian motion. From this perspective, it is not so surprising that a reasonable NN-QM system might involve architectures which are analytic nowhere (or differentiable nowhere, which is a stronger condition).

Let us also point out that the inclusion of nowhere-differentiable paths in the standard quantum mechanics path integral is not merely a mathematical curiosity, but has physical implications. For instance, in Appendix \ref{app:commutators} we review the standard argument that non-trivial commutators such as $[ \hat{x} , \hat{p} ] = i \hbar$ are a direct consequence of the inclusion of non-differentiable trajectories in the path integral. In particular, if one studies a NN-QM system $x ( t ) = \phi_\theta ( t )$ whose architecture is a continuously differentiable function of $t$ at each point, then it would necessarily follow that $[ \hat{x} , \hat{\dot{x}} ] = 0$. Assuming that $\dot{x}$ is related to the canonical momentum, this contradicts the fundamental commutation relations that we expect for a quantum system. Therefore, if we would like our NN-QM to reproduce basic features of quantum mechanics such as commutation relations and the associated uncertainty principles, nowhere-differentiable architectures are essential.

\subsection{Markov Processes and Reflection Positivity}\label{sec:markov}

We now turn to a Markov process description of the quantum system, where a Markov process $y_t$ could arise as the ``output-space" description of a NN with $y_t=x_t$, or in the parameter space perspective if $y_t = \theta_t$. The latter case achieves Markov $x_{\theta_t}$ if the architecture is injective, but we emphasize the former output space viewpoint in this section for simplicity. In this perspective, $x_t$ defines joint probability distributions (\ref{stochastic_joint_pde}) over outputs, rather than the ``parameter-space'' description in terms of a NN architecture and probability density over parameters $\theta$. This has the advantage of being ``representation-independent'' in the sense that one need not choose a particular parameterization involving specific $\theta$. As a result, we can present a mechanism for reflection-positivity, somewhat analogous to parameter splitting, which is phrased directly as a sufficient condition on the process $x_t$ and without referring to its parameterization. Furthermore, we will be able to easily implement both RP and the OS axiom \ref{euclidean_invariance} of symmetry and stationarity at once.

To motivation this condition, suppose that we would like to restrict attention to stochastic processes with nowhere-differentiable sample paths, motivated by the importance of such paths for commutation relations which we mentioned at the end of the last subsection. Famous examples of such processes include Brownian motion, Ornstein-Uhlenbeck processes, and any diffusion process obeying a stochastic differential equation
\begin{align}\label{sde}
    d x_t = \mu ( x_t ) \, dt + \sigma ( x_t ) \, d W_t \, ,
\end{align}
where $W_t$ is a Wiener path. If $\sigma ( x_t ) \neq 0$, any solution to an equation of the form (\ref{sde}) is nowhere-differentiable almost surely. All of these processes share the Markov property, which is the statement that ``the future is independent of the past, given the present.'' This is clear given an expression (\ref{sde}) for the infinitesimal increments $d x_t$, which manifestly depend only on the present value of $x_t$ but not on its values at earlier times.

More precisely, a stochastic process $x_t$ is said to be \emph{Markov} if, given times $t_1 < t_2 < \ldots < t_n < t_{n+1}$, one has
\begin{align}\label{markov}
    \mathbb{P} \left( x ( t_{n+1} ) \mid x ( t_1 ) = x_1 \, , \, x ( t_2 ) = x_2 \, , \, \ldots \, , x ( t_n ) = x_n \right) = \mathbb{P} \left( x ( t_{n+1} ) \mid x ( t_n ) = x_n \right) \, ,
\end{align}
where a vertical bar indicates a conditional probability distribution. Said differently, conditioning on a collection of observations $x ( t_i ) = x_i$ of the stochastic process at several past points gives no additional information about the future evolution of the stochastic process than conditioning on the single, latest observation $x ( t_n ) = x_n$.

Intuitively, the Markov property plays a similar role as the parameter splitting mechanism in NN-QM. If a NN-QM $\phi_\theta ( t )$ enjoys parameter splitting with respect to some time $t_n$, this means that there are parameters $\theta^-$ which do not affect the outputs of the network for times $t > t_n$, although they are relevant for the model's behavior at earlier times $t < t_n$. These parameters $\theta^-$ are similar to the outputs $x ( t_1 )$, $\ldots$, $x ( t_{n-1} )$ in the Markov condition (\ref{markov}). Given the value of $x ( t_n )$, which is roughly analogous to the parameters $\theta^0$ in the splitting mechanism, future outputs like $x ( t_{n+1} )$ are independent of the behavior of the stochastic process for $t < t_n$, just as the future behavior of a NN with parameter splitting is independent of the parameters $\theta^-$.

Because parameter splitting implies reflection positivity, and the Markov property is conceptually similar to parameter splitting, one might ask whether the Markov property is related to RP. The answer is affirmative, and this connection has been known for some time. In fact, the condition of reflection positivity was originally introduced in \cite{Osterwalder:1973dx} as a weakening of the Markov condition used by Nelson \cite{NELSON197397}. A proof that the Markov property implies reflection positivity, in the general $d$-dimensional case, can be found in the chapter ``Probability Theory and Euclidean Field Theory'' by Nelson in the lecture notes \cite{Velo:1973tpu}, and a similar two-line argument (using slightly different language) appears in Proposition 1.5 of \cite{KLEIN1978277}. Hence it is natural to think of RP as a generalization of the Markov property \cite{generalization_of_markov}.

Although this result is well-understood, even for $d$-dimensional Euclidean field theories, we find it useful to give an explicit proof in the $d = 1$ case of quantum mechanics which we focus on in this work. By virtue of restricting to a one-dimensional spacetime, we are able to give an argument which is more elementary than the one for general $d$, and which makes the resemblance to the parameter splitting mechanism clear. For a proof of the analogous result for \emph{discrete} time stochastic processes, see Proposition 2.2.5 of \cite{ARVESON1986173}.

\begin{thm}\label{markov_thm}
    Every symmetric Markov process is reflection positive.
\end{thm}

\begin{proof}
Let $x_t$ be Markov and symmetric, and fix a collection of positive times $0 < t_1 < \ldots < t_n$. Consider the joint probability distribution
\begin{align}
    \mathbb{P} = \mathbb{P} \left( x ( - t_n ) , \ldots, x ( - t_1 ) , x ( 0 ) , x ( t_1 ) , \ldots , x ( t_n ) \right) \, ,
\end{align}
where we have added the time-zero value $x ( 0 )$ (which is not in the list of $t_i$ and $-t_i$). The quantity $\mathbb{P}$ can be written in conditional form as
\begin{align}
    \mathbb{P} &= \mathbb{P} \left( x ( t_1 ) , \ldots , x ( t_n ) \mid x ( - t_n ) , \ldots, x ( - t_1 ) , x ( 0 ) \right) \cdot \mathbb{P} \left( x ( - t_n ) , \ldots, x ( - t_1 ) , x ( 0 ) \right) \, \nonumber \\
    &= \mathbb{P} \left( x ( t_1 ) , \ldots , x ( t_n ) \mid x ( 0 ) \right) \cdot \mathbb{P} \left( x ( - t_n ) , \ldots, x ( - t_1 ) , x ( 0 ) \right) \, ,
\end{align}
where in the second line we have used the Markov property. Similarly, one can write
\begin{align}
    \mathbb{P} \left( x ( - t_n ) , \ldots, x ( - t_1 ) , x ( 0 ) \right) = \mathbb{P} \left( x ( - t_n ) , \ldots, x ( - t_1 ) \mid x ( 0 ) \right) \cdot \mathbb{P} ( x ( 0 ) ) \, ,
\end{align}
so that
\begin{align}\label{P_decomp}
    \mathbb{P} = \mathbb{P} \left( x ( t_1 ) , \ldots , x ( t_n ) \mid x ( 0 ) \right) \cdot \mathbb{P} \left( x ( - t_n ) , \ldots, x ( - t_1 ) \mid x ( 0 ) \right) \cdot \mathbb{P} ( x ( 0 ) ) \, .
\end{align}
Now given any appropriate function $F \left( x ( t_1 ) , \ldots , x ( t_n ) \right)$ that depends on the values of $x_t$ at positive times, one has
\begin{align}
    \langle F ( T F )^\ast \rangle &= \int \mathbb{P} \left( x_{-n} , \ldots, x_{-1} , x_0 , x_1 , \ldots , x_n \right) \, d x_{-n} \ldots d x_{-1} \, d x_0 \, dx_{1} \ldots d x_n \nonumber \\
    &\qquad \qquad \cdot F \left( x_1 , \ldots , x_n \right) \left( F \left( x_{-1} , \ldots , x_{-n} \right) \right)^\ast \, ,
\end{align}
where $T$ is the time-reversal operator and where we use the shorthand notation $x_i = x ( t_i )$, $x_{-i} = x ( - t_{i} )$, $d x_{-i} = d x ( - t_i )$, $d x_i = d x ( t_i )$. Using the decomposition (\ref{P_decomp}), this is
\begin{align}\label{rp_intermediate}
    \langle F ( T F )^\ast \rangle &= \int \mathbb{P} \left( x_1 , \ldots , x_n \mid x_0 \right) \cdot \mathbb{P} \left( x_{-n} , \ldots, x_{-1} \mid x_0 \right) \cdot \mathbb{P} ( x_0 ) \nonumber \\
    &\qquad \cdot \, d x_{-n} \ldots d x_{-1} \, d x_0 \, dx_{1} \ldots d x_n \cdot F \left( x_1 , \ldots , x_n \right) \left( F \left( x_{-1} , \ldots , x_{-n} \right) \right)^\ast \, \nonumber \\
    &= \int d x_0 \, \mathbb{P} ( x_0 ) \, \left( \int dx_1 \ldots d x_n \, \mathbb{P} \left( x_1 , \ldots , x_n \mid x_0 \right) F \left( x_1 , \ldots , x_n \right) \right) \nonumber \\
    &\qquad \cdot \left( \int dx_{-n} \ldots dx_{-1} \, \mathbb{P} \left( x_{-n} , \ldots, x_{-1} \mid x_0 \right) F \left( x_{-1} , \ldots , x_{-n} \right) \right)^\ast \, .
\end{align}
The integral in the parentheses of the first line in the final equation of (\ref{rp_intermediate}) computes a particular conditional expectation value of the function $F$ in the stochastic process $x_t$. Similarly, the integral in the final line of (\ref{rp_intermediate}) computes the \emph{same} conditional expectation value in the stochastic process $x_{-t}$. Since $x_t$ is symmetric, the processes $x_t$ and $x_{-t}$ are equivalent and these expectations are equal. We therefore conclude that
\begin{align}
    \langle F ( T F )^\ast \rangle = \int d x_0 \, \mathbb{P} ( x_0 ) \, \left| \int dx_1 \ldots d x_n \, \mathbb{P} \left( x_1 , \ldots , x_n \mid x_0 \right) F \left( x_1 , \ldots , x_n \right) \right|^2 \, ,
\end{align}
which is manifestly non-negative.
\end{proof}

Comparing, for instance, the penultimate lines of equations (\ref{param_splitting_final}) and (\ref{rp_intermediate}) highlights the technical similarities between the parameter splitting and Markov mechanisms for achieving reflection positivity. In both cases, one factorizes the expectation value $\langle F ( T F )^\ast \rangle$ into an integral of a product of two factors, one of which is another expectation value associated with only positive times and the other of which concerns only negative times. Using an appropriate symmetry -- arising from either the parameter absorption trick, in the splitting mechanism, or from the explicit assumption that $x_t$ and $x_{-t}$ are equivalent, in the Markov argument -- one then observes that these two factors are equal, which means that $\langle F ( T F )^\ast \rangle$ is the integral of a perfect square, and is hence positive.

However, a notable difference between the two mechanisms is the ease or difficulty of combining reflection positivity with translation invariance. We saw in section \ref{sec:param_split} that incorporating time translation symmetry into the parameter splitting mechanism seems tricky, and na\"ively appears to require introducing an uncountable infinity of parameters. However, from the ``output-space'' perspective, combining these two features is easy. Indeed, we immediately see that every symmetric, stationary Markov process is positive with respect to reflections about any time $t = a$, and thus satisfies both of the OS axioms \ref{euclidean_invariance} and \ref{reflection_positivity}. Morally, this is because the definition (\ref{markov}) of a Markov process makes no reference to any specific time $t = 0$. Rather, the ``independence of the future from the past, given the present'' occurs universally for any value of the present time.

\section{Deep Neural Network Quantum Mechanics}\label{sec:deep}

It is natural to seek general methods for constructing new quantum mechanical models out of old ones. In Section \ref{sec:markov}, we have seen that any symmetric stationary Markov process $x_t$ (which, if it is also a MQM, must admit a neural network representation $x_t = \phi_\theta ( t )$, by Theorem \ref{KKL_theorem}) defines a quantum system that satisfies the Osterwalder-Schrader axioms \ref{euclidean_invariance} and \ref{reflection_positivity}, which are necessary conditions for the existence of a well-defined real-time extension of the Euclidean theory. In this section, we will investigate the allowed transformations that one can perform on such a process $x_t$ which preserve the possibility of continuing the associated quantum system to real-time.

\subsection{Neural Networks Preserve Reflection Positivity}

First let us discuss the Markov property, which was used to achieve reflection positivity in Theorem \ref{markov_thm}. Unfortunately, the space of Markov processes is rather ``fragile'' in the sense that it is not closed under many operations. For instance, a linear combination of Markov processes need not be Markov. Intuitively, this is because knowing the sum of two numbers provides strictly less information than knowing each of the two summands individually. If $z_t = x_t + y_t$ where $x_t$, $y_t$ are Markov, then knowing the present values of $x_t$ and $y_t$ gives enough information to predict their future behavior, and hence the behavior of their sum $z_t$. But if one only knows the present value of $z_t$, it is impossible to recover the precise current values of $x_t$ and $y_t$ since many pairs of numbers have the same sum, and thus the present value of $z_t$ is not sufficient to predict its future behavior.

Similarly, the space of Markov processes is not closed under the operation of applying deterministic functions. As an illustrative example, if $W_t$ is a Wiener process or Brownian path, \cite{function_of_brownian} classified all functions $f : \mathbb{R} \to \mathbb{R}$ for which $x_t = f ( W_t )$ is Markov, and found only four families of possibilities. Again, this is because a general function is not injective, so knowing the value of $f ( W_t )$ is typically insufficient to recover the value of $W_t$.

However, it was pointed out in \cite{generalization_of_markov} that the space of RP processes behaves more nicely with respect to these operations. In particular, one has the following results.

\begin{prop}\label{lincomb}
    If $x_t$ and $y_t$ are two reflection positive processes, then any linear combination $z_t = a x_t + b y_t$, for $a, b \in \mathbb{R}$, is also reflection positive.
\end{prop}

\begin{prop}\label{detfun}
    If $x_t$ is reflection positive and $f : \mathbb{R} \to \mathbb{R}$ is a bounded measurable deterministic function, then the process $z_t = f ( x_t )$ is reflection positive.
\end{prop}

It is also straightforward to see that the operations of taking linear combinations and/or deterministic functions preserve the properties of symmetry and stationarity. These observations suggests a strategy for generating a large class of Euclidean quantum models that admit real-time continuations: begin with some collection of simple examples, such as symmetric stationary Markov processes, and apply a sequence of such operations to produce new systems which still obey the OS axioms \ref{euclidean_invariance} and \ref{reflection_positivity}.

One can generalize Proposition \ref{detfun} further by considering \emph{random} functions of RP processes, such as those defined by a neural network, which still preserve reflection positivity. This is the main result of this section, and is formalized in the following theorem.

\begin{thm}\label{deep_nnqm_theorem}
    Let $x_t$ be a reflection-positive stochastic process and consider a family of bounded measurable random functions realized by a neural network architecture $\phi_\theta$ with parameter density $P ( \theta )$. Then the stochastic process
    \begin{align}
        y_t = \phi_\theta ( x_t )
    \end{align}
    is also reflection-positive.
\end{thm}

Note that $y_t$ has two sources of randomness: to generate a sample path of $y_t$, we first generate an entire realization of the process $x_{t}$, \emph{and} we perform a random draw of the parameters $\theta$ to obtain an instance of the function $\phi_\theta$, then apply $\phi_\theta$ to $x_t$ pointwise.

\begin{proof}

Fix a collection of positive times $0 < t_1 < \ldots < t_n$. Following the notation in the proof of Theorem \ref{markov_thm}, we write $y_i = y ( t_i )$, $y_{-i} = y ( - t_i )$, and so on. Let $F ( y_1 , \ldots, y_n )$ be a bounded measurable function which depends only on the values of $y_i = y ( t_i )$ at these positive times. Then consider the expectation value
\begin{align}
    \langle F ( T F )^\ast \rangle_y &= \int d \theta \, P ( \theta ) \int d x_{-n} \ldots d x_{-1} \, dx_1 \, \ldots \, d x_n \, \mathbb{P} ( x_{-n} , \ldots, x_n ) \nonumber \\
    &\qquad \qquad \qquad \qquad F \left( \phi_\theta ( x_1 ) , \ldots, \phi_\theta ( x_n ) \right) \left( F \left( \phi_\theta ( x_{-1} ) , \ldots, \phi_\theta ( x_{-n} ) \right) \right)^\ast \, .
\end{align}
We use the notation $\langle \, \cdot \, \rangle_y$ to emphasize that this is an expectation value taken with respect to the process $y_t = \phi_\theta ( x_t )$. Now define the function $\widetilde{F} = F \circ \phi_\theta$ of $n$ variables by
\begin{align}
    \widetilde{F} ( x_1 , \ldots , x_n ) = F \left( \phi_\theta ( x_1 ) , \ldots, \phi_\theta ( x_n ) \right) \, ,
\end{align}
so that
\begin{align}
    \langle F ( T F )^\ast \rangle_y &= \int d \theta \, P ( \theta ) \int d x_{-n} \ldots d x_{-1} \, dx_1 \, \ldots \, d x_n \, \mathbb{P} ( x_{-n} , \ldots, x_n ) \nonumber \\
    &\qquad \qquad \qquad \qquad \widetilde{F} \left( x_1 , \ldots , x_n \right) \left( \widetilde{F} \left( x_{-1} , \ldots , x_{-n} \right) \right)^\ast \nonumber \\
    &= \int d \theta P ( \theta ) \, \langle \widetilde{F} ( T \widetilde{F} )^\ast \rangle_x \, .
\end{align}
We have introduced the symbol $\langle \, \cdot \, \rangle_x$ to refer to the expectation value taken with respect to the stochastic process $x_t$. For each fixed choice of the parameters $\theta$, the function $\widetilde{F}$ is the composition of two bounded measurable functions of $n$ variables, and thus satisfies the technical condition on the function appearing in the statement (\ref{rp_condition}) of reflection positivity. Since we have assumed the stochastic process $x_t$ is reflection positive, one has
\begin{align}
    \langle \widetilde{F} ( T \widetilde{F} )^\ast \rangle_x \geq 0 \, .
\end{align}
We therefore see that $\langle F ( T F )^\ast \rangle_y$ is the integral of a non-negative definite quantity, multiplied by a non-negative parameter density $P ( \theta )$, so we conclude that
\begin{align}
    \langle F ( T F )^\ast \rangle_y \geq 0 \, ,
\end{align}
and hence $y_t$ is reflection positive.
\end{proof}
Theorem \ref{deep_nnqm_theorem} motivates the following construction. Let $y_t^{(1)}$, $\ldots$, $y_t^{(n)}$ be a finite set of reflection-positive processes (for instance, symmetric Markov processes). If the $y_t^{(i)}$ are MQMs, we may realize each of them as neural networks by presenting them using the KKL decomposition. We then view these processes as inputs to a new neural network. Choose a width $D_1$ for the next layer and define a collection of $D_1$ processes $z^{(1, i)}_t$ by
\begin{align}
    z^{(1, i)}_t = \phi_\theta^{(1, i)} \left( \sum_{j=1}^{n} c_{j}^{(1, i)} y_t^{(j)} \right) \, , \qquad i = 1 , \ldots, D_1 \, ,
\end{align}
where $\phi_\theta^{(1, i)}$ are neural networks (or general random functions) and $c_{j}^{(1, i)}$ are real coefficients, which can be chosen either to be constants or random variables with some probability distribution. We then continue: choosing another width $D_2$, we define
\begin{align}
    z^{(2, i)} = \phi_\theta^{(2, i)} \left( \sum_{j=1}^{D_1} c_{j}^{(2, i)} z_t^{(j)} \right) \, , \qquad i = 1 , \ldots, D_2 \, ,
\end{align}
and so on, up to a collection of processes $z^{(N, i)}$, $i = 1 , \ldots, D_N$, for a positive integer $N$. Finally, we construct the output layer
\begin{align}
    x_t =  \phi_\theta^{(x)} \left( \sum_{j=1}^{D_N} c_{j}^{(N)} z_t^{(N, j)} \right) \, .
\end{align}
We refer to any quantum system $x_t$ obtained by a sequence of steps of this form as a \emph{deep neural network quantum mechanics} (deep NN-QM). Our theorem guarantees that reflection positivity is preserved at each layer of this network, so the resulting stochastic process $x_t$ is reflection positive. If all of the input processes $y_t^{(i)}$ are symmetric and stationary, so that they satisfy the OS axiom \ref{euclidean_invariance}, then $x_t$ enjoys this property as well.

Although it is not our primary focus here, let us remark that certain choices of deep NN-QM architectures also preserve the mixing property, in the sense that if each of the $y_t^{(i)}$ obeys the cluster decomposition axiom \ref{cluster_decomposition}, then $x_t$ also inherits this property. This is not, however, true in general: a counter-example is a network which simply shifts the input process by a random offset, $y_t = x_t + \theta$ with $\theta$ drawn randomly and separately for each sample path $x_t$. As is known in ergodic theory, such an ensemble of mixing processes with random shifts may fail to be mixing (this is in contrast with the application of a \emph{deterministic} measurable function, which always preserves the mixing property).

Therefore, given any collection of input process $y_t^{(i)}$ which satisfy all of the OS axioms, the deep NN-QM construction generically produces a new quantum system that also satisfies \ref{temperedness}, \ref{euclidean_invariance}, \ref{symmetric}, and \ref{reflection_positivity}, and for certain choices also \ref{cluster_decomposition}; in the latter case the deep NN-QM hence gives rise to a well-defined real-time quantum mechanical theory.

The quantum systems produced by this deep NN-QM procedure are non-trivial, even if the input processes $y_t^{(i)}$ are quite simple. For instance, suppose that all of the $y_t^{(i)}$ are Ornstein-Uhlenbeck processes, which are associated with quantum harmonic oscillators with some mass parameter. Then in addition to being mixing, symmetric, stationary, and Markov -- and consequently satisfying all of the OS axioms -- the systems $y_t^{(i)}$ are Gaussian, and are in some sense trivial since they are analogous to free quantum field theories. However, the output $x_t$ of a generic deep NN-QM built from these Markov, Gaussian input processes will be neither Markov nor Gaussian, although such an $x_t$ is symmetric, stationary, and RP. The non-Gaussianity of $x_t$ is clear since a non-linear transformation of a Gaussian process is no longer Gaussian, and the neural networks $\phi_\theta^{(i, j)}$ used in this construction are nonlinear in general. In particular, the higher-order cumulants of a process obtained from applying a generic deep NN-QM to Gaussian inputs will be non-vanishing, which signals that they are interacting quantum mechanical theories.

\subsection{Examples and Numerics}\label{sec:examples}

In this section, we will present some numerical results for quantum observables obtained by simulating stochastic processes on a computer. The simplest model which can be studied in this way is the quantum harmonic oscillator, which is associated with the statistics of the Ornstein-Uhlenbeck process and is discussed in Section \ref{sec:OU}. We then use OU processes as inputs for a deep NN-QM construction in Section \ref{sec:NN_of_OU}. As we have mentioned, since the Ornstein-Uhlenbeck process is symmetric, stationary, and Markov, any such deep NN-QM constructed from OU building blocks is symmetric, stationary, and reflection positive. We will choose examples of deep NN-QM architectures for which the output process is also mixing, which gives rise to families of Euclidean quantum models obeying all of the OS axioms and which can thus be continued to real time.

\subsubsection{Ornstein-Uhlenbeck Process}\label{sec:OU}

The Ornstein-Uhlenbeck process \cite{PhysRev.36.823} is a famous and well-studied stochastic process, which enjoys many desirable features that make it amenable to theoretical analysis, while still exhibiting sufficient richness as to be interesting. One can think of the OU process as a mean-reverting random walk, in the sense that the trajectory $x_t$ of such a process exhibits Brownian fluctuations in addition to a ``restoring force'' which acts to return the position to its equilibrium value $x_t = 0$. This two effects are modeled mathematically via the two terms on the right side of the stochastic differential equation
\begin{align}\label{OU_SDE}
    dx_t = - \theta x_t \, dt + \sigma \, d W_t \, ,
\end{align}
where $W_t$ is a Wiener path. Here $\theta$ measures the strength of the deterministic ``drift'' part of the SDE, which leads to the mean-reverting behavior, and the ``noise'' parameter $\sigma$ describes the size of the random fluctuations. This is an example of the general form (\ref{sde}) of the SDE which defines an It\^o diffusion. As we have remarked, the solution of any such SDE is a Markov process which, assuming symmetry, is RP.

An OU process with a fixed initial condition, say $x ( 0 ) = 0$, is not stationary. However, the Ornstein-Uhlenbeck process admits a stationary distribution, in the sense that if the initial condition $x_0 = x ( 0 )$ is drawn from a particular probability distribution $\mathbb{P}_{s} ( x_0 )$, then the resulting stochastic process is stationary. This distribution is given by
\begin{align}\label{ou_stationary}
    \mathbb{P}_{s} ( x_0 ) = \sqrt{ \frac{\theta}{\pi \sigma^2} } \exp \left( - \frac{ \theta x_0^2}{\sigma^2} \right) \, ,
\end{align}
which is a Gaussian distribution with mean zero and variance $\frac{\sigma^2}{2 \theta}$. When we speak of ``the'' Ornstein-Uhlenbeck process, we always mean the version of this process where the initial position of the process is drawn from the stationary distribution (\ref{ou_stationary}). In this case, the OU process is stationary, symmetric, and Markov, and it also turns out to be Gaussian.\footnote{By a theorem of Doob, the OU process is the \emph{unique} stationary Gaussian Markov process \cite{DoobJ.L.1942TBMa}.}

The statistics of the Ornstein-Uhlenbeck process precisely correspond to the Euclidean-time quantum harmonic oscillator, with a Hamiltonian that has been shifted by a constant so that the ground state energy is $E_0 = 0$ rather than $E_0 = \frac{1}{2} \hbar \omega$. As we will review in more detail around equation (\ref{transformed_hamiltonian}), this can be seen by performing a similarity transformation on the Fokker-Planck operator of the Ornstein-Uhlenbeck process \cite{Pavliotis2014}, bringing it to the form of an effective Hamiltonian
\begin{align}\label{ou_ham}
    H_{\text{OU}} = - \frac{\sigma^2}{2} \frac{d^2}{dx^2} + \frac{\theta^2}{2 \sigma^2} x^2 - \frac{\theta}{2} \, .
\end{align}
Comparing this to the Euclidean-time Hamiltonian of the quantum harmonic oscillator,
\begin{align}
    H_{\text{SHO}} = - \frac{\hbar^2}{2m} \frac{d^2}{dx^2} + \frac{1}{2} m \omega^2 x^2 \, ,
\end{align}
we see that the dictionary between the parameters is
\begin{align}\label{parameter_map}
    \theta = \omega \, , \qquad \sigma^2 = \frac{\hbar}{m} \, ,
\end{align}
and that the two Hamiltonians are related by
\begin{align}\label{ham_relation}
    H_{\text{OU}} = \frac{1}{\hbar} \left( H_{\text{SHO}} - \frac{1}{2} \hbar \omega \right) \, .
\end{align}
The overall factor of $\frac{1}{\hbar}$ is due to the fact that the quantum mechanical time evolution operator is $e^{- H t / \hbar}$, whereas the generator of time translations in the stochastic process has no factor of $\hbar$. We therefore see that the two Hamiltonians $H_{\text{OU}}$ and $H_{\text{SHO}}$ agree, up to this overall scaling and the shift by $\frac{\theta}{2} = \frac{\hbar \omega}{2}$ which we mentioned.

\subsubsection*{\ul{\it Two-point function}}

One of the simplest quantities that one can compute numerically for the OU process is the two-point function $G^{(2)} ( t, s )$. The well-known closed-form expression for this correlator is
\begin{align}\label{OU_two_point}
    G^{(2)} ( t, s ) = \langle x ( t ) x ( s ) \rangle = \frac{\sigma^2}{2 \theta} \exp \left( - \theta | t - s | \right) \, .
\end{align}
Because the Ornstein-Uhlenbeck process is Gaussian, all of the higher correlation functions $G^{(n)}$, $n \geq 3$, are completely fixed by the two-point function. The fact that the two-point function decomposes into a product of one-point functions at large separation,
\begin{align}
    \lim_{| t - s | \to \infty} \langle x ( t ) x ( s ) \rangle = \langle x ( t ) \rangle \langle x ( s ) \rangle = 0 \, ,
\end{align}
is an implication of the mixing property of the OU process, which is physically interpreted as the statement that the quantum harmonic oscillator enjoys cluster decomposition.

As a first test of our numerics, we compare the simulated two-point function for the Ornstein-Uhlenbeck process to the theoretical expression (\ref{OU_two_point}) in Figure \ref{fig:ou_correlation}, finding agreement between the numerical estimate and the theoretical prediction.

\begin{figure}[htbp]
    \includegraphics[width=\linewidth]{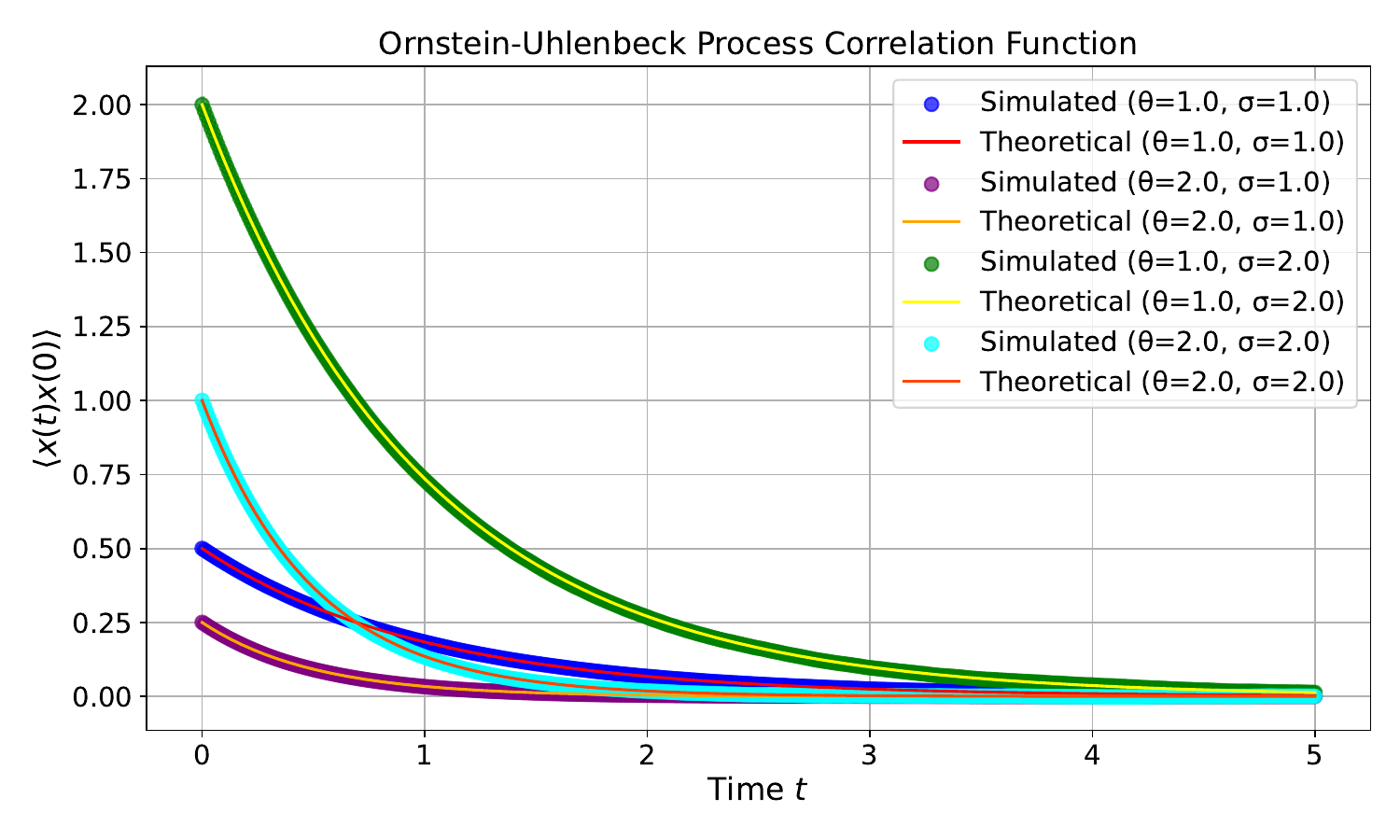}
    \caption{The correlation function $G^{(2)} ( t , 0 )$ obtained by simulating Ornstein-Uhlenbeck processes with $\theta, \sigma \in \{ 1, 2 \}$ for a total time $T = 5$ and step size $\Delta t = 0.01$. The initial position is drawn from the stationary distribution (\ref{ou_stationary}) and updates are performed using the Euler-Maruyama method. We carry out 200 epochs of simulations with 10,000 sample paths per epoch, compute the average correlator $G^{(2)} ( t , 0 )$ for each time $t$, and use the standard deviation across epochs for the error bars. The experimental results are shown as scatter plots with cool colors (green, cyan, purple, blue), while the corresponding analytical curves (\ref{OU_two_point}) for the two-point functions are drawn above the data points in warm colors (red, orange, yellow, and orange-red) and lie within all of the error bars. The marker size for the simulated results has been enlarged to enhance visibility; the error bars are plotted but are not visible because they are smaller than the corresponding markers.}
    \label{fig:ou_correlation}
\end{figure}

In view of the spectral representation (\ref{spectral_two_point}) for the two-point function and the fact that the matrix element $\langle 0 \mid x \mid n \rangle$ vanishes for the harmonic oscillator unless $n = 1$, we can immediately read off the gap between the ground state energy $E_0^{\text{OU}}$ and the first excited state energy $E_1^{\text{OU}}$ from the rate of the exponential decay of this two-point function. In this case, we find $E_1^{\text{OU}} - E_0^{\text{OU}} = \theta$. Using the dictionary (\ref{parameter_map}) and the relation $E_n^{\text{OU}} = \frac{1}{\hbar} E_n^{\text{SHO}}$, which follows from (\ref{ham_relation}), this gives $E_1^{\text{SHO}} - E_0^{\text{SHO}} = \hbar \omega$, which agrees with the gap to the first excited state in the harmonic oscillator. We will review a different method for extracting spectral data about the quantum oscillator from the OU process below.

\subsubsection*{\ul{\it Commutator}}

We have mentioned, in Section \ref{sec:param_split} and Appendix \ref{app:commutators}, that the existence of nowhere-differentiable paths is crucial for non-trivial commutation relations. We would now like to see how this works for the Ornstein-Uhlenbeck realization of the harmonic oscillator.

Since the Ornstein-Uhlenbeck process satisfies a stochastic differential equation with non-zero noise term, in the continuum limit the sample paths of this process will be differentiable nowhere almost surely. This means that, in a numerical simulation of OU paths using the Euler-Maruyama method with step size $\Delta t$, the estimated derivative
\begin{align}
    \frac{\Delta x}{\Delta t} = \frac{x_{t + \Delta t} - x_t}{\Delta t}
\end{align}
diverges as the discretization becomes finer and finer ($\Delta t \to 0$). However, the expression
\begin{align}\label{Ct_defn}
    C \left( t \right) = x_t \frac{x_t - x_{t-\Delta t}}{\Delta t} - x_t \frac{x_{t + \Delta t} - x_t}{\Delta t} \, ,
\end{align}
which estimates the commutator $[ \hat{x} ( t ) , \hat{p} ( t ) ]$, should remain \emph{finite} as the step size is taken small (as we have argued in Appendix \ref{app:commutators} for the case of the free quantum particle). Since $[\hat{x} ( t ) , \hat{p} ( t ) ] = \hbar$ for quantum mechanics in Euclidean signature, we expect $C(t)$ to be constant in time and proportional to $\hbar$, which scales as $\sigma^2$ according to the dictionary between OU and SHO parameters given in equation (\ref{parameter_map}). In our numerical simulations we find that this expectation is indeed borne out; the results are presented in Figure \ref{fig:ou_commutator}.

\begin{figure}[htbp]
    \includegraphics[width=\linewidth]{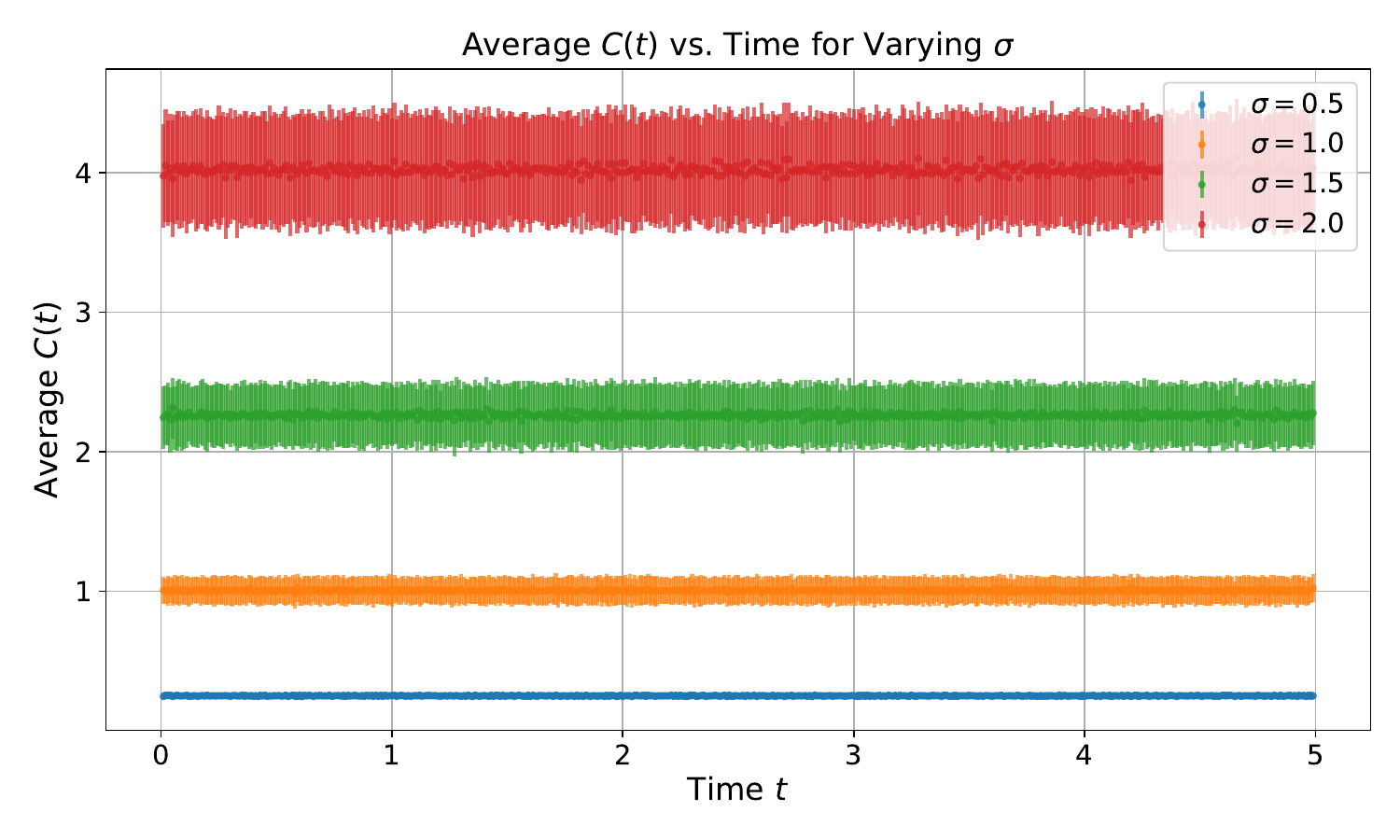}
    \caption{The values of the commutator $C(t)$, defined in (\ref{Ct_defn}), for Ornstein-Uhlenbeck processes with $\theta = 1$ and varying values of $\sigma$. We generate $10,000$ sample paths per epoch for 200 epochs, using the Euler-Maruyama method with step size $\Delta t = 0.01$ and total time $T = 5$, and plot the average value $\langle C ( t ) \rangle$ with error bars given by the standard deviation across epochs. The value of the commutator is constant in time and scales as the square of $\sigma$, as expected from the relation $\hbar = \sigma^2 m$ of equation (\ref{parameter_map}). This reproduces the commutator relation $[ \hat{x} , \hat{p} ] = \hbar$ for a quantum model in Euclidean signature. Although the displayed error bars appear bigger for the commutators with larger values of $\sigma$, the ratio of the standard deviation to the mean is approximately constant across experiments.}
    \label{fig:ou_commutator}
\end{figure}

\subsubsection*{\ul{\it Uncertainty relations}}

One of the most basic features of quantum mechanics is the Heisenberg uncertainty principle (HUP). In terms of the uncertainty $\Delta A$ associated with an observable $\hat{A}$,
\begin{align}
    \left( \Delta A \right)^2 = \langle \hat{A}^2 \rangle - \langle \hat{A} \rangle^2 \, ,
\end{align}
the HUP bounds the product of $\Delta A$ and the uncertainty $\Delta B$ in a second observable $\hat{B}$:
\begin{align}
    \Delta A \, \Delta B \geq \frac{1}{2} \left| \big\langle \big[ \hat{A} , \hat{B} \big] \big\rangle \right| \, . 
\end{align}
As we have reviewed, in Euclidean signature the position and momentum operators obey $[ \hat{x} , \hat{p} ] = \hbar$. Therefore, one obtains the general lower bound
\begin{align}
    \Delta x \, \Delta p \geq \frac{\hbar}{2} \, .
\end{align}
In the ground state of the harmonic oscillator, which corresponds to the stationary distribution of the OU process, this inequality is saturated and one finds $\Delta x \, \Delta p = \frac{\hbar}{2}$. 

We would like to verify the saturation of this uncertainty relation numerically using simulations of the Ornstein-Uhlenbeck process. However, to do this, we must first identify a suitable definition of expectation values of powers of the momentum operator $\hat{p}$. When $m = 1$, we expect that $p = \dot{x}$, and we have already pointed out that the continuum Ornstein-Uhlenbeck process has sample paths which are differentiable nowhere. We therefore cannot use the na\"ive definition of the derivative $\dot{x}$ for computing momentum expectation values, as these derivatives will diverge as the step size $\Delta t$ is taken to zero.

There are two ways to sidestep this issue, which turn out to be equivalent for the harmonic oscillator. One way is to represent the momentum operator in terms of the derivative $\frac{\partial}{\partial x}$ acting on wavefunctions, as we do in ordinary quantum mechanics, and compute expectation values of this derivative using the stationary distribution of the Ornstein-Uhlenbeck process. The second way, which we follow here, is motivated by the observation that the OU two-point function (\ref{OU_two_point}) is non-differentiable at the point $t = s$, but admits a one-sided derivative. That is, if we restrict to $t > s$, one has
\begin{align}
    \langle x ( t ) x ( s ) \rangle = \frac{\sigma^2}{2 \theta} \exp \left( - \theta ( t - s ) \right) \, .
\end{align}
One may then take derivatives of this expectation value to define momentum correlators. Following this strategy, we define the second moment of the momentum operator as
\begin{align}\label{psq_defn}
    \langle p ( t )^2 \rangle &= \lim_{\substack{h \to 0^+ \\ k \to 0^+}} \left( \frac{d}{d h} \frac{d}{d k} \langle x ( t + h ) x ( t - k ) \rangle \right) \, .
\end{align}
The notation $h \to 0^+$, $k \to 0^+$ mean that the variables $h, k$ are taken to zero from above (i.e. through strictly positive values). This definition ensures that the ordering
\begin{align}
    t - k < t < t + h
\end{align}
is maintained throughout the limiting procedure, so that we avoid the non-differentiability of the two-point function at coincident points. In numerical simulations, we use a finite-difference approximation to the quantity (\ref{psq_defn}):
\begin{align}
    \langle p ( t )^2 \rangle = \left\langle \frac{ x ( t + \Delta t ) x ( t - \Delta t ) - x ( t + \Delta t ) x ( t ) - x ( t ) x ( t - \Delta t ) + x(t)^2}{(\Delta t )^2} \right\rangle \, .
\end{align}
\begin{table}[htbp]
   \noindent\makebox[\linewidth][l]{\hspace*{-1.2cm}
  \begin{tabular}{ccccccccc}
    \toprule
    $\theta$ & $\sigma$ & $dt$ & \shortstack{$(\Delta x)^2$ \\ (sim)} & \shortstack{$(\Delta x)^2$ \\ (th)} & \shortstack{$(\Delta p)^2$ \\ (sim)} & \shortstack{$(\Delta p)^2$ \\ (th)} & \shortstack{$(\Delta x)^2\cdot(\Delta p)^2$ \\ (sim)} & \shortstack{$(\Delta x)^2\cdot(\Delta p)^2$ \\ (th)} \\
    \midrule
    $\frac{1}{2}$ & $\frac{1}{2}$ & 0.100 & $0.2564 \pm 0.0005$ & 0.2500 & $0.0641 \pm 0.0008$ & 0.0625 & $0.0164 \pm 0.0002$ & 0.0156 \\
    $\frac{1}{2}$ & $\frac{1}{2}$ & 0.010 & $0.2509 \pm 0.0015$ & 0.2500 & $0.0633 \pm 0.0077$ & 0.0625 & $0.0159 \pm 0.0019$ & 0.0156 \\
    $\frac{1}{2}$ & $1$         & 0.100 & $1.0256 \pm 0.0018$ & 1.0000 & $0.2560 \pm 0.0034$ & 0.2500 & $0.2626 \pm 0.0034$ & 0.2500 \\
    $\frac{1}{2}$ & $1$         & 0.010 & $1.0023 \pm 0.0064$ & 1.0000 & $0.2503 \pm 0.0316$ & 0.2500 & $0.2509 \pm 0.0316$ & 0.2500 \\
    $\frac{1}{2}$ & $2$         & 0.100 & $4.1004 \pm 0.0077$ & 4.0000 & $1.0270 \pm 0.0125$ & 1.0000 & $4.2109 \pm 0.0513$ & 4.0000 \\
    $\frac{1}{2}$ & $2$         & 0.010 & $4.0080 \pm 0.0261$ & 4.0000 & $1.0099 \pm 0.1334$ & 1.0000 & $4.0475 \pm 0.5330$ & 4.0000 \\
    $1$         & $\frac{1}{2}$ & 0.100 & $0.1316 \pm 0.0002$ & 0.1250 & $0.1315 \pm 0.0008$ & 0.1250 & $0.0173 \pm 0.0001$ & 0.0156 \\
    $1$         & $\frac{1}{2}$ & 0.010 & $0.1257 \pm 0.0005$ & 0.1250 & $0.1252 \pm 0.0075$ & 0.1250 & $0.0157 \pm 0.0009$ & 0.0156 \\
    $1$         & $1$         & 0.100 & $0.5262 \pm 0.0007$ & 0.5000 & $0.5263 \pm 0.0033$ & 0.5000 & $0.2770 \pm 0.0017$ & 0.2500 \\
    $1$         & $1$         & 0.010 & $0.5023 \pm 0.0023$ & 0.5000 & $0.5073 \pm 0.0307$ & 0.5000 & $0.2548 \pm 0.0155$ & 0.2500 \\
    $1$         & $2$         & 0.100 & $2.1044 \pm 0.0029$ & 2.0000 & $2.1072 \pm 0.0137$ & 2.0000 & $4.4343 \pm 0.0277$ & 4.0000 \\
    $1$         & $2$         & 0.010 & $2.0097 \pm 0.0090$ & 2.0000 & $2.0036 \pm 0.1307$ & 2.0000 & $4.0269 \pm 0.2646$ & 4.0000 \\
    $2$         & $\frac{1}{2}$ & 0.100 & $0.0694 \pm 0.0001$ & 0.0625 & $0.2778 \pm 0.0008$ & 0.2500 & $0.0193 \pm 0.0001$ & 0.0156 \\
    $2$         & $\frac{1}{2}$ & 0.010 & $0.0631 \pm 0.0002$ & 0.0625 & $0.2525 \pm 0.0080$ & 0.2500 & $0.0159 \pm 0.0005$ & 0.0156 \\
    $2$         & $1$         & 0.100 & $0.2777 \pm 0.0003$ & 0.2500 & $1.1114 \pm 0.0032$ & 1.0000 & $0.3087 \pm 0.0009$ & 0.2500 \\
    $2$         & $1$         & 0.010 & $0.2524 \pm 0.0009$ & 0.2500 & $1.0154 \pm 0.0317$ & 1.0000 & $0.2563 \pm 0.0079$ & 0.2500 \\
    $2$         & $2$         & 0.100 & $1.1109 \pm 0.0011$ & 1.0000 & $4.4458 \pm 0.0143$ & 4.0000 & $4.9390 \pm 0.0155$ & 4.0000 \\
    $2$         & $2$         & 0.010 & $1.0098 \pm 0.0032$ & 1.0000 & $4.0375 \pm 0.1261$ & 4.0000 & $4.0772 \pm 0.1263$ & 4.0000 \\
    \bottomrule
  \end{tabular}
  }
    \caption{We compare the simulated (``sim'') values of $\left( \Delta x \right)^2$, $\left( \Delta p \right)^2$ (defined using (\ref{psq_defn}) and (\ref{pdefn})), and their product, to the predicted theoretical (``th'') values $\left( \Delta x \right)^2 = \frac{\sigma^2}{2 \theta}$,  $\left( \Delta p \right)^2 = \frac{\sigma^2 \theta}{2}$, and $(\Delta x)^2\cdot(\Delta p)^2 = \frac{\sigma^4}{4}$. Simulations are performed using $200$ epochs with $10,000$ OU sample path per epoch, where each path consists of $1,000$ time steps. We average across time steps and epochs, and the errors indicated by $\pm$ are the standard deviation across epochs. We contrast two different sizes of the time step used in the Euler-Maruyama method, $dt = 0.1$ and $dt = 0.01$, and we see that the simulation results become closer to the theoretical values for the smaller choice of $dt$ due to reduced discretization error. The results confirm the relation $\left( \Delta x \right)^2 \left( \Delta p \right)^2 = \frac{\hbar^2}{4}$ for the ground state of the harmonic oscillator, where the effective $\hbar$ of the process is determined by $\sigma^2$ as in (\ref{parameter_map}).}
    \label{tab:hup}
\end{table}
Similarly, we define
\begin{align}\label{pdefn}
    \langle p ( t ) \rangle = \frac{d}{dt} \langle x ( t ) \rangle \, ,
\end{align}
and as $\langle x ( t ) \rangle = 0$ this quantity vanishes for the true continuum OU process.
% ; in our numerical implementation it is zero to within the error induced by discretization.

The results of the numerical computation of $\left( \Delta x \right)^2$ and $\left( \Delta p \right)^2$, for various values of $\theta$, $\sigma$, and the discretization time $\Delta t$, are presented in Table \ref{tab:hup}, which confirms that the uncertainty relation is saturated for the harmonic oscillator ground state within the numerical error induced by our simulation scheme.

\subsubsection*{\ul{\it Spectrum and eigenfunctions}}

Finally, we turn to the numerical computation of the energy eigenvalues and eigenfunctions of the harmonic oscillator. Because the spectrum and energy eigenstates are properties of the Hamiltonian operator, which is the generator of infinitesimal time translations, one might expect that this data can be extracted from the update rule (\ref{OU_SDE}) that describes how OU sample paths evolve over a small time step. Indeed this is the case -- by a standard argument, which we will briefly review, the generator of a stochastic process which admits a stationary distribution can be transformed into a Schr\"odinger-type operator via a certain similarity transformation, which can then be diagonalized to obtain the energies. See, for instance, section 4.9 of \cite{Pavliotis2014} for a more thorough discussion of this point.

We first mention that, rather than describing the time evolution of a specific sample path $x_t$ via the SDE (\ref{OU_SDE}), we can equivalently specify a differential equation that models the dynamics of the probability density $\mathbb{P} ( x, t )$ of an ensemble of OU paths at a time $t$. This equation takes the form
\begin{align}\label{fp_pde}
    \frac{\partial \mathbb{P}}{\partial t} &= \theta \frac{\partial}{\partial x} \left( x \mathbb{P} \right) + \frac{\sigma^2}{2} \frac{\partial^2 \mathbb{P}}{\partial x^2} \nonumber \\
    &= \mathcal{L} P \, ,
\end{align}
which is called the Fokker-Planck equation, and
\begin{align}\label{fokker_planck_defn}
    \mathcal{L} f = \frac{\sigma^2}{2} \frac{\partial^2 f}{\partial x^2} + \theta \frac{\partial}{\partial x} \left( x f \right)
\end{align}
is called the generator of the process, or the Fokker-Planck operator. This operator $\mathcal{L}$ generates infinitesimal time translations of the \emph{probability density} $\mathbb{P} ( x, t )$, but this is not the same as the operator that time-evolves the \emph{wavefunction} $\psi ( x, t )$ in conventional quantum mechanics, since $\mathbb{P} ( x, t ) = \left| \psi ( x, t ) \right|^2$. One way to see this is to note that the Fokker-Planck operator is not self-adjoint with respect to the inner product on functions:
\begin{align}\label{not_self_adjoint}
    \langle f , \mathcal{L} g \rangle &= \int dx \, f ( x ) \left( \frac{\sigma^2}{2} \frac{\partial^2 g}{\partial x^2} + \theta g + \theta x \frac{\partial g}{\partial x} \right) \nonumber \\
    &\neq \int dx \, \left( \frac{\sigma^2}{2} \frac{\partial^2 f}{\partial x^2} + \theta f + \theta x \frac{\partial f}{\partial x} \right) g(x) \nonumber \\
    &= \langle \mathcal{L} f , g \rangle \, .
\end{align}
On the other hand, we expect that the Hamiltonian operator should be self-adjoint when acting on wavefunctions. The failure of self-adjointness is due to the final term $\theta x \frac{\partial g}{\partial x}$ in the first line of (\ref{not_self_adjoint}); performing an integration by parts on this term does not yield the expression on the second line. The standard technique for removing such a first derivative term is through the introduction of an integrating factor. Given a differential operator
\begin{align}
    D F = a F''(x) + p ( x ) F' ( x ) + q ( x ) F ( x ) \, ,
\end{align}
for a constant $a$ and functions $p(x)$ and $q(x)$, one can change variables to a new function
\begin{align}
    G ( x ) = \exp \left( \frac{1}{2 a} \int_0^{x} p ( y ) \, dy \right) F ( x ) \, ,
\end{align}
in terms of which the differential operator takes the form
\begin{align}
    D F = \exp \left( - \frac{1}{2 a} \int_0^{x} p ( y ) \, dy \right)\left( a G''(x) + \left( q ( x ) - \frac{1}{2} p'(x) - \frac{1}{4a } p ( x )^2 \right) G ( x ) \right) \, ,
\end{align}
which has no term proportional to $G'$.

In the case of the operator $D = \mathcal{L}$ defined in equation (\ref{fokker_planck_defn}), taking $a = \frac{\sigma^2}{2}$ and $p(x) = \theta x$, this suggests the change of variables
\begin{align}\label{P_to_psi}
    \psi ( x, t ) =  \exp \left( \frac{\theta}{2 \sigma^2} x^2 \right) \mathbb{P} ( x, t )  \, ,
\end{align}
so that the Fokker-Planck operator $\mathcal{L}$ acts as
\begin{align}\label{L_psi_rewritten}
    \mathcal{L} \mathbb{P} = \exp \left( - \frac{\theta}{2 \sigma^2} x^2 \right) \left( \frac{\sigma^2}{2} \frac{\partial^2}{\partial x^2} - \frac{\theta^2}{2 \sigma^2} x^2 + \frac{\theta}{2}  \right) \psi ( x, t) \, ,
\end{align}
and in terms of which the differential equation (\ref{fp_pde}) can be written as
\begin{align}
    \frac{\partial \psi}{\partial t} = \left( \frac{\sigma^2}{2} \frac{\partial^2}{\partial x^2} - \frac{\theta^2}{2 \sigma^2} x^2 + \frac{\theta}{2} \right) \psi \, ,
\end{align}
which we identify as the Euclidean Schr\"odinger equation $\partial_t \psi = - H_{\text{OU}} \psi$ where
\begin{align}\label{transformed_hamiltonian}
    H_{\text{OU}} = - \frac{\sigma^2}{2} \frac{d^2}{dx^2} + \frac{\theta^2}{2 \sigma^2} x^2 - \frac{\theta}{2} \, ,
\end{align}
in agreement with the result for $H_{\text{OU}}$ we quoted earlier in (\ref{ou_ham}). The operator (\ref{transformed_hamiltonian}) is now self-adjoint with respect to the inner product on twice-differentiable square-normalizable functions, as is appropriate for a Hamiltonian operator in quantum mechanics.

Combining equations (\ref{P_to_psi}), (\ref{L_psi_rewritten}), and (\ref{transformed_hamiltonian}) gives the relation
\begin{align}\label{H_to_L}
    \mathcal{L} = - \exp \left( - \frac{\theta}{2 \sigma^2} x^2 \right) H_{\text{OU}} \exp \left( \frac{\theta}{2 \sigma^2} x^2 \right) \, .
\end{align}
In terms of the stationary distribution (\ref{ou_stationary}) for the Ornstein-Uhlenbeck process, equation (\ref{H_to_L}) can be expressed as
\begin{align}\label{L_to_H}
    H_{\text{OU}} = - \mathbb{P}_s ( x )^{-1/2} \, \mathcal{L} \, \mathbb{P}^{1/2}_s \,,
\end{align}
Thus, by performing the similarity transformation (\ref{L_to_H}), we can convert from the (non-self-adjoint) Fokker-Planck operator $\mathcal{L}$ that generates time evolution of the probability distribution $\mathbb{P}(x, t)$ to the self-adjoint operator $H_{\text{OU}}$ which generates time evolution of the wavefunction $\psi(x, t)$. Although we have reviewed the steps of this procedure for the Ornstein-Uhlenbeck process, the same argument applies to more general stochastic processes which admit a stationary distribution, $H = - \mathbb{P}_s ( x )^{-1/2} \, \mathcal{L} \, \mathbb{P}^{1/2}_s$.

These observations suggest a numerical algorithm that one can perform to compute the energies and eigenstates associated with any stochastic process that admits a stationary distribution. First, perform many draws from this stationary distribution to obtain an ensemble of initial positions $x_0$. Then use the update rule of the stochastic process to time-evolve each initial position by a small time step $\Delta t$, which gives a collection of one-step sample paths $(x_0, x_{\Delta t})$. Bin these $(x_0, x_{\Delta t})$ pairs to create a histogram, which represents a discrete approximation to the joint probability distribution $\mathbb{P} ( ( x_{\Delta t}, \Delta t ) , ( x_0, 0 ) )$. We then normalize this histogram so that the sum of the entries in each row is unity; this yields a discretized estimate of the \emph{conditional} probability distribution
\begin{align}
    \mathbb{P} ( x_{\Delta t} , \Delta t \mid x_0 , 0 ) = \frac{\mathbb{P} ( ( x_{\Delta t}, \Delta t ) , ( x_0, 0 ) )}{\mathbb{P} ( x_0, 0 ) } \, ,
\end{align}
where the denominator $\mathbb{P} ( x_0, 0 )$ is the marginal distribution for the position $x_0$ at time $0$. On the one hand, this conditional distribution obeys
\begin{align}\label{conditional_integral}
    \mathbb{P} ( x_{\Delta t} , \Delta t )  = \int d x_0 \, \mathbb{P} ( x_{\Delta t} , \Delta t \mid x_0 , 0 ) \cdot \mathbb{P} ( x_0 , 0 ) \, ,
\end{align}
but on the other hand, as $\mathcal{L}$ generates infinitesimal time evolution of $\mathbb{P} ( x , t ) $, we have
\begin{align}
    \mathbb{P} ( x_{\Delta t} , \Delta t ) \approx \mathbb{P} ( x_0 , 0 ) + \Delta t \, \mathcal{L} \, \mathbb{P} ( x_0 , 0 ) \, .
\end{align}
We now compare these two expressions to obtain an estimate for $\mathcal{L}$. In our discrete approximation, the integral in equation (\ref{conditional_integral}) becomes a matrix-vector product; the $(i, j)$ entry of the matrix $\mathbb{P} ( x_{\Delta  t} , \Delta t \mid x_0 , 0 )$ represents the conditional probability that $x_{\Delta t}$ lies in bin $i$ given that $x_0$ lies in bin $j$. Therefore, in matrix notation, we have
\begin{align}
    \mathbb{P} ( x_{\Delta t} , \Delta t )_i = \sum_j \mathbb{P} ( x_{\Delta t} , \Delta t \mid x_0 , 0 )_{ij} \mathbb{P} ( x_0 , 0 )_j \approx \mathbb{P} ( x_0 , 0 )_i + \Delta t \, \sum_j \mathcal{L}_{ij} \, \mathbb{P} ( x_0 , 0 )_j \, .
\end{align}
Solving for the matrix discretization of the Fokker-Planck operator $\mathcal{L}$ yields
\begin{align}
    \mathcal{L}_{ij} \approx \frac{\mathbb{P} ( x_{\Delta t} , \Delta t \mid x_0 , 0 )_{ij} - \delta_{ij}}{\Delta t} \, ,
\end{align}
where the Kronecker delta $\delta_{ij}$ gives the entries in the identity matrix. This formula allows us to numerically compute the approximate Fokker-Planck operator. Next one performs the similarity transformation (\ref{L_to_H}) using the stationary distribution; if $\mathbb{P}_s$ is not known in closed form, it can be estimated by simulating many sample paths with arbitrarily chosen initial positions, evolving each for a long time, and recording the steady-state distribution to which the ensemble settles down in the long run. Finally, one numerically diagonalizes the matrix $H$ to obtain approximate energy eigenvalues and eigenfunctions.

The energy eigenvalues computed using this procedure are displayed in Figure \ref{fig:ou_energy_eigenvalues}. Because the true spectrum of $H_{\text{OU}}$ is $E_n = \theta n$, which is independent of $\sigma$, we only show results for different values of $\theta$ and fix $\sigma = 1$ across all experiments. The numerically estimated wavefunctions of the first three energy eigenstates are shown in Figure \ref{fig:ou_energy_eigenstates} and compared to the true eigenstates.

\begin{figure}[htbp]
    \includegraphics[width=\linewidth]{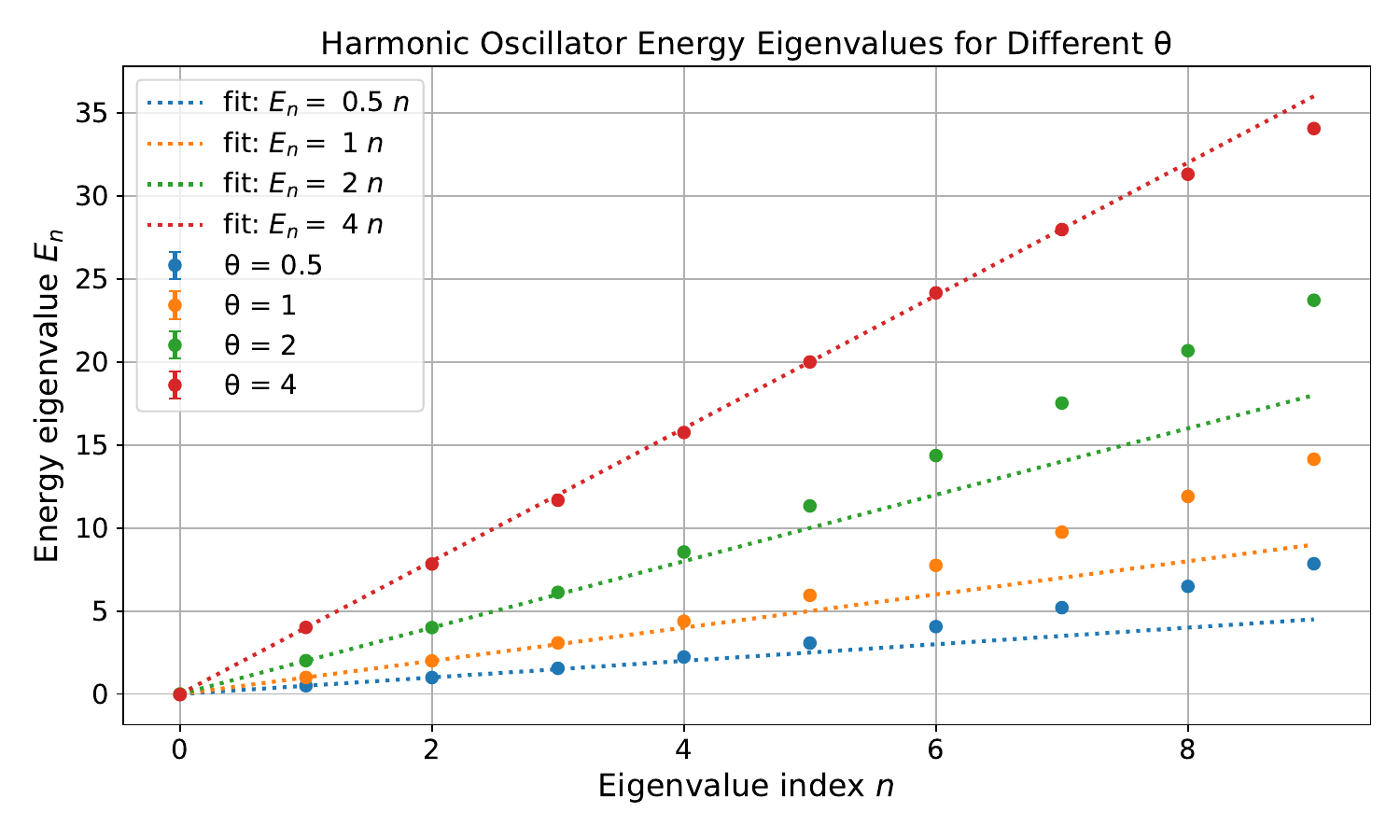}
    \caption{We display the results of computing the first $10$ energy eigenvalues of the harmonic oscillator for various values of $\theta$, using the algorithm described above. For each choice of $\theta$, we carry out $20$ separate experiments, each of which generates $5,000,000$ sample paths to estimate the Hamiltonian. Error bars are included, computed using the standard deviation across these $20$ experiments, but are too small to be seen. Dotted lines indicate the true energy eigenvalues, $E_n = \theta n$, where we recall that the ground state energy has been shifted to zero. The experimental values are more accurate for low-lying energies at small $n$ and then accrue larger errors for higher $E_n$, although the magnitude and sign of the errors differs between the different choices of $\theta$.}
    \label{fig:ou_energy_eigenvalues}
\end{figure}

\begin{figure}[htbp]
    \includegraphics[width=\linewidth]{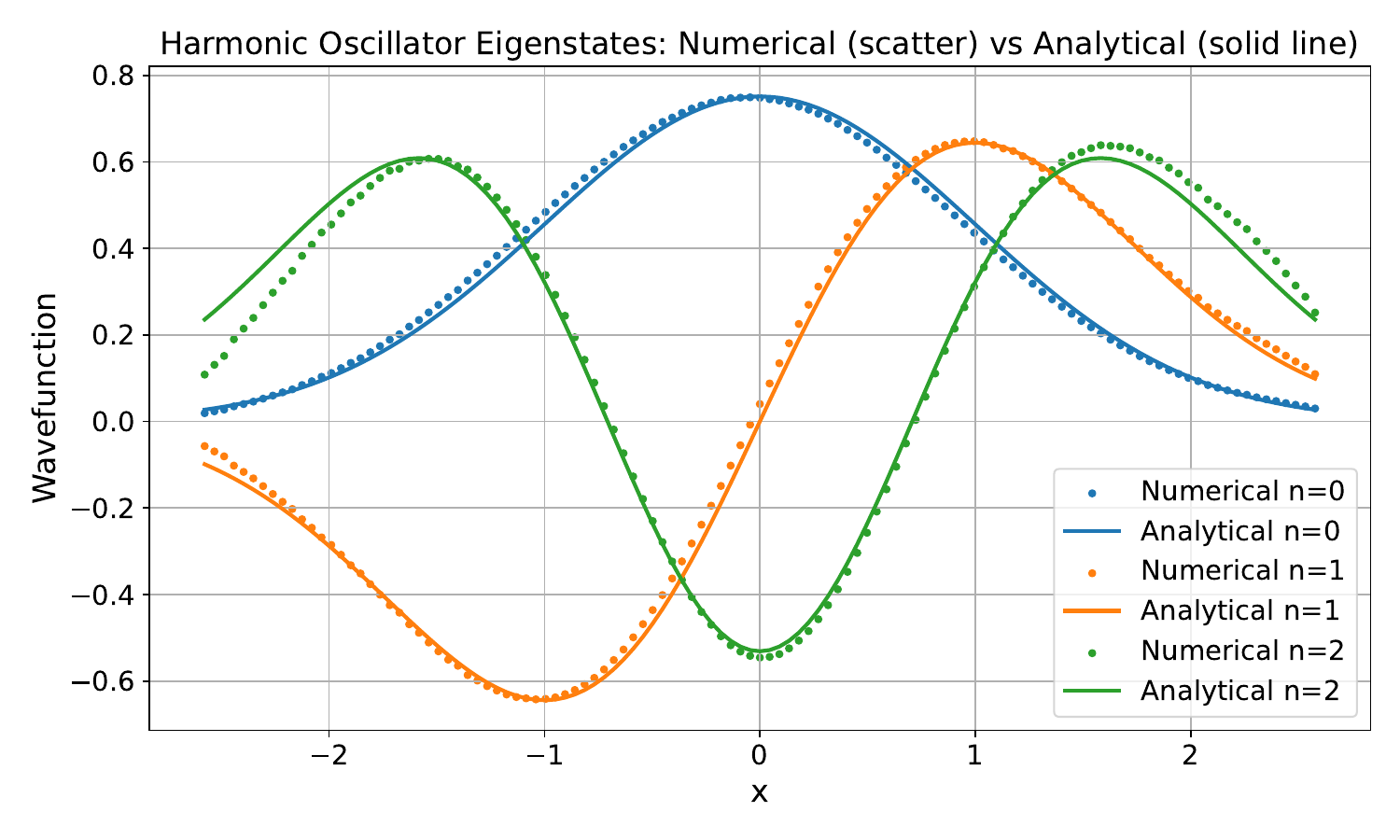}
    \caption{We compare the first three numerical energy eigenstates (scatter), obtained by diagonalizing the estimated Hamiltonian, to the true harmonic oscillator eigenstates (solid line). For concreteness, we take $\theta = \sigma = 1$. The Hamiltonian is computed by generating $50,000,000$ one-step sample paths and applying a similarity transformation to the estimated Fokker-Planck operator following the procedure described in the main text.}
    \label{fig:ou_energy_eigenstates}
\end{figure}

\subsubsection{Deep NN-QM with Ornstein-Uhlenbeck Inputs}\label{sec:NN_of_OU}

We now turn to an investigation of deep neural network quantum mechanics, where a stochastic process $x_t$ is obtained via a neural network whose inputs are themselves stochastic processes satisfying all of the OS axioms.
Since the Ornstein-Uhlenbeck process is the \emph{unique} stationary Gaussian process that is also Markov, this is a natural place to begin. We will comment on this further in the conclusion.

We focus on the architecture
\begin{align}\label{nn_ou_arch}
    x_t = \frac{1}{\sqrt{N}} \sum_{i=1}^{N} \sum_{j=1}^{d} w_i^{(1)} \sigma \left( w_{ij}^{(0)} y_t^{(j)} \right) \, ,
\end{align}
where $y_t^{(j)}$, $j = 1 , \ldots , d$, are a collection of input processes which are stationary, symmetric, reflection positive, and mixing. In the numerical simulations below, for concreteness we take $d = 3$, $N = 10$, we choose the $y_t^{(j)}$ to be Ornstein-Uhlenbeck processes with parameters $\theta^{(j)}$ and $\sigma^{(j)}$. The noise parameters $\sigma^{(j)}$ are not to be confused with the activation function $\sigma$ in equation (\ref{nn_ou_arch}), which we take to be
\begin{align}\label{tanh}
    \sigma ( z ) = \tanh ( z ) \, .
\end{align}
One can either choose the weights $w_i^{(1)}$ and $w_{ij}^{(0)}$ to be identical across all sample paths $x_t$ -- in which case $y_t$ is a deterministic function of a linear combination of RP processes, which is RP by Propositions \ref{lincomb} and \ref{detfun} -- or we can take the weights to be random variables which are drawn independently for each sample path, which is RP by the generalized result given in Theorem \ref{deep_nnqm_theorem}. Here we take the latter approach, since this illustrates the novel construction proposed in this work. Specifically, all of the weight variables are taken to be independent and identically drawn from Gaussians with mean $0$ and variance $1$,
\begin{align}\label{wdist}
    w \sim \mathcal{N} ( 0 , 1 ) \, ,
\end{align}
where $w$ schematically represents all $w_i^{(1)}$ and $w_{ij}^{(1)}$. 

Let us briefly comment on the bias variables, which we have set to zero in this construction. The reason is related to the mixing property, or cluster decomposition. As we alluded to above, if $y_t$ is a mixing process and $f$ is a measurable \emph{deterministic} function, then $x_t = f ( y_t )$ is also mixing. However, if $f$ is a \emph{random} function, then mixing may not be preserved. A simple example is adding a random shift $\theta \sim P ( \theta )$ which is drawn separately for each sample path. If $y_t$ is mixing and
\begin{align}\label{shift_mixing}
    \phi_\theta ( y ) = y + \theta \, ,
\end{align}
then the process $x_t = \phi_\theta ( y_t )$ is not mixing in general. This is because the statistics of the process $x_t$, roughly speaking, involve an ensemble averaging over various shifts $\theta$ which differ between sample paths, and this averaging of shifts can introduce correlations which do not decay at large separation. Such a random shift (\ref{shift_mixing}) is essentially what randomly chosen bias variables in the neural network implement, and indeed we see that the mixing property (or cluster decomposition) fails if we include these random biases in our network. Therefore, we turn off the bias parameters, which leads to an output process $y_t$ whose empirical two-point function obeys cluster decomposition, as we investigate below.

\subsubsection*{\ul{\it Two-point function}}

As in our discussion of the Ornstein-Uhlenbeck process, the first natural observable to study for our deep NN-QM model is the two-point function. Because we take $d = 3$ input OU processes $y_t^{(i)}$, each of which is determined by a drift parameter $\theta^{(i)}$ and a noise amplitude $\sigma^{(i)}$, there are six available parameters to tune in the definition of $x_t$. The experimental values of $G^{(2)} ( t, 0 )$ for several choices of these parameters are displayed in Figure \ref{fig:NN_ou_two_point}. All other parameters in the neural network -- including the width $N = 10$ of the hidden layers, the distribution (\ref{wdist}) of weights, and the choice (\ref{tanh}) of tanh activation function -- are fixed across all of the experiments.

The results of these simulations confirm several basic properties that one might expect of our deep NN-QM model. First, the correlation function $G^{(2)} ( t, 0 )$ decays to zero at large $t$, which gives evidence that mixing or cluster decomposition holds in this theory (as we mentioned above, this is a consequence of our choice to set the biases to zero, and mixing fails if we choose random non-zero biases). Second, we see that the behavior of the two-point function changes as the parameters $\theta^{(i)}$ and $\sigma^{(i)}$ of the input OU processes are varied, as it should. Third, although we have no theoretical prediction for an analytical expression for the two-point function in this model to which the experimental results could be compared, by equation (\ref{spectral_two_point}) we expect the leading exponential fall-off of $G^{(2)}$ to be determined by the gap $E_1$ from the ground state to the first excited state (since $E_0 = 0$ by construction in the stationary state of a stochastic process). By performing exponential fits with the functional form $G^{(2)} = A e^{- \lambda t}$, we therefore obtain approximations of the first excited state energies in these deep NN-QM models, which are also recorded in Figure \ref{fig:NN_ou_two_point}.

\begin{figure}[htbp]
    \includegraphics[width=\linewidth]{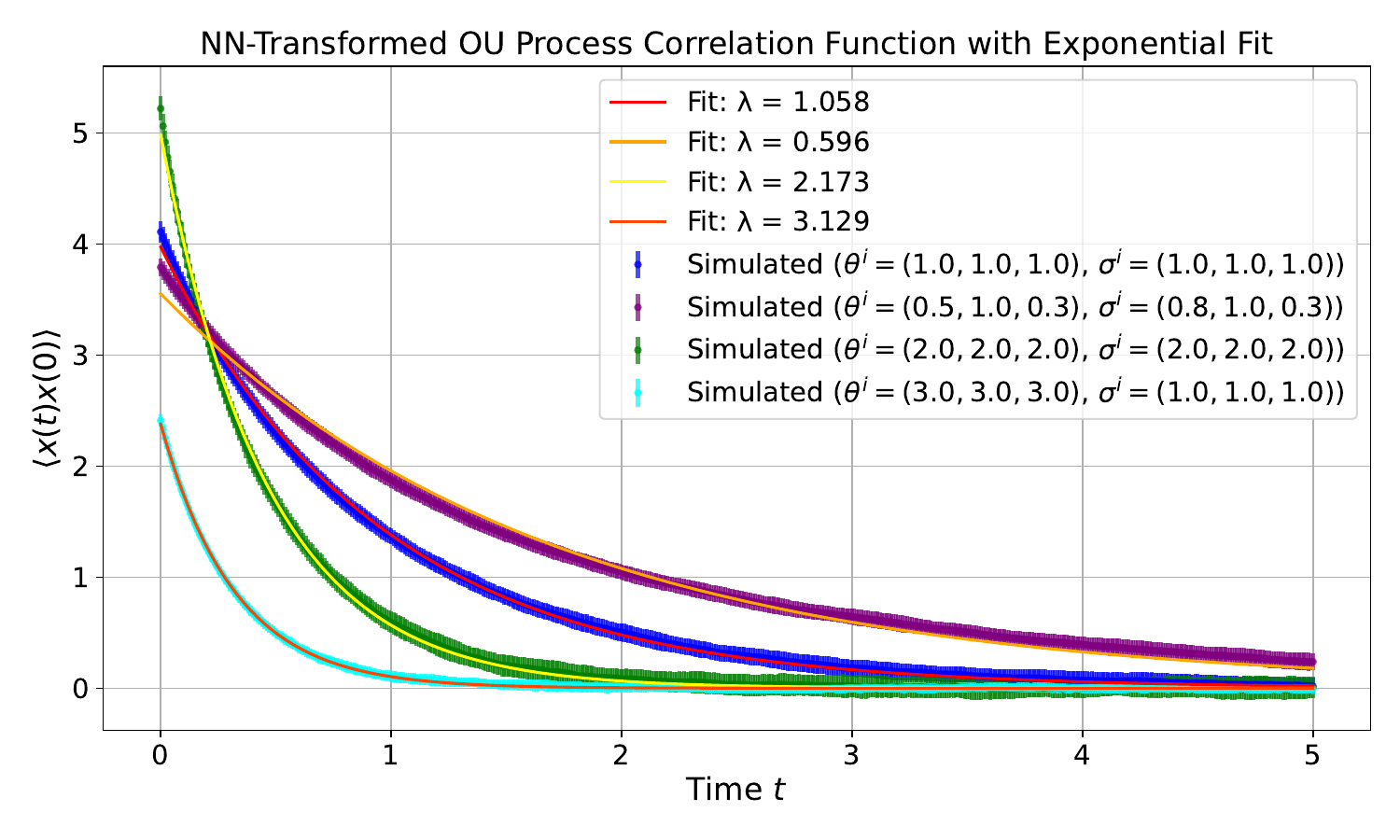}
    \caption{The two-point function $G^{(2)} ( t, 0 ) = \langle x ( t ) x ( 0 ) \rangle$ for the deep NN-QM described in the main text, as the parameters $\theta^i$ and $\sigma^i$ of the three input Ornstein-Uhlenbeck processes $y_t^{(i)}$ are varied. In all cases, the correlation function decays to zero at large time separations, consistent with cluster decomposition. For each set of parameters, we perform $100$ experiments with $5,000$ paths per experiment; the average values of $G^{(2)}$ are shown in cool colors (blue, purple, green, cyan) with error bars indicating the standard deviation across experiments. Although no theoretical result is available for comparison, we perform a fit of each set of data to a decaying exponential $G^{(2)} ( t, 0 ) = A e^{- \lambda t}$ and show the best-fit curves in warm colors. The best-fit values of $\lambda$ yield estimates of the first excited state energy $E_1$ in each deep NN-QM.}
    \label{fig:NN_ou_two_point}
\end{figure}

\subsubsection*{\ul{\it Commutator}}

For the Ornstein-Uhlenbeck process, we found that the discrete approximation $C(t)$ of the commutator $[ \hat{x} ( t ) , \hat{p} ( t ) ]$, defined in equation (\ref{Ct_defn}), was constant in time with a magnitude set by the effective $\hbar$ of the stochastic process. This should, of course, hold for any conventional quantum system with a quadratic kinetic term. However, since the deep NN-QM that we have defined here is obtained from a fairly involved transformation of several OU processes, it is not clear \emph{a priori} that it corresponds to a familiar quantum system for which $p(t) \sim \dot{x} ( t )$ so that $C(t)$ is constant.

Nonetheless, repeating the numerical computation of $C(t)$ via simulations as described above, we find that this is indeed the case. These results are presented in Figure \ref{fig:NN_ou_commutator}, which shows that $C(t)$ is also constant for deep NN-QM, with an effective $\hbar$ that is a function of the noise parameters $\sigma^{(i)}$ defining the input OU processes.

\begin{figure}[htbp]
    \includegraphics[width=\linewidth]{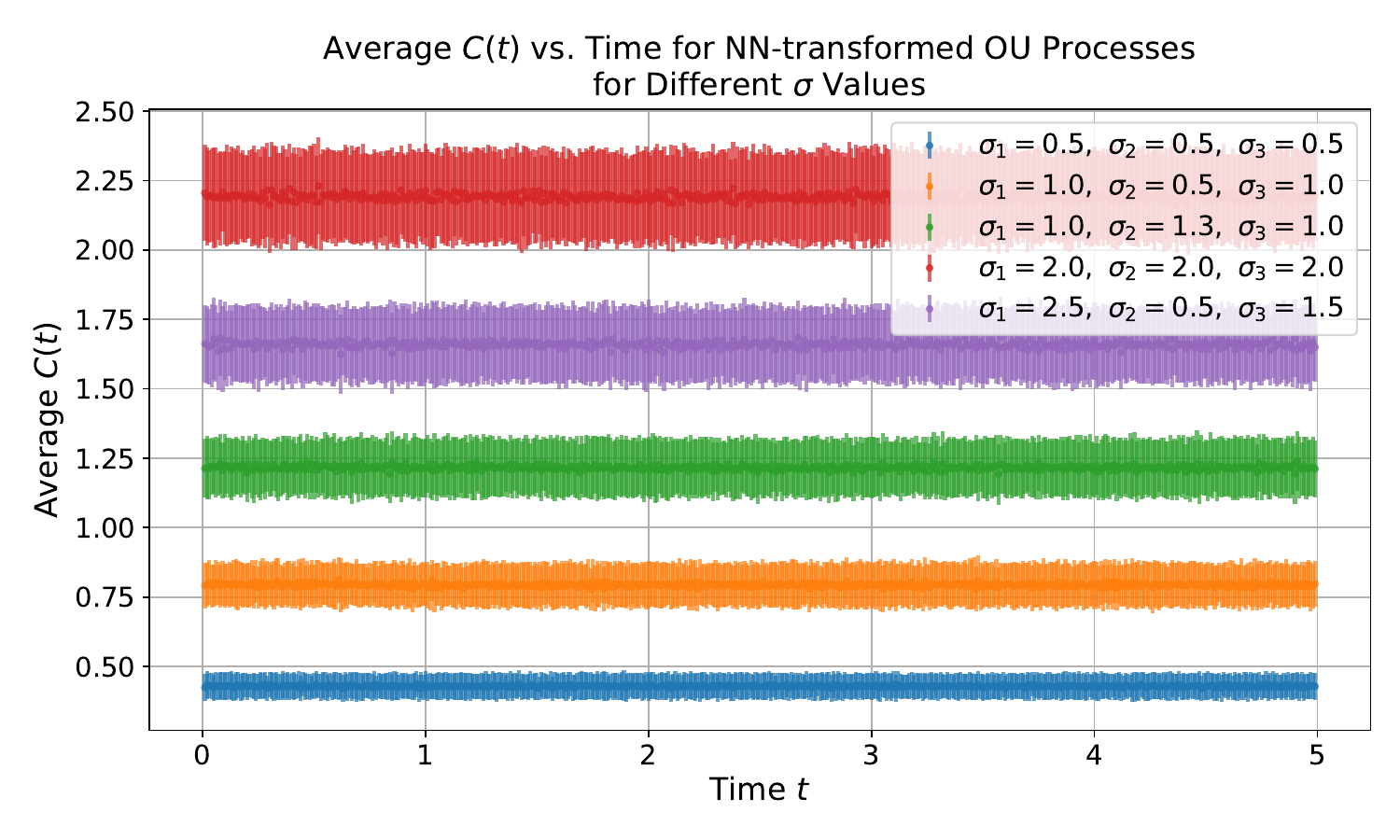}
    \caption{We show the average value of the discretized commutator $C(t)$ over time, as in Figure \ref{fig:ou_commutator} for the OU process, but here for the deep NN-QM $x_t$ obtained from three Ornstein-Uhlenbeck inputs $y_t^{(i)}$ and a single hidden layer of $N = 10$ neurons, as described in the main text. For each tuple of parameters $(\sigma_1, \sigma_2, \sigma_3)$ associated with the input Ornstein-Uhlenbeck processes, we perform $200$ experiments with $10,000$ sample paths per experiment, all with step size $\Delta t = 0.01$ and total simulation time $T = 5$. The displayed error bars represent the standard deviation across experiments. The resulting deep NN-QM still exhibits a conventional commutation relation $[\hat{x} , \hat{p}] = \hbar_{\text{eff}}$ with an effective Planck constant that depends on the noise parameters $\sigma_i$ of the three input processes.}
    \label{fig:NN_ou_commutator}
\end{figure}

\subsubsection*{\ul{\it Uncertainty relations}}

It is also interesting to consider the behavior of the uncertainties $\Delta x$ and $\Delta p$ in the deep NN-QM model. These uncertainties must, of course, obey the Heisenberg uncertainty relation, but we no longer expect the product $\Delta x \, \Delta p$ to saturate the bound, since this is a special property of minimum uncertainty states such as Gaussian wavepackets.

\begin{table}[htbp]
   \noindent\makebox[\linewidth][l]{\hspace*{-1.75cm}
  \begin{tabular}{ccccccccccc}
    \toprule
    \(\theta_1\) & \(\theta_2\) & \(\theta_3\) & \(\sigma_1\) & \(\sigma_2\) & \(\sigma_3\) & \(dt\) & \( \left( \Delta x \right)^2 \) & \( \left( \Delta p \right)^2 \) & \( \left( \Delta x \right)^2 \cdot \left( \Delta p \right)^2 \) & \( \frac{\hbar_{\text{eff}}^2}{4} \)\\
    \midrule
    \(\frac{1}{2}\) & \(\frac{1}{2}\) & \(\frac{1}{2}\) & \(\frac{1}{2}\) & \(\frac{1}{2}\) & \(\frac{1}{2}\) & 0.100 & 0.3031 \(\pm\) 0.0029 & 0.1360 \(\pm\) 0.0018 & 0.0412 \(\pm\) 0.0008 & 0.0307 \(\pm\)  0.0005 \\
    \(\frac{1}{2}\) & \(\frac{1}{2}\) & \(\frac{1}{2}\) & \(\frac{1}{2}\) & \(\frac{1}{2}\) & \(\frac{1}{2}\) & 0.010 & 0.2997 \(\pm\) 0.0033 & 0.1480 \(\pm\) 0.0178 & 0.0444 \(\pm\) 0.0053 & 0.0307 \(\pm\)  0.0006 \\
    \(1\)         & \(1\)         & \(1\)         & \(1\)         & \(1\)         & \(1\)         & 0.100 & 0.4175 \(\pm\) 0.0037 & 0.9227 \(\pm\) 0.0086 & 0.3852 \(\pm\) 0.0066 & 0.2687 \(\pm\)  0.0044\\
    \(1\)         & \(1\)         & \(1\)         & \(1\)         & \(1\)         & \(1\)         & 0.010 & 0.4104 \(\pm\) 0.0040 & 1.2741 \(\pm\) 0.0553 & 0.5229 \(\pm\) 0.0232 & 0.2850 \(\pm\)  0.0052 \\
    \(2\)         & \(2\)         & \(2\)         & \(2\)         & \(2\)         & \(2\)         & 0.100 & 0.5393 \(\pm\) 0.0044 & 4.8038 \(\pm\) 0.0352 & 2.5908 \(\pm\) 0.0387 & 1.9196 \(\pm\)  0.0292 \\
    \(2\)         & \(2\)         & \(2\)         & \(2\)         & \(2\)         & \(2\)         & 0.010 & 0.5236 \(\pm\) 0.0047 & 10.4247 \(\pm\) 0.1741 & 5.4586 \(\pm\) 0.1127 & 2.4223 \(\pm\)  0.0360 \\
    \(\frac{1}{2}\) & \(1\)         & \(2\)         & \(\frac{1}{2}\) & \(1\)         & \(2\)         & 0.100 & 0.4365 \(\pm\) 0.0042 & 2.1579 \(\pm\) 0.0235 & 0.9420 \(\pm\) 0.0183 & 0.6104 \(\pm\)  0.0121 \\
    \(\frac{1}{2}\) & \(1\)         & \(2\)         & \(\frac{1}{2}\) & \(1\)         & \(2\)         & 0.010 & 0.4263 \(\pm\) 0.0039 & 3.7006 \(\pm\) 0.1075 & 1.5777 \(\pm\) 0.0515 & 0.6834 \(\pm\) 0.0150 \\
    \(2\)         & \(1\)         & \(\frac{1}{2}\) & \(2\)         & \(1\)         & \(\frac{1}{2}\) & 0.100 & 0.4371 \(\pm\) 0.0046 & 2.1610 \(\pm\) 0.0255 & 0.9446 \(\pm\) 0.0204 & 0.6100 \(\pm\) 0.0120 \\
    \(2\)         & \(1\)         & \(\frac{1}{2}\) & \(2\)         & \(1\)         & \(\frac{1}{2}\) & 0.010 & 0.4255 \(\pm\) 0.0039 & 3.6972 \(\pm\) 0.1051 & 1.5731 \(\pm\) 0.0492 & 0.6819 \(\pm\)  0.0138 \\
    \(\frac{1}{2}\) & \(2\)         & \(1\)         & \(2\)         & \(\frac{1}{2}\) & \(1\)         & 0.100 & 0.5329 \(\pm\) 0.0059 & 1.0278 \(\pm\) 0.0113 & 0.5477 \(\pm\) 0.0108 & 0.2674 \(\pm\) 0.0054 \\
    \(\frac{1}{2}\) & \(2\)         & \(1\)         & \(2\)         & \(\frac{1}{2}\) & \(1\)         & 0.010 & 0.5287 \(\pm\) 0.0061 & 1.8866 \(\pm\) 0.0762 & 0.9974 \(\pm\) 0.0415 & 0.3101 \(\pm\)  0.0066 \\
    \(1\)         & \(\frac{1}{2}\) & \(2\)         & \(1\)         & \(2\)         & \(\frac{1}{2}\) & 0.100 & 0.5330 \(\pm\) 0.0056 & 1.0282 \(\pm\) 0.0107 & 0.5481 \(\pm\) 0.0105 & 0.2671 \(\pm\)  0.0053 \\
    \(1\)         & \(\frac{1}{2}\) & \(2\)         & \(1\)         & \(2\)         & \(\frac{1}{2}\) & 0.010 & 0.5286 \(\pm\) 0.0057 & 1.8828 \(\pm\) 0.0840 & 0.9952 \(\pm\) 0.0451 & 0.3098 \(\pm\) 0.0057 \\
    \(2\)         & \(\frac{1}{2}\) & \(2\)         & \(\frac{1}{2}\) & \(1\)         & \(1\)         & 0.100 & 0.3678 \(\pm\) 0.0038 & 0.9056 \(\pm\) 0.0099 & 0.3331 \(\pm\) 0.0063 & 0.1943 \(\pm\) 0.0035 \\
    \(2\)         & \(\frac{1}{2}\) & \(2\)         & \(\frac{1}{2}\) & \(1\)         & \(1\)         & 0.010 & 0.3603 \(\pm\) 0.0042 & 1.1025 \(\pm\) 0.0492 & 0.3973 \(\pm\) 0.0188 & 0.1934 \(\pm\) 0.0037 \\
    \bottomrule
  \end{tabular}
    }
  \caption{We compute the uncertainties in position and momentum for the deep NN-QM process with three OU inputs. The definitions (\ref{pdefn}) and (\ref{psq_defn}) are used to compute the momentum uncertainties. We simulate $200$ epochs with $10,000$ paths per epoch, each with $1,000$ update steps, for each combination of parameters, and the standard deviation across epochs appears following the symbol $\pm$ for each quantity. We compare the simulated uncertainties for several values of the parameters $\theta^{(i)}$ and $\sigma^{(i)}$ defining the input OU processes, and for two values of the time step $dt$. We also include the value $\frac{\hbar_{\text{eff}}^2}{4}$, which gives a lower bound for the product $\left( \Delta x \right)^2 \cdot \left( \Delta p \right)^2$ by the Heisenberg uncertainty principle, where $\hbar_{\text{eff}}$ is estimated from the commutator $C(t)$ which approximates $[\hat{x} , \hat{p} ]$. We find that $\left( \Delta x \right)^2 \cdot \left( \Delta p \right)^2 > \frac{\hbar_{\text{eff}}^2}{4}$ in all cases, so the uncertainty bound is not saturated.}
  \label{tab:uncertainties_NN_OU}
\end{table}

The simulated uncertaintes $\left( \Delta x \right)^2$ and $\left( \Delta p \right)^2$ for several choices of parameters are displayed in Table \ref{tab:uncertainties_NN_OU}. Again, we have no theoretical prediction for the product of uncertainties, but the lower bound implied by the Heisenberg uncertainty relation is displayed in the final column of Table \ref{tab:uncertainties_NN_OU}. This lower bound is not uniform across all rows of this table since it is set by the effective value of $\hbar$ in the model. We estimate $\hbar_{\text{eff}}$ by computing the commutator $C(t)$, and as we have already seen in Figure \ref{fig:NN_ou_commutator}, this quantity depends on the parameters. We find that the Heisenberg uncertainty relation is not saturated for any choices of parameters that we consider in this deep NN-QM theory.

\subsubsection*{\ul{\it Spectrum and eigenfunctions}}

Finally, we repeat the calculation of the energy eigenvalues and eigenfunctions for the deep NN-QM model. These results are shown in Figure \ref{fig:NN_ou_energy_eigenvalues}. We see that the energy eigenvalues are approximately linear in $n$ for small $n$, as in the harmonic oscillator, but upon investigating higher energy eigenvalues we find that the curve is concave down. The rate at which the energies ``level off'' seems to be set by the $\theta$ parameters of the input Ornstein-Uhlenbeck processes, with larger $\theta$ values giving models that deviate from the linear SHO-like spectrum more quickly. We also see that, unlike the OU example, the spectrum appears to depend on the noise parameters $\sigma^{(i)}$ in addition to the $\theta^{(i)}$.

We do not display the energy eigenfunctions since their shape is very similar to that of the OU eigenfunctions in Figure \ref{fig:ou_energy_eigenstates}. This suggests that -- at least for low-lying states -- the effect of the deep NN transformation primarily modifies the energy eigen\emph{values}, while leaving the energy eigen\emph{functions} essentially unchanged. This is similar to the behavior of deformations of quantum systems where the seed Hamiltonian $H_0$ is mapped to a function $H = f ( H_0 )$ of the undeformed Hamiltonian \cite{Gross:2019ach,Gross:2019uxi,Ferko:2023ozb,Ferko:2023iha}. However, at high energies, we see that the deep NN-QM spectrum levels off, so the effect of the transformation on the eigenfunctions for high-energy states may be more dramatic.

\begin{figure}[htbp]
    \includegraphics[width=\linewidth]{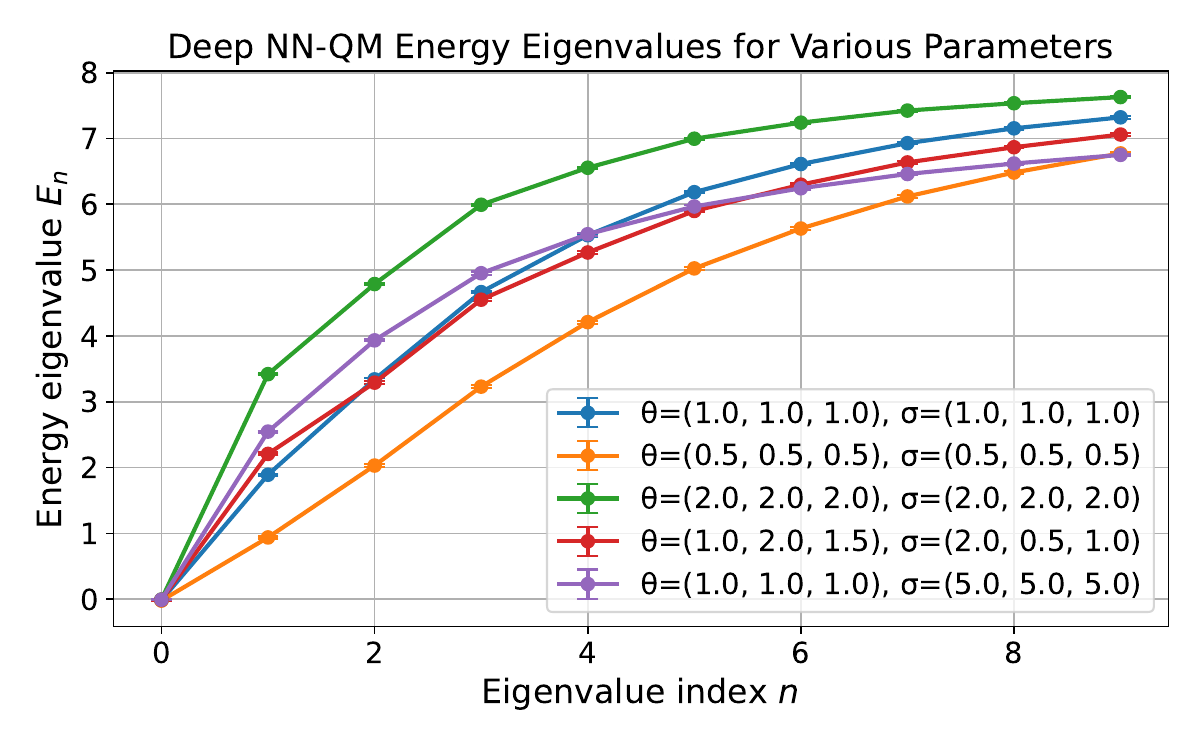}
    \caption{The first $10$ energy eigenvalues for our deep NN-QM model. Energies are computed by numerically estimating the Hamiltonian operator using the transition matrix and a similarity transformation involving the empirical stationary distribution. We perform $20$ trials and display the mean energies with error bars (visible upon zooming in) set by the standard deviation across trials. Unlike the harmonic oscillator energy levels, the deep NN-QM energies do not grow linearly, but rather level off at large $n$.}
    \label{fig:NN_ou_energy_eigenvalues}
\end{figure}

\section{Conclusion}\label{sec:conclusion}

In this paper we developed a neural network (NN) approach to quantum mechanics (QM), focusing on universality of the description and the appearance of unitarity via reflection positivity. The defining data of the construction is a neural network and associated parameter density, $(\phi_\theta(t), P(\theta))$, where $\theta$ are the parameters and the functional form of $\phi_\theta(t)$ is the architecture. This is restriction of a scalar NN-FT to the case $d=1$, yielding a stochastic process (SP) that, in general, is not quantum mechanical; it is a NN-SP.

We were led to a number of questions: $1)$ when is a SP actually a Euclidean QM theory? $2)$ is every Euclidean QM theory representable by a NN? The first question depends in part on which axioms one wants to demand of a QM theory --- and, of course, many are canonical --- but we began by adopting minimal requirements, in which mean-square continuity of $x(t)$ and the K\"all\'en-Lehmann spectral representation ensure the applicability of the Kosambi-Karhunen-Lo\`{e}ve theorem \cite{kosambi1943statistics,karhunen1947,Loeve1948} from the literature on stochastic processes. The associated decomposition can be interpreted as a neural network with varying neurons, answering the second question in the affirmative: any theory satisfying the minimal requirements admits a NN description.

Of course, these requirements are not sufficient, as unitarity is central in conventional quantum mechanics. In constructive approaches to QFT, unitarity is encoded in a property called reflection positivity (RP) that is satisfied by Euclidean correlators, and the property may also be studied in $d=1$ field theories, i.e., quantum mechanics. We were therefore led to the study of RP in NN descriptions of QM systems, and focused on two constructions. The first was a ``parameter splitting" mechanism that allows the relevant quantity for RP to be expressed as perfect square integrand in an integral over NN parameters $\theta$. This mechanism may be applied as a concrete modification to \emph{any} neural network, but unfortunately simple instantiations break translation invariance. This can be remedied by performing the parameter splitting for all times, at the cost of introducing a continuous infinity of NN parameters and nowhere-analytic paths $x(t)$. The latter is familiar from ordinary QM, as non-differentiability (which implies non-analyticity) is the origin of non-zero commutators in the Feynman path integral. In the second, we studied Markov processes, which are known to be RP, and 
we showed that a NN acting on any Markov or RP process produces an RP process, preserving unitarity.

This provides a mechanism for defining a vast array of QM theories, since a NN acting on an RP process preserves RP. In particular,  NN layers may be successively applied to the original process, defining a different NN-QM theory at each successive layer. We studied this idea in numerical examples, including an Ornstein-Uhlenbeck process (the SP that is the Euclidean realization of the quantum harmonic oscillator) and standard neural network layers acting on it. In each case we studied the appearance of non-zero commutators, uncertainty, and the spectrum.

Our work opens many interesting directions for future work:
\begin{enumerate}
\item \textbf{Hamiltonian engineering in deep NN-QM.} In the ML literature and NN-FT, the recursive application of NN layers allows one to understand how the correlations at one layer relate to and influence the correlations at the next. In the context of our construction, it is natural to ask whether this influence may be utilized to engineer neural network QM theories whose Hamiltonians have desired properties.

One avenue for engineering desired QM theories with deep NN-QM is to learn them. Specifically, let $x_t$ be any Markov or RP process. Then since 
\begin{equation}
    f_\varphi(x_t) = \text{RP}
\end{equation}
for \emph{any} neural network $f$ with its own parameters $\varphi$, we may use traditional machine learning to optimize the parameters $\varphi$ to obtain a desired QM theory, e.g. according to its spectrum or other properties. This idea also applies for translation-invariant (stationary) and/or time reflection invariant (symmetric) theories, as these properties persist in passing from the process $x_t \mapsto f_\varphi(x_t)$. In this case, $f_\varphi$ is one neural network, with some initial draw $\varphi\sim P(\varphi)$ that is subsequently optimized.

Alternatively, instead of using a single (deterministic) neural network to convert $x_t$ to another NN-QM process, one could apply the entire ensemble $f_\varphi$ (for all $\varphi$ in the parameter domain, drawn appropriately) to $x_t$. In such a case this ensemble theory is optimized by using a final auxiliary neural network such as a normalizing flow to optimize the density on $\varphi$ in order to achieve the desired theory.

Interestingly, the OU process provides a natural starting point for such studies. To satisfy the OS QM axioms, we need a stationary symmetric RP process. Specializing to Markov rather than general RP, the OU process is the \emph{unique} \cite{DoobJ.L.1942TBMa} such process that is also Gaussian. Therefore NN(OU) is \emph{the} way to perturb (or, more precisely, modify) a Markov Gaussian process in this context.

\item \textbf{Completeness of deep NN-QM}. A natural future direction is to attempt to prove or disprove the following conjecture.
\begin{conjecture}\label{conjecture}
    Every reflection-positive process admits a representation as a deep NN-QM whose inputs $y_t^{(i)}$ are symmetric Markov processes.
\end{conjecture}
The results of this work establish that every stochastic process $x_t$ that admits a deep NN-QM representation with symmetric Markov inputs is reflection positive. The content of the conjecture is that this set of processes is \emph{complete}: every RP process can be obtained from such a deep NN-QM. Said differently, the space of reflection positive processes is the closure of the space of symmetric Markov processes under the operations of taking linear combinations and applying random functions.  If true, this would provide another example of the universality of neural networks (besides Theorem \ref{KKL_theorem}), and give an alternative characterization of reflection positive stochastic processes. Indeed, rather than defining reflection positivity by (\ref{rp_condition}), one could equivalently define an RP process as one obtained from applying deep neural networks to collections of symmetric Markov processes.

\item \textbf{Higher-dimensional field theories}. The standard construction of a free scalar field in $d$-dimensional Euclidean quantum field theory, also called the ``Gaussian free field'' in more mathematical literature, proceeds by taking
\begin{align}\label{gff}
    \varphi ( x ) = \sum_{k=1}^{\infty} \xi_k \psi_k ( x ) \, ,
\end{align}
where the $\psi_k$ are an orthonormal basis for the Sobolev space $H^1 ( \Omega )$ on some domain $\Omega \subset \mathbb{R}^d$ and the $\xi_k$ are i.i.d. random variables drawn from a normal distribution with mean zero and variance $1$.

The behavior of this sum is very different in $d = 1$ compared to $d \geq 2$. For $d = 1$, (\ref{gff}) can be viewed as the KKL decomposition of a stochastic process, which almost surely converges to a continuous nowhere-differentiable function. However, for $d \geq 2$, the sum diverges almost surely for any value of the input $x$. In this case, $\varphi$ does not define an ordinary function, but rather a generalized function or distribution. This is a signal of the familiar fact that, for quantum field theories in $d \geq 2$ dimensions, the path integral receives contributions only from distributions; ordinary functions are a set of measure zero with respect to the path integral measure.

An important extension of the analysis in this work is to construct neural network quantum field theories in $d \geq 2$ dimensions which satisfy all of the OS axioms. This will require a detailed investigation of neural networks which define generalized functions, or distributions, rather than ordinary functions. 
\end{enumerate}
We hope to return to some of these interesting and important directions in future work.

\section*{Acknowledgements}

We thank Hamza Ahmed, Ning Bao, and Jonathan Weitsman for helpful discussions. This work is supported by the National Science Foundation under Cooperative Agreement PHY-2019786 (the NSF AI Institute for Artificial Intelligence and Fundamental Interactions). J.H. is supported by NSF CAREER grant PHY-1848089.

\appendix

\section{Commutators and Non-Differentiability}\label{app:commutators}

In this appendix, we will review one argument that the quantum-mechanical path integral must have contributions from paths that are differentiable nowhere. In particular, the inclusion of non-differentiable paths is required in order to reproduce non-trivial commutation relations such as $[\hat{x} , \hat{p}] \neq 0$. We emphasize that this is a standard textbook result which we include only to make the present work self-contained.\footnote{Here we follow Section 3.2.1 of \cite{Skinner2018}. For other discussions, see Section 7.3 of \cite{FeynmanHibbsStyer2010} or the note \cite{Ong2012}.}

For simplicity, we consider the path integral formulation for a free quantum particle with position $x(t)$ and restrict to a finite time interval $t \in [ a, b ]$. On this interval, the dynamics is described by the action
\begin{align}
    S [ x ] = \int_a^b dt \, \left( \frac{1}{2} m \dot{x}^2 \right) \, ,
\end{align}
and we will choose units where $m = 1$. Using this normalization, the momentum canonically conjugate to the position $x$ is $p = \dot{x}$.

Our aim is to compute the commutator $[\hat{x} , \hat{p}]$ using the path integral. Of course, in the path integral formulation, the variables $x$ and $\dot{x}$ are ordinary commuting functions rather than operators. Therefore, in order to study the commutator, we should use time-ordering and compare the results of (i) inserting $x ( t_+ ) \dot{x} ( t_- )$, where $t_+ > t_-$, in the path integral, and (ii) using the opposite ordering $x ( t_- ) \dot{x} ( t_+ )$ for the insertion. In the limit as $t_+$ and $t_-$ become coincident, the difference between these quantities (i) and (ii) should measure the commutator $[\hat{x} , \hat{p}]$ at the common value $t = t_+ = t_-$. More precisely, we wish to compute
\begin{align}\label{commutator_defn}
    C ( t_- , t , t_+ ) = \int \mathcal{D} x \, e^{- S [ x ] / \hbar } \left( x ( t ) \dot{x} ( t_- ) - x ( t ) \dot{x} ( t_+ ) \right) \, ,
\end{align}
and then take the limit of this quantity as $t_+ \to t$ from above and $t_- \to t$ from below.

We begin by discretizing the path integral, fixing a partition of the interval $[a, b]$ into a sequence of increasing times $t_k$ separated by intervals $\Delta t$, so that
\begin{align}
    a = t_0 < t_1 < \cdots < t_N = b \, ,
\end{align}
where $t_1 = a + \Delta t$, $t_2 = a + 2 \Delta t$, and so on, up to $t_N = a + N \Delta t = b$. For concreteness, we will also fix boundary conditions for the path integral so that $x ( t_a ) = x_a$ and $x ( t_b ) = x_b$. We indicate the variables representing the positions $x ( t_i )$ with subscripts, such as
\begin{align}
    x_0 = x ( t_0 ) = x_a \, , \quad x_1 = x ( t_1 ) = x ( t_0 + \Delta t ) \, , \quad \cdots \, , \quad x_k = x ( t_k ) = x ( t_0 + k \Delta t ) \, , \quad \cdots \, .
\end{align}
Using this discretization scheme, let us compute the regularized commutator (\ref{commutator_defn}) at the $j$-th time in our sequence and the two times which immediately precede and follow it, i.e.
\begin{align}
    t = t_j \, , \quad t_- = t_{j-1} = t_j - \Delta t \, , \quad t_+ = t_{j+1} = t_j + \Delta t \, .
\end{align}
In particular, we must replace the velocities $\dot{x} ( t_{\pm} )$ by the finite difference quotients
\begin{align}
    \dot{x} ( t_+ ) = \frac{x_{j+1} - x_j}{\Delta t} \, , \qquad \dot{x} ( t_- ) = \frac{x_j - x_{j-1}}{\Delta t} \, ,
\end{align}
which approach the desired derivatives in the limit $\Delta t \to 0$, assuming this limit exists.

In terms of the free particle propagator,
\begin{align}\label{free_propagator}
    K ( y , t + \Delta t ; x , t ) = \frac{1}{\sqrt{ 2 \pi \hbar \Delta t}} \exp \left( - \frac{ ( x - y )^2}{2 \hbar \Delta t} \right) \, ,
\end{align}
the quantity defined in (\ref{commutator_defn}) is
\begin{align}\label{commutator_intermediate}
    C ( t_{j-1} , t_j , t_{j+1} ) &= \int \left( \prod_{k = 1}^{N - 1} d x_k \right) K ( x_b , t_b ; x_{N-1} , t_{N-1} ) \cdot K ( x_{N-1} , t_{N-1} ; x_{N-2} , t_{N-2} ) \cdot \, \cdots \nonumber \\
    &\qquad \cdot K ( x_{j+1} , t_{j+1} ; x_{j} , t_{j} ) \left( x_j \frac{x_j - x_{j-1}}{\Delta t} - x_j \frac{x_{j+1} - x_j}{\Delta t} \right) K ( x_j , t_j ; x_{j-1} , t_{j-1} ) \nonumber \\
    &\qquad \cdot \, \cdots \, \cdot  K ( x_2, t_2 ; x_1, t_1 )  \cdot K ( x_1, t_1 ; x_a, t_a ) \, .
\end{align}
We note that the derivative of the propagator (\ref{free_propagator}) with respect to its final endpoint is
\begin{align}
    \partial_y K ( y , t + \Delta t ; x , t ) &= \frac{x - y}{\hbar \Delta t} \cdot \frac{1}{\sqrt{ 2 \pi \hbar \Delta t}} \cdot \exp \left( - \frac{ ( x - y )^2}{2 \hbar \Delta t} \right) \nonumber \\
    &= - \frac{(y - x)}{\hbar \Delta t} K ( y , t + \Delta t ; x , t ) \, ,
\end{align}
with a similar formula applying to the derivative $\partial_x$ with respect to the initial endpoint. Therefore, the two expressions appearing in (\ref{commutator_intermediate}) can be rewritten using the relations
\begin{align}
    \left( x_j \frac{x_j - x_{j-1}}{\Delta t}  \right) K ( x_j , t_j ; x_{j-1} , t_{j-1} ) &= - \hbar x_j \partial_{x_j} K ( x_j , t_j ; x_{j-1} , t_{j-1} ) \, , \nonumber \\
    K ( x_{j+1} , t_{j+1} ; x_{j} , t_{j} ) \left( - x_j \frac{x_{j+1} - x_j}{\Delta t} \right) &= \hbar x_j \partial_{x_j} K ( x_{j+1} , t_{j+1} ; x_{j} , t_{j} ) \, ,
\end{align}
respectively. Collecting these terms and rewriting them using the product rule, we have
\begin{align}\label{commutator_intermediate_two}
    C ( t_{j-1} , t_j , t_{j+1} ) &= - \hbar \int \left( \prod_{k = 1}^{N - 1} d x_k \right) K ( x_b , t_b ; x_{N-1} , t_{N-1} ) \cdot \, \cdots \nonumber \\
    &\qquad \cdot x_j \partial_{x_j} \Big( K ( x_{j+1} , t_{j+1} ; x_{j} , t_{j} )  K ( x_j , t_j ; x_{j-1} , t_{j-1} ) \Big) \cdot \, \cdots \, \cdot K ( x_1, t_1 ; x_a, t_a ) \, .
\end{align}
Since $x_j$ is integrated over, we may integrate by parts and discard the boundary term at the cost of a sign. The only other dependence on $x_j$ in the integrand is in the $x_j$ prefactor multiplying the $\partial_{x_j}$ derivative, which is therefore replaced with $\partial_{x_j} x_j = 1$:
\begin{align}\label{commutator_intermediate_hree}
    C ( t_{j-1} , t_j , t_{j+1} ) &= \hbar \int \left( \prod_{i = 1}^{N - 1} d x_i \right) K ( x_b , t_b ; x_{N-1} , t_{N-1} ) \cdot \, \cdots \nonumber \\
    &\qquad \cdot K ( x_{j+1} , t_{j+1} ; x_{j} , t_{j} )  K ( x_j , t_j ; x_{j-1} , t_{j-1} ) \cdot \, \cdots \, \cdot K ( x_1, t_1 ; x_a, t_a ) \, .
\end{align}
We can then evaluate the remaining integral over $x_j$ using the propagator identity
\begin{align}
    \int d x_j \, K ( x_{j+1} , t_{j+1} ; x_{j} , t_{j} )  K ( x_j , t_j ; x_{j-1} , t_{j-1} ) = K ( {x_{j+1}} , t_{j+1} ; x_{j-1} , t_{j-1} ) \, ,
\end{align}
and using a similar identity for each of the other integration variables $x_k$ collapses the entire expression to a single propagator:
\begin{align}
    C ( t_{j-1} , t_j , t_{j+1} ) = \hbar K ( x_b , t_b ; x_a , t_a ) \, .
\end{align}
The result is independent of $\Delta t$ and the discretization, so we may freely take the limit $\Delta t \to 0$ so that $t_- = t_{j-1}$ approaches $t = t_j$ from below and $t_+ = t_{j+1}$ approaches $t = t_j$ from above. The result is $\hbar$ times the propagator which would have been computed by the path integral (\ref{commutator_defn}) if we had \emph{not} inserted $x ( t ) \dot{x} ( t_- ) - x ( t ) \dot{x} ( t_+ )$. This is exactly what one would expect if we had evaluated the commutator in operator language,
\begin{align}
    \langle x_b \mid [ \hat{x} ( t ) , \hat{p} ( t ) ] \mid x_a \rangle = \hbar \, \langle x_b \mid x_a \rangle \, ,
\end{align}
since $[ \hat{x} ( t ) , \hat{p} ( t ) ] = \hbar$, with no factor of $i$ because we work in Euclidean signature.

We conclude that the path integral formulation can indeed reproduce the desired commutation relations. However, for this argument it was crucial that the path integral receives contributions from paths that are non-differentiable. If all of the trajectories $x ( t )$ in the path integral were instead smoothly differentiable, then one would have
\begin{align}
    \lim_{\Delta t \to 0} \left(  x_j \frac{x_j - x_{j-1}}{\Delta t} - x_j \frac{x_{j+1} - x_j}{\Delta t} \right) = 0 \, ,
\end{align}
and thus $C ( t_- , t , t_+ ) = 0$ in the continuum limit.\footnote{See \cite{Koch:2014cma} for a discussion of other modifications to quantum mechanics that arise from restricting to differentiable paths in the path integral.}

\bibliographystyle{utphys}
\bibliography{master}

\providecommand{\href}[2]{#2}\begingroup\raggedright\begin{thebibliography}{10}

\bibitem{Halverson:2020trp}
J.~Halverson, A.~Maiti, and K.~Stoner, ``{Neural Networks and Quantum Field Theory},'' \href{http://dx.doi.org/10.1088/2632-2153/abeca3}{{\em Mach. Learn. Sci. Tech.} {\bfseries 2} no.~3, (2021) 035002}, \href{http://arxiv.org/abs/2008.08601}{{\ttfamily arXiv:2008.08601 [cs.LG]}}.

\bibitem{Halverson:2021aot}
J.~Halverson, ``{Building Quantum Field Theories Out of Neurons},'' \href{http://arxiv.org/abs/2112.04527}{{\ttfamily arXiv:2112.04527 [hep-th]}}.

\bibitem{Demirtas:2023fir}
M.~Demirtas, J.~Halverson, A.~Maiti, M.~D. Schwartz, and K.~Stoner, ``{Neural network field theories: non-Gaussianity, actions, and locality},'' \href{http://dx.doi.org/10.1088/2632-2153/ad17d3}{{\em Mach. Learn. Sci. Tech.} {\bfseries 5} no.~1, (2024) 015002}, \href{http://arxiv.org/abs/2307.03223}{{\ttfamily arXiv:2307.03223 [hep-th]}}.

\bibitem{Maiti:2021fpy}
A.~Maiti, K.~Stoner, and J.~Halverson, ``{Symmetry-via-Duality: Invariant Neural Network Densities from Parameter-Space Correlators},'' \href{http://arxiv.org/abs/2106.00694}{{\ttfamily arXiv:2106.00694 [cs.LG]}}.

\bibitem{Halverson:2024axc}
J.~Halverson, J.~Naskar, and J.~Tian, ``{Conformal Fields from Neural Networks},'' \href{http://arxiv.org/abs/2409.12222}{{\ttfamily arXiv:2409.12222 [hep-th]}}.

\bibitem{neal}
R.~M. Neal, {\em BAYESIAN LEARNING FOR NEURAL NETWORKS}.
\newblock PhD thesis, University of Toronto, 1995.

\bibitem{williams}
C.~K. Williams, ``Computing with infinite networks,'' in {\em Advances in neural information processing systems}, pp.~295--301.
\newblock 1997.

\bibitem{Matthews2018GaussianPB}
A.~G. d.~G. Matthews, M.~Rowland, J.~Hron, R.~E. Turner, and Z.~Ghahramani, ``Gaussian process behaviour in wide deep neural networks,'' {\em arXiv preprint arXiv:1804.11271} (2018) .

\bibitem{GarrigaAlonso2019DeepCN}
A.~Garriga-Alonso, L.~Aitchison, and C.~E. Rasmussen, ``Deep convolutional networks as shallow gaussian processes,'' {\em arXiv preprint arXiv:1808.05587} (2019) .

\bibitem{yangTPorig}
G.~Yang, ``Scaling limits of wide neural networks with weight sharing: Gaussian process behavior, gradient independence, and neural tangent kernel derivation,'' {\em arXiv preprint arXiv:1902.04760} (2019) .

\bibitem{yangTP1}
G.~Yang, ``Tensor programs i: Wide feedforward or recurrent neural networks of any architecture are gaussian processes,'' {\em arXiv preprint arXiv:1910.12478} (2019) .

\bibitem{yangTP2}
G.~Yang, ``Tensor programs ii: Neural tangent kernel for any architecture,'' {\em arXiv preprint arXiv:2006.14548} (2020) .

\bibitem{Novak2018BayesianCN}
R.~Novak, L.~Xiao, J.~Lee, Y.~Bahri, D.~A. Abolafia, J.~Pennington, and J.~Sohl-Dickstein, ``Bayesian convolutional neural networks with many channels are gaussian processes,'' {\em arXiv preprint arXiv:1810.05148} (2018) .

\bibitem{Yaida2019NonGaussianPA}
S.~Yaida, ``Non-gaussian processes and neural networks at finite widths,'' {\em arXiv preprint arXiv:1910.00019} (2019) .

\bibitem{Halverson_2021}
J.~Halverson, A.~Maiti, and K.~Stoner, ``Neural networks and quantum field theory,'' \href{http://dx.doi.org/10.1088/2632-2153/abeca3}{{\em Machine Learning: Science and Technology} {\bfseries 104} (2021) }.

\bibitem{Roberts_2022}
D.~A. Roberts, S.~Yaida, and B.~Hanin, {\em The Principles of Deep Learning Theory}.
\newblock Cambridge University Press, 2022.

\bibitem{Naveh_2021}
G.~Naveh, O.~B. David, H.~Sompolinsky, and Z.~Ringel, ``Predicting the outputs of finite deep neural networks trained with noisy gradients,'' \href{http://dx.doi.org/10.1103/physreve.104.064301}{{\em Physical Review E} {\bfseries 104} (2021) }.

\bibitem{antognini2019finite}
J.~M. Antognini, ``Finite size corrections for neural network gaussian processes,'' {\em arXiv preprint arXiv:1908.10030} (2019) .

\bibitem{Dyer2020AsymptoticsOW}
E.~Dyer and G.~Gur-Ari, ``Asymptotics of wide networks from feynman diagrams,'' {\em arXiv preprint arXiv:1909.11304} (2020) .

\bibitem{maiti2021symmetryviaduality}
A.~Maiti, K.~Stoner, and J.~Halverson, ``Symmetry-via-duality: Invariant neural network densities from parameter-space correlators,'' {\em arXiv preprint arXiv:2106.00694} (2021) .

\bibitem{erdmenger2021quantifying}
J.~Erdmenger, K.~T. Grosvenor, and R.~Jefferson, ``Towards quantifying information flows: relative entropy in deep neural networks and the renormalization group,'' {\em arXiv preprint arXiv:2107.06898} (2021) .

\bibitem{grosvenor2022edge}
K.~T. Grosvenor and R.~Jefferson, ``The edge of chaos: quantum field theory and deep neural networks,'' {\em arXiv preprint arXiv:2109.13247} (2022) .

\bibitem{Erbin_2022}
H.~Erbin, V.~Lahoche, and D.~O. Samary, ``Non-perturbative renormalization for the neural network-{QFT} correspondence,'' \href{http://dx.doi.org/10.1088/2632-2153/ac4f69}{{\em Machine Learning: Science and Technology} {\bfseries 3} (2022) 015027}.

\bibitem{erbin2022renormalization}
H.~Erbin, V.~Lahoche, and D.~O. Samary, ``Renormalization in the neural network-quantum field theory correspondence,'' {\em arXiv preprint arXiv:2212.11811} (2022) .

\bibitem{maloney2022solvablemodelneuralscaling}
A.~Maloney, D.~A. Roberts, and J.~Sully, ``A solvable model of neural scaling laws,'' 2022.
\newblock \url{https://arxiv.org/abs/2210.16859}.

\bibitem{banta2023structures}
I.~Banta, T.~Cai, N.~Craig, and Z.~Zhang, ``Structures of neural network effective theories,'' {\em arXiv preprint arXiv:2305.02334} (2023) .

\bibitem{10.1088/2632-2153/adc872}
Z.~Zhang, ``Neural scaling laws from large-n field theory: Solvable model beyond the ridgeless limit,'' {\em Machine Learning: Science and Technology} (2025) . \url{http://iopscience.iop.org/article/10.1088/2632-2153/adc872}.

\bibitem{halverson2021building}
J.~Halverson, ``Building quantum field theories out of neurons,'' {\em arXiv preprint arXiv:2112.04527} (2021) .

\bibitem{halverson2024tasilecturesphysicsmachine}
J.~Halverson, ``Tasi lectures on physics for machine learning,'' 2024.
\newblock \url{https://arxiv.org/abs/2408.00082}.

\bibitem{Osterwalder:1973dx}
K.~Osterwalder and R.~Schrader, ``{Axioms for Euclidean Green's functions},'' {\em Communications in Mathematical Physics} {\bfseries 31} no.~2, (1973) 83 -- 112.

\bibitem{Osterwalder:1974tc}
K.~Osterwalder and R.~Schrader, ``{Axioms for Euclidean Green's Functions. 2.},'' \href{http://dx.doi.org/10.1007/BF01608978}{{\em Commun. Math. Phys.} {\bfseries 42} (1975) 281}.

\bibitem{osti_4606723}
A.~S. Wightman and L.~Garding, ``Fields as operator-valued distributions in relativistic quantum theory,'' {\em Arkiv Fys.} {\bfseries Vol: 28} (01, 1965) . \url{https://www.osti.gov/biblio/4606723}.

\bibitem{glimm2012quantum}
J.~Glimm and A.~Jaffe, {\em Quantum Physics: A Functional Integral Point of View}.
\newblock Springer New York, 2012.
\newblock \url{https://books.google.com/books?id=VSjjBwAAQBAJ}.

\bibitem{Kravchuk:2021kwe}
P.~Kravchuk, J.~Qiao, and S.~Rychkov, ``{Distributions in CFT. Part II. Minkowski space},'' \href{http://dx.doi.org/10.1007/JHEP08(2021)094}{{\em JHEP} {\bfseries 08} (2021) 094}, \href{http://arxiv.org/abs/2104.02090}{{\ttfamily arXiv:2104.02090 [hep-th]}}.

\bibitem{Hashimoto:2024aga}
K.~Hashimoto, Y.~Hirono, J.~Maeda, and J.~Totsuka-Yoshinaka, ``{Neural network representation of quantum systems},'' \href{http://dx.doi.org/10.1088/2632-2153/ad81ac}{{\em Mach. Learn. Sci. Tech.} {\bfseries 5} no.~4, (2024) 045039}, \href{http://arxiv.org/abs/2403.11420}{{\ttfamily arXiv:2403.11420 [hep-th]}}.

\bibitem{PhysRev.150.1079}
E.~Nelson, ``Derivation of the schr\"odinger equation from newtonian mechanics,'' \href{http://dx.doi.org/10.1103/PhysRev.150.1079}{{\em Phys. Rev.} {\bfseries 150} (Oct, 1966) 1079--1085}. \url{https://link.aps.org/doi/10.1103/PhysRev.150.1079}.

\bibitem{Fényes1952}
I.~Fényes, ``Eine wahrscheinlichkeitstheoretische begründung und interpretation der quantenmechanik,'' \href{http://dx.doi.org/10.1007/BF01338578}{{\em Zeitschrift für Physik} {\bfseries 132} no.~1, (1952) 81--106}. \url{https://doi.org/10.1007/BF01338578}.

\bibitem{DELAPENAAUERBACH1967603}
L.~{De La Peña-Auerbach}, ``A simple derivation of the schroedinger equation from the theory of markoff processes,'' \href{http://dx.doi.org/https://doi.org/10.1016/0375-9601(67)90639-1}{{\em Physics Letters A} {\bfseries 24} no.~11, (1967) 603--604}. \url{https://www.sciencedirect.com/science/article/pii/0375960167906391}.

\bibitem{Parisi:1980ys}
G.~Parisi and Y.-s. Wu, ``{Perturbation Theory Without Gauge Fixing},'' {\em Sci. Sin.} {\bfseries 24} (1981) 483.

\bibitem{Jaffe2015}
A.~Jaffe, ``Stochastic quantization, reflection positivity, and quantum fields,'' \href{http://dx.doi.org/10.1007/s10955-015-1320-z}{{\em Journal of Statistical Physics} {\bfseries 161} no.~1, (2015) 1--15}. \url{https://doi.org/10.1007/s10955-015-1320-z}.

\bibitem{Kallen:1952zz}
G.~Kallen, ``{On the definition of the Renormalization Constants in Quantum Electrodynamics},'' \href{http://dx.doi.org/10.1007/978-3-319-00627-7_90}{{\em Helv. Phys. Acta} {\bfseries 25} no.~4, (1952) 417}.

\bibitem{Lehmann1954berEV}
H.~Lehmann, ``{\"U}ber eigenschaften von ausbreitungsfunktionen und renormierungskonstanten quantisierter felder,'' {\em Il Nuovo Cimento (1943-1954)} {\bfseries 11} (1954) 342--357. \url{https://api.semanticscholar.org/CorpusID:120848922}.

\bibitem{nmj/1118796540}
T.~Hida and L.~Streit, ``{On quantum theory in terms of white noise},'' {\em Nagoya Mathematical Journal} {\bfseries 68} no.~none, (1977) 21 -- 34.

\bibitem{Altland:2006si}
A.~Altland and B.~Simons, \href{http://dx.doi.org/10.1017/9781108781244}{{\em {Condensed Matter Field Theory}}}.
\newblock Cambridge University Press, 8, 2023.

\bibitem{kosambi1943statistics}
D.~D. Kosambi, ``Statistics in function space,'' {\em Journal of the Indian Mathematical Society} {\bfseries 7} (1943) 76--88.

\bibitem{karhunen1947}
K.~Karhunen, ``Zur spektraltheorie stochastischer prozesse,'' {\em Annales Academiae Scientiarum Fennicae. Series A. I. Mathematica-Physica} {\bfseries 34} (1946) .

\bibitem{Loeve1948}
M.~Lo\`eve, ``Fonctions aleatoires du second ordre,'' in {\em Processus Stochastiques et Mouvement Brownien}, P.~L\'evy, ed.
\newblock 1948.

\bibitem{Zhu_2023}
Y.~Zhu, Y.-H. Tang, and C.~Kim, ``Learning stochastic dynamics with statistics-informed neural network,'' \href{http://dx.doi.org/10.1016/j.jcp.2022.111819}{{\em Journal of Computational Physics} {\bfseries 474} (Feb., 2023) 111819}. \url{http://dx.doi.org/10.1016/j.jcp.2022.111819}.

\bibitem{cmp/1104162595}
L.~Jak{\'o}bczyk and F.~Strocchi, ``{Euclidean formulation of quantum field theory without positivity},'' {\em Communications in Mathematical Physics} {\bfseries 119} no.~4, (1988) 529 -- 541.

\bibitem{NELSON197397}
E.~Nelson, ``Construction of quantum fields from markoff fields,'' \href{http://dx.doi.org/https://doi.org/10.1016/0022-1236(73)90091-8}{{\em Journal of Functional Analysis} {\bfseries 12} no.~1, (1973) 97--112}.

\bibitem{Velo:1973tpu}
G.~Velo and A.~S. Wightman, eds., \href{http://dx.doi.org/10.1007/BFb0113079}{{\em {Constructive quantum field theory. 1973 school of mathematical physics, erice}}}.
\newblock 1973.

\bibitem{KLEIN1978277}
A.~Klein, ``The semigroup characterization of osterwalder-schrader path spaces and the construction of euclidean fields,'' \href{http://dx.doi.org/https://doi.org/10.1016/0022-1236(78)90009-5}{{\em Journal of Functional Analysis} {\bfseries 27} no.~3, (1978) 277--291}. \url{https://www.sciencedirect.com/science/article/pii/0022123678900095}.

\bibitem{generalization_of_markov}
A.~Klein, ``{A Generalization of Markov Processes},'' \href{http://dx.doi.org/10.1214/aop/1176995616}{{\em The Annals of Probability} {\bfseries 6} no.~1, (1978) 128 -- 132}. \url{https://doi.org/10.1214/aop/1176995616}.

\bibitem{ARVESON1986173}
W.~Arveson, ``Markov operators and os-positive processes,'' \href{http://dx.doi.org/https://doi.org/10.1016/0022-1236(86)90071-6}{{\em Journal of Functional Analysis} {\bfseries 66} no.~2, (1986) 173--234}. \url{https://www.sciencedirect.com/science/article/pii/0022123686900716}.

\bibitem{function_of_brownian}
J.~B. Walsh, ``Functions of brownian motion,'' {\em Proceedings of the American Mathematical Society} {\bfseries 49} no.~1, (1975) 227--231. \url{http://www.jstor.org/stable/2039821}.

\bibitem{PhysRev.36.823}
G.~E. Uhlenbeck and L.~S. Ornstein, ``On the theory of the brownian motion,'' \href{http://dx.doi.org/10.1103/PhysRev.36.823}{{\em Phys. Rev.} {\bfseries 36} (Sep, 1930) 823--841}. \url{https://link.aps.org/doi/10.1103/PhysRev.36.823}.

\bibitem{DoobJ.L.1942TBMa}
J.~L. Doob, ``The brownian movement and stochastic equations,'' {\em Annals of mathematics} {\bfseries 43} no.~2, (1942) 351--369.

\bibitem{Pavliotis2014}
G.~A. Pavliotis, \href{http://dx.doi.org/10.1007/978-1-4939-1323-7}{{\em Stochastic Processes and Applications: Diffusion Processes, the Fokker–Planck and Langevin Equations}}, vol.~60 of {\em Texts in Applied Mathematics}.
\newblock Springer, 2014.

\bibitem{Gross:2019ach}
D.~J. Gross, J.~Kruthoff, A.~Rolph, and E.~Shaghoulian, ``{$T\overline{T}$ in AdS$_2$ and Quantum Mechanics},'' \href{http://dx.doi.org/10.1103/PhysRevD.101.026011}{{\em Phys. Rev. D} {\bfseries 101} no.~2, (2020) 026011}, \href{http://arxiv.org/abs/1907.04873}{{\ttfamily arXiv:1907.04873 [hep-th]}}.

\bibitem{Gross:2019uxi}
D.~J. Gross, J.~Kruthoff, A.~Rolph, and E.~Shaghoulian, ``{Hamiltonian deformations in quantum mechanics, $T\bar T$, and the SYK model},'' \href{http://dx.doi.org/10.1103/PhysRevD.102.046019}{{\em Phys. Rev. D} {\bfseries 102} no.~4, (2020) 046019}, \href{http://arxiv.org/abs/1912.06132}{{\ttfamily arXiv:1912.06132 [hep-th]}}.

\bibitem{Ferko:2023ozb}
C.~Ferko and A.~Gupta, ``{ModMax oscillators and root-TT\textasciimacron{}-like flows in supersymmetric quantum mechanics},'' \href{http://dx.doi.org/10.1103/PhysRevD.108.046013}{{\em Phys. Rev. D} {\bfseries 108} no.~4, (2023) 046013}, \href{http://arxiv.org/abs/2306.14575}{{\ttfamily arXiv:2306.14575 [hep-th]}}.

\bibitem{Ferko:2023iha}
C.~Ferko, A.~Gupta, and E.~Iyer, ``{Quantization of the ModMax oscillator},'' \href{http://dx.doi.org/10.1103/PhysRevD.108.126021}{{\em Phys. Rev. D} {\bfseries 108} no.~12, (2023) 126021}, \href{http://arxiv.org/abs/2310.06015}{{\ttfamily arXiv:2310.06015 [hep-th]}}.

\bibitem{Skinner2018}
D.~Skinner, ``Quantum field theory ii,'' 2018.
\newblock \url{https://www.damtp.cam.ac.uk/user/dbs26/AQFT.html}. Lecture notes for Advanced Quantum Field Theory, University of Cambridge.

\bibitem{FeynmanHibbsStyer2010}
R.~P. Feynman, A.~R. Hibbs, and D.~F. Styer, {\em Quantum Mechanics and Path Integrals: Emended Edition}.
\newblock Dover Publications, Mineola, NY, emended~ed., 2010.

\bibitem{Ong2012}
Y.~C. Ong, ``Note: Where is the commutation relation hiding in the path integral formulation?'' 2012.
\newblock \url{https://ncatlab.org/nlab/files/Ong-CommutationRelationInPathIntegral.pdf}.

\bibitem{Koch:2014cma}
B.~Koch and I.~Reyes, ``{Differentiable-Path Integrals in Quantum Mechanics},'' \href{http://dx.doi.org/10.1142/S0219887815501005}{{\em Int. J. Geom. Meth. Mod. Phys.} {\bfseries 12} no.~09, (2015) 1550100}, \href{http://arxiv.org/abs/1404.6551}{{\ttfamily arXiv:1404.6551 [quant-ph]}}.

\end{thebibliography}\endgroup

\end{document}